%% file: nonquotienting-maps.tex
\documentclass[11pt, a4paper, twoside, leqno]{amsart}
\usepackage[centering, totalwidth = 380pt, totalheight = 600pt]{geometry}
\usepackage{color}
\definecolor{darkgreen}{rgb}{0,0.45,0}
\usepackage[colorlinks,citecolor=darkgreen,linkcolor=darkgreen]{hyperref}
\usepackage[utf8]{inputenc}
\usepackage{enumerate, mathpartir,
  amsmath, amsthm, amssymb, stmaryrd, microtype, graphicx, empheq,
  wrapfig, picinpar, xspace, xparse, leftidx} 
\usepackage{mathrsfs}
\usepackage[arrow, matrix, tips, curve, graph, rotate]{xy}
\SelectTips{cm}{10}
\input{./preamble-common}

\newtheorem{thm}{Theorem}[section]
\newtheorem{Thm}[thm]{Theorem}
 
\newtheorem{Lemma}[thm]{Lemma} 
 
\newtheorem{Prop}[thm]{Proposition} 
 
\newtheorem{Cor}[thm]{Corollary} 

\theoremstyle{definition}
 
\newtheorem{Defn}[thm]{Definition} 
 
\newtheorem{Exs}[thm]{Examples} 
\newtheorem{Ex}[thm]{Example}
 
\newtheorem{Rk}[thm]{Remark} 
 
{}
 
\numberwithin{equation}{section}

\renewcommand{\Aut}{\mathrm{Aut}}

\renewcommand{\cat}[1]{\mathbf{#1}}
\newcommand{\psh}[1][\C]{{\mathscr P {#1}}}

\newcommand{\fampt}{\cat{FAM}_\mathrm{pt}}
\newcommand{\nfampt}{\cat{NFAM}_\mathrm{pt}}
\newcommand{\anpt}{\cat{AN}_\mathrm{pt}}
\tikzset{
medge/.style={draw,shape=isosceles triangle,isosceles triangle apex angle=#1,shape border rotate=90},
medge/.default={90},
mtype/.style={rectangle,draw=none}
}

\newcommand{\trees}{\T}

\newcommand{\PS}{{\PP\!{}_s}}
\newcommand{\uan}{\cat{AN}_U}
\newcommand{\umd}{\cat{MND}_U}
\newcommand{\shp}{\cat{SH}_U}
\newcommand{\shpmon}{\cat{SHM}_U}

\renewcommand{\A}{{\mathscr A}}
\renewcommand{\B}{{\mathscr B}}
\renewcommand{\C}{{\mathscr C}}
\renewcommand{\D}{{\mathscr D}}
\renewcommand{\E}{{\mathscr E}}
\renewcommand{\F}{{\mathscr F}}
\renewcommand{\G}{{\mathscr G}}

\renewcommand{\I}{{\mathscr I}}
\renewcommand{\J}{{\mathscr J}}
\renewcommand{\K}{{\mathscr K}}
\renewcommand{\L}{{\mathscr L}}
\renewcommand{\M}{{\mathscr M}}

\renewcommand{\O}{{\mathscr O}}
\renewcommand{\P}{{\mathscr P}}

\let\sec=\S
\renewcommand{\S}{{\mathscr S}}
\renewcommand{\T}{{\mathscr T}}

\newcommand{\aut}{\mathfrak S}
\renewcommand{\id}{\mathrm{id}}

\DeclareMathAlphabet{\mathbf}{OT1}{cmr}{b}{n}

\renewcommand{\vec}{\boldsymbol}

\tikzset{circuit/.style={draw,minimum width=2cm,minimum height=1.5em}}

\begin{document}
 \leftmargini=2em

\author[R.\ Garner]{Richard Garner}
\address{Department of Mathematics, Macquarie University, NSW 2109, Australia} 
\email{richard.garner@mq.edu.au}

\author[T.\ Hirschowitz]{Tom Hirschowitz}
\address{Univ. Savoie Mont Blanc, CNRS, LAMA, F-73000 Chambéry, France}
\email{tom.hirschowitz@univ-smb.fr}

\thanks{Hirschowitz acknowledges the support of the French ANR projet
  blanc R\'ecr\'e ANR-11-BS02-0010; Garner acknowledges the support of
  the Australian Research Council Discovery Projects DP110102360 and
  DP130101969. This research was partially undertaken during a visit
  by Hirschowitz to Sydney funded by Macquarie University.}%

\title{Shapely monads and analytic functors}

\begin{abstract}
  In this paper, we give precise mathematical form to the idea of a
  structure whose data and axioms are faithfully represented by a
  graphical calculus; some prominent examples are operads,
  polycategories, properads, and PROPs. Building on the established
  presentation of such structures as algebras for monads on presheaf
  categories, we describe a characteristic property of the associated
  monads---the \emph{shapeliness} of the title---which says that ``any
  two operations of the same shape agree''.

  An important part of this work is the study of \emph{analytic}
  functors between presheaf categories, which are a common
  generalisation of Joyal's analytic endofunctors on sets and of the
  parametric right adjoint functors on presheaf categories introduced
  by Diers and studied by Carboni--Johnstone, Leinster and Weber. Our
  shapely monads will be found among the analytic endofunctors, and
  may be characterised as the submonads of a \emph{universal} analytic
  monad with ``exactly one operation of each shape''.

  In fact, shapeliness also gives a way to \emph{define} the data and
  axioms of a structure directly from its graphical calculus, by
  generating a free shapely monad on the basic operations of the
  calculus. In this paper we do this for some of the examples listed
  above; in future work, we intend to use this to obtain canonical
  notions of denotational model for graphical calculi such as Milner's
  bigraphs, Lafont's interaction nets, or Girard's multiplicative
  proof nets. \looseness=-1
\end{abstract}
\maketitle

\section{Introduction}
\label{sec:introduction}

In mathematics and computer science, we often encounter structures
which are faithfully encoded by a graphical calculus of the following
sort. The basic data of the structure are depicted as certain
diagrams; the basic operations of the structure act by glueing
together these diagrams along certain parts of their boundaries; and
the axioms of the structure are just those necessary to ensure that
``every two ways of glueing a compound diagram together agree''.

Commonly, such calculi depict structures wherein ``functions'',
``arrows'' or ``processes'' are wired together along input or output
``ports''. For instance, we have
\emph{multicategories}~\cite{Lambek1969Deductive}, whose arrows have
many inputs but only one output;
\emph{polycategories}~\cite{Szabo1977The-logic}, whose arrows have
multiple inputs and outputs, with composition subject to a linear
wiring discipline; and \emph{coloured
  properads}~\cite{Vallette2007A-Koszul} and
\textsc{prop}s~\cite{Mac-Lane1965Categorical}, which are like
polycategories but allow for non-linear wirings.

Mathematical structures such as these are important in algebraic
topology and homological algebra---encoding, for example, operations
arising on infinite loop spaces~\cite{May1972The-geometry} or on
Hochschild cochains~\cite{McClure2002A-solution}---but also in logic
and computer science. For example, polycategories encode the
underlying semantics of a linear sequent
calculus~\cite{Lambek1993From}, while \textsc{prop}s have recently
been used as an algebraic foundation for notions of computational
network such as signal flow graphs~\cite{Bonchi2015} and Bayesian
networks~\cite{Causaltheories}. Other kinds of graphical structures
arising in computer science include \emph{proof nets}~\cite[\sec
2]{Girard1987Linear}, \emph{interaction nets}~\cite{Lafont90} and
\emph{bigraphs}~\cite{Milner:bigraphs}.

There is an established approach to describing structures of the above
kind using monads on presheaf categories. The presheaf category
captures the essential topology of the underlying graphical calculus,
while the monad encodes both the wiring operations of the structure
and the axioms that they obey; the algebras for the monad are the
instances of the structure. One aspect which this approach does not
account for is that the axioms should be determined by the requirement
that ``every two ways of wiring a compound diagram together agree''.
The first main contribution of this paper is to rectify this: we
explain the observed form of the axioms as a property of the
associated monad---which we term \emph{shapeliness}---stating that
``every two operations of the same shape coincide''.

In fact, shapeliness gives not just a way of \emph{characterising} the
monads encoding graphical structures, but also a systematic way of
\emph{generating} them. For a given graphical calculus, it is
typically easy to find a presheaf category encoding the basic
diagram-shapes of the calculus, and an endofunctor thereon encoding
the basic wiring operations; we then obtain the desired monad as the
\emph{free shapely monad} on the given endofunctor. The algebras for
this monad can be seen as denotational models of the graphical
calculus in question; and though we do not do this here, one can
envisage this being used to attach workable denotational semantics to,
for example, interaction nets; the syntactic part of bigraphs; or
\textsc{mll} proof nets without units.
 
Formalising the notion of shapely monad turns out to be a delicate
task. In the end, we will define a monad on
$\psh = [\C^\mathrm{op}, \cat{Set}]$ to be \emph{shapely} just when it
is a submonad of a \emph{universal} shapely monad $\mathsf U$ with
``exactly one operation of each shape''. We will find $\mathsf U$ by
seeking a terminal object in a suitable monoidal category of
\emph{endofunctors} of $\psh$; once found, terminality will
automatically endow this object with a monad structure, so giving the
desired~$\mathsf U$.

This leaves the problem of choosing a suitable monoidal category of
endofunctors. An obvious but wrong choice would be the whole functor
category $[\psh, \psh]$: with this choice, $\mathsf U$ would be the
monad constant at $1$, and a general monad would be shapely just when
it took values in subobjects of $1$. This is manifestly not what we
want; the problem is that terminality in the full functor category
encodes the property of having ``exactly one operation of each shape''
for what are overly crude notions of ``operation'' and ``operation
shape''.

Refining these notions, as we shall do, means looking for a terminal
object in some smaller category of endofunctors of $\psh$. Choosing
this category turns out to be an interesting design problem: some
natural candidates have a terminal object, but are not closed under
composition, while others are closed under composition, but fail to
have a terminal object. Our eventual solution will triangulate between
these failures, but we make no claims to its definitiveness; in fact,
we consider the value of our work to lie as much in the exploration of
the problem's design space as in the particular solution we adopt.

The technical foundation of our approach will be a theory of
\emph{analytic functors} between presheaf categories, and the second
main contribution of this paper is to develop such a theory. Analytic
endofunctors of $\cat{Set}$ were introduced by Joyal
in~\cite{Joyal1986Foncteurs} as a categorical foundation for
enumerative combinatorics; their importance for computer science has
been recognised in work such
as~\cite{Abbott2004Constructing,Fiore2008The-cartesian,Hasegawa2002Two-applications}.
An endofunctor $F$ of $\cat{Set}$ is \emph{analytic} in Joyal's sense
when it can be written in the form:
\begin{equation*}
  FX = \textstyle \sum_{i \in I} {X^{\alpha_i}}_{/ G_i}
\end{equation*}
for an $I$-indexed family of natural numbers $\alpha_i$ and subgroups
$G_i \leqslant \aut_{\alpha_i}$; the quotients are by the permutation
actions of $G_i$ on the factors of $X^{\alpha_i}$.

Various authors have considered various ways of extending Joyal's
notion to arbitrary presheaf categories, as, for example,
in~\cite{Abbott2004Constructing,Fiore2008The-cartesian}. However, it
turns out that these prior notions are not appropriate to our needs,
since the notions of analyticity they describe are not general enough
to capture the functors and monads of interest (see Remark~\ref{rk:6}
below). Instead, guided by our applications, we choose to call a
functor $F \colon \psh[\D] \rightarrow \psh$ \emph{analytic} if it
takes the form:
\begin{equation}\label{eq:40}
  FX(c) = \textstyle \sum_{i \in I_c} \psh[\D](\alpha_i, X)_{/ G_i}
\end{equation}
for some family of presheaves $\alpha_i \in \psh[\D]$ and subgroups
$G_i \leqslant \mathrm{Aut}(\alpha_i)$ of the automorphism group of
each $\alpha_i$. Just as in the case studied
in~\cite{Joyal1986Foncteurs}, these generalised analytic functors have
a strongly combinatorial flavour; moreover, all of the monads derived
from graphical calculi that we will consider turn out to be analytic
in this
sense.

It is therefore reasonable that we should look for a universal shapely
monad among analytic endofunctors of a presheaf category. While we do
not succeed in doing this in full generality, our attempts to do so
lead us to develop various results of independent interest, including
the following:
\begin{enumerate}[(i)]
\item We give a combinatorial representation of the category of
  finitary analytic functors and transformations
  $\psh \rightarrow \psh[\D]$ (Proposition~\ref{prop:4}) and show that
  this category always has a terminal object
  (Proposition~\ref{prop:9}).\vskip0.25\baselineskip
\item We show that analytic functors between presheaf categories need
  \emph{not} be closed under composition (Proposition~\ref{prop:6});
  this is by contrast to analytic endofunctors of $\cat{Set}$, which
  \emph{are} composition-closed.\vskip0.25\baselineskip
\item We introduce a condition on analytic functors which we call
  \emph{cellularity}, that is sufficient to ensure that they do
  compose (Proposition~\ref{prop:13}). \vskip0.25\baselineskip
\item We see that, unfortunately, the introduction of cellularity also
  destroys the terminal object among finitary analytic endofunctors
  (Proposition~\ref{prop:no-terminal-object}).
\end{enumerate}

This last, negative, result delineates the boundary of our knowledge
in the general case; but a small adaptation allows us to obtain
positive results in the specific cases necessary to talk about
\textsc{prop}s, properads and polycategories. As we recall, each of
these notions may be encoded by a monad on a suitable presheaf
category of \emph{polygraphs} (Definitions~\ref{def:36}
and~\ref{def:35}). What we further show is that:
\begin{itemize}
\item For endofunctors of the category of polygraphs, the notion of
  cellularity can be augmented with simple \emph{ad hoc} conditions to
  obtain a class of analytic functors which is composition-closed and
  admits a terminal object (Propositions~\ref{prop:23}
  and~\ref{prop:31}).\vskip0.25\baselineskip
\item In this context, therefore, there exists a universal shapely
  monad, and moreover, we may describe concretely the free shapely
  monad on any shapely endofunctor
  (Proposition~\ref{prop:28}).\vskip0.25\baselineskip
\item The monads for polycategories, properads and \textsc{prop}s are
  the free shapely monads on the endofunctors encoding their basic
  wiring operations (Theorems~\ref{thm:polycats} and~\ref{thm:2}).
\end{itemize}

We conclude this introduction with a brief overview of the rest of the
paper. We start in Section~\ref{sec:motivating-examples} by developing
some motivating examples of graphical calculi and the algebraic
structures they describe. These graphical calculi include that for
symmetric monoidal categories introduced
in~\cite{Joyal1991The-geometry} (but see
also~\cite{Selinger2011A-survey}); the associated algebraic structures
include the \emph{polycategories} of~\cite{Szabo1975Polycategories},
the (coloured) \emph{properads} of~\cite{Vallette2007A-Koszul}, and
the \textsc{prop}s of~\cite{Mac-Lane1965Categorical}. We then explain
how these algebraic structures can be described as the algebras for a
monad on a presheaf category.

In Section~\ref{sec:famil-funct-shap}, we begin our pursuit of the
notion of universal shapely monad. We do not immediately consider the
analytic functors discussed above, but rather the narrower class of
\emph{familially representable} or \emph{familial} functors,
introduced by Diers~\cite{Diers1978Spectres} and studied by Johnstone,
Leinster and Weber~\cite{Carboni1995Connected,
  Weber2004Generic,Leinster2004Higher}; these are precisely the
analytic functors whose expression~\eqref{eq:40} involves only
\emph{trivial} groups $G_i$. We recall basic aspects of the theory of
familial functors, including closure under composition, but show that
there is typically no terminal object among familial endofunctors, and
hence no universal shapely monad among them.

In Section~\ref{sec:analyt-funct-shap}, we attempt to fix up the lack
of a terminal object among familial endofunctors by passing to the
more general analytic functors. As is visible from (i) and (ii) above,
we succeed in doing this, but only at the cost of breaking the
composability of familial endofunctors.
Section~\ref{sec:cell-analyt-funct} attempts to fix this new problem
by introducing the more restricted class of \emph{cellular} analytic
functors; as in (iii) and (iv) above, this does indeed resolve the problem of
composability but at the same time reintroduces the problem of the
existence of a terminal object.

At this point, in Section~\ref{sec:free-shapely-monads}, we declare
ourselves unable to find a further refinement of the notion of
cellularity that, in full generality, fixes both composability and
existence of a terminal object. However, in the presheaf categories
relevant to the motivating examples of
Section~\ref{sec:motivating-examples}, we are able to impose an
additional \emph{ad hoc} condition on top of cellularity which is
sufficient to ensure that the cellular functors in this class both
compose \emph{and} admit a terminal object: using this, we finally
obtain the desired notion of shapely monad, and are able to exhibit
the monads encoding the graphical structures of interest as free
shapely monads on the basic wiring operations of the structure.

\section{Motivating examples}
\label{sec:motivating-examples}

\subsection{Some examples of graphical calculi}
\label{sec:some-exampl-graph}

Before developing our general theory of shapeliness, we describe some
of the examples of monads derived from graphical calculi that our
theory is intended to capture. The graphical calculi which we consider
will involve diagrams built out of labelled boxes
\begin{equation}\label{eq:21}
  \begin{tikzpicture}[y=0.80pt, x=0.80pt, yscale=-1.000000, xscale=1.000000, inner sep=0pt, outer sep=0pt]
    \path[draw=black,line join=miter,line cap=butt,even odd rule,line
    width=0.500pt,rounded corners=0.0000cm] (60.0000,902.3622) rectangle
    (105.0000,932.3622);
    \path[draw=black,line join=miter,line cap=butt,even odd rule,line width=0.500pt]
    (60.0000,872.3622) .. controls (60.0000,887.2619) and (75.2525,887.9675) ..
    (75.0000,902.3622) node[above=0.1cm,at start] {$A_1$};
    \path[draw=black,line join=miter,line cap=butt,even odd rule,line width=0.500pt]
    (90.0000,902.3622) .. controls (90.0000,887.7150) and (105.0000,887.3622) ..
    (105.0000,872.3622) node[above=0.1cm,at end] {\,$A_n$};
    \path[draw=black,line join=miter,line cap=butt,even odd rule,line width=0.500pt]
    (75.0000,932.3622) .. controls (75.5051,946.7568) and (59.7475,947.4624) ..
    (60.0000,962.3622) node[below=0.1cm,at end] {$B_1$};
    \path[draw=black,line join=miter,line cap=butt,even odd rule,line width=0.500pt]
    (90.0000,932.3622) .. controls (90.2525,947.7670) and (104.7475,947.2099) ..
    (105.0000,962.3622) node[below=0.1cm,at end] {\,$B_m$};
    \path (81.9188,916.8571) node {$f$};
    \path (81.9188,862) node {\,$\cdots$};
    \path (81.9188,972) node {\,$\cdots$};
  \end{tikzpicture}
\end{equation}
with a finite number of ``input'' wires (positioned above the box) and
``output'' wires (positioned below). There are various interpretations
we could give to such a box, for example:
\begin{enumerate}[(i)]
\item As a derivation in a linear sequent calculus of
  $A_1, \dots, A_n \mathbin{\vdash} B_1, \dots, B_m$;
\item As a linear map
  $A_1 \otimes \dots \otimes A_n \rightarrow B_1 \otimes \dots \otimes
  B_m$ between $k$-vector spaces;
\item As a program in the typed $\lambda$-calculus of type
  $A_1 \times \dots \times A_n \rightarrow B_1 \times \dots \times
  B_m$.
\end{enumerate}
Each of these interpretations will be associated to a different
graphical calculus; the difference between them is in the rules
governing how boxes can be plugged together to form larger diagrams.
For example:
\begin{enumerate}[(i)]
\item Given proofs $f$ of $C, D \vdash E, F, Z, I, J$ and $g$ of
  $A, B, Z \vdash G, H$ in the linear sequent calculus, we can cut
  along the proposition $Z$ to obtain a proof of
  $A,B,C,D \vdash E,F,G,H,I,J$. Thus, in the corresponding graphical
  calculus, we can plug together the boxes representing $f$ and $g$ to
  obtain a diagram:
  
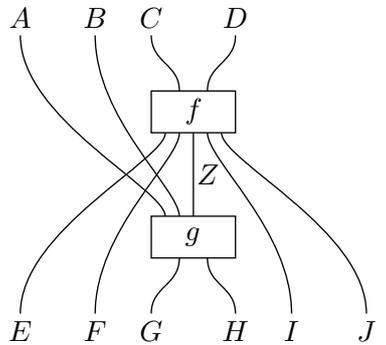
\begin{equation*}
    \begin{tikzpicture}[y=0.70pt, x=0.70, yscale=-1.000000, xscale=1.000000, inner sep=0pt, outer sep=0pt]
      \path[draw=black,line join=miter,line cap=butt,even odd rule,line
      width=0.500pt,rounded corners=0.0000cm] (60.0000,902.3622) rectangle
      (105.0000,924.8622);
      \path[draw=black,line join=miter,line cap=butt,even odd rule,line width=0.500pt]
      (60.0000,872.3622) .. controls (60.0000,887.2619) and (75.2525,887.9675) ..
      (75.0000,902.3622) node[above=0.1cm,at start] {$C$};
      \path[draw=black,line join=miter,line cap=butt,even odd rule,line width=0.500pt]
      (90.0000,902.3622) .. controls (90.0000,887.7150) and (105.0000,887.3622) ..
      (105.0000,872.3622) node[above=0.1cm,at end] {$D$};
      \path[draw=black,line join=miter,line cap=butt,even odd rule,line width=0.500pt]
      (67.5000,924.8622) .. controls (68.0051,939.2568) and (-10.0000,977.3621) ..
      (-10.0000,1022.3622) node[below=0.1cm,at end] {$E$}; 
      \path[draw=black,line join=miter,line cap=butt,even odd rule,line width=0.500pt]
      (97.5000,924.8622) .. controls (97.7525,940.2670) and (175.0000,977.2098) ..
      (175.0000,1022.3622) node[below=0.1cm,at end] {$J$\rlap{ .}};
      \path[fill=black,line join=miter,line cap=butt,line width=0.500pt]
      (82.6331,913.2857) node {$f$};
      \path[draw=black,line join=miter,line cap=butt,even odd rule,line width=0.500pt]
      (75.0000,924.8622) .. controls (75.5051,939.2568) and (30.0000,977.3621) ..
      (30.0000,1022.3622) node[below=0.1cm,at end] {$F$};
      \path[draw=black,line join=miter,line cap=butt,even odd rule,line width=0.500pt]
      (90.0000,924.8622) .. controls (90.2525,940.2670) and (135.0000,977.3621) ..
      (135.0000,1022.3622) node[below=0.1cm,at end] {$I$};
      \path[draw=black,line join=miter,line cap=butt,even odd rule,line width=0.500pt]
      (82.5000,924.8622) -- (82.5000,969.8622) node[midway, right=0.05cm] {$Z$};
      \path[draw=black,line join=miter,line cap=butt,even odd rule,line
      width=0.500pt,rounded corners=0.0000cm] (60.0000,969.8622) rectangle
      (105.0000,992.3622);
      \path[draw=black,line join=miter,line cap=butt,even odd rule,line width=0.500pt]
      (-10.0000,872.3622) .. controls (-10.0000,917.3622) and (67.2475,954.9624) ..
      (67.5000,969.8622) node[above=0.1cm,at start] {$A$};
      \path[draw=black,line join=miter,line cap=butt,even odd rule,line width=0.500pt]
      (30.0000,872.3622) .. controls (30.0000,917.3622) and (74.7475,954.9624) ..
      (75.0000,969.8622) node[above=0.1cm,at start]
      {$B$};
      \path[draw=black,line join=miter,line cap=butt,even odd rule,line width=0.500pt]
      (75.0000,992.3622) .. controls (75.0000,1007.3622) and (60.0000,1007.3622) ..
      (60.0000,1022.3622) node[below=0.1cm,at end] {$G$};
      \path[draw=black,line join=miter,line cap=butt,even odd rule,line width=0.500pt]
      (90.0000,992.3622) .. controls (90.0000,1007.3622) and (105.0000,1007.3622) ..
      (105.0000,1022.3622) node[below=0.1cm,at end] {$H$};
      \path[fill=black,line join=miter,line cap=butt,line width=0.500pt]
      (82.2759,981.1428) node{$g$};
    \end{tikzpicture}
  \end{equation*}
\item Given $k$-linear maps $f \colon A \otimes B \rightarrow C$,
  $g \colon E \rightarrow F \otimes G$ and
  $h \colon C \otimes F \otimes G \rightarrow K$, we can consider the
  $k$-linear map $A \otimes E \otimes B \rightarrow K$ which sends
  $a \otimes e \otimes b$ to $h(f(a \otimes b) \otimes g(e))$. Thus,
  in the corresponding graphical calculus, we can plug together the
  boxes representing $f$, $g$ and $h$ to obtain a diagram:
  
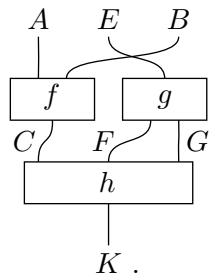
\begin{equation*}
    \begin{tikzpicture}[y=0.70pt, x=0.70, yscale=-1.000000, xscale=1.000000, inner sep=0pt, outer sep=0pt]
      \path[draw=black,line join=miter,line cap=butt,even odd rule,line width=0.500pt,rounded corners=0.0000cm]
      (60.0000,909.8622) rectangle (105.0000,932.3622);
      \path[draw=black,line join=miter,line cap=butt,even odd rule,line width=0.500pt]
      (75.0000,887.3622) .. controls (75.0000,902.2619) and
      (75.2525,895.4675) .. (75.0000,909.8622) node[above=0.1cm,at start] {$A$};
      \path[draw=black,line join=miter,line cap=butt,even odd rule,line width=0.500pt]
      (90.0000,909.8622) .. controls (90.0000,895.2150) and
      (150.0000,902.3622) .. (150.0000,887.3622) node[above=0.1cm,at end] {$B$};
      \path[fill=black,line join=miter,line cap=butt,line width=0.500pt]
      (77.3731,926.9180) node[above right] (text4430) {$f$};
      \path[draw=black,line join=miter,line cap=butt,even odd rule,line width=0.500pt,rounded corners=0.0000cm]
      (120.0000,909.8622) rectangle (165.0000,932.3622);
      \path[draw=black,line join=miter,line cap=butt,even odd rule,line width=0.500pt]
      (135.0000,932.3622) .. controls (135.5051,946.7569) and
      (112.2475,939.9624) .. (112.5000,954.8622) node[left=0.2cm,midway] {$F$};
      \path[draw=black,line join=miter,line cap=butt,even odd rule,line width=0.500pt]
      (150.0000,932.3622) .. controls (150.2525,947.7670) and
      (149.7475,939.7099) .. (150.0000,954.8622) node[right=0.1cm,midway] {$G$};
      \path[fill=black,line join=miter,line cap=butt,line width=0.500pt]
      (138.6155,925.9180) node[above right] (text4430-5) {$g$};
      \path[draw=black,line join=miter,line cap=butt,even odd rule,line width=0.500pt]
      (142.5000,909.8622) .. controls (142.5000,895.2150) and
      (112.5000,902.3622) .. (112.5000,887.3622) node[above=0.1cm,at end] {$E$};
      \path[draw=black,line join=miter,line cap=butt,even odd rule,line width=0.500pt]
      (82.5000,932.3622) .. controls (83.0051,946.7569) and
      (74.7475,939.9624) .. (75.0000,954.8622) node[left=0.13cm,midway] {$C$};
      \path[draw=black,line join=miter,line cap=butt,even odd rule,line width=0.500pt,rounded corners=0.0000cm]
      (67.5000,954.8622) rectangle (157.5000,977.3622);
      \path[draw=black,line join=miter,line cap=butt,even odd rule,line width=0.500pt]
      (112.5000,977.3622) -- (112.5000,999.8622) node[below=0.1cm,at
      end] {$K$\rlap{ .}};
      \path[fill=black,line join=miter,line cap=butt,line width=0.500pt]
      (107.3731,971.3470) node[above right] (text4430-9) {$h$};
    \end{tikzpicture}
  \end{equation*}
\item Given programs $f \colon A \rightarrow B$ and
  $g \colon B \times A \rightarrow C$, there is a composite program
  $\lambda a.\, g(f(a),a) \colon A \rightarrow C$; thus, in the
  corresponding graphical calculus, we can plug together the boxes for
  $f$ and $g$ to obtain a diagram:
  
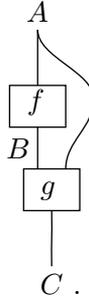
\begin{equation*}
    \begin{tikzpicture}[y=0.70pt, x=0.70, yscale=-1.000000, xscale=1.000000, inner sep=0pt, outer sep=0pt]
      \path[draw=black,line join=miter,line cap=butt,even odd rule,line width=0.500pt]
      (67.5000,902.3622) rectangle (97.5000,924.8622);
      \path[draw=black,line join=miter,line cap=butt,even odd rule,line width=0.500pt]
      (90.0000,902.3622) .. controls (90.0000,887.3622) and
      (105.0000,879.8622) .. (105.0000,864.8622) .. controls
      (105.0000,849.8622) and (75.0000,842.3622) .. (75.0000,827.3622);
      \path[fill=black,line join=miter,line cap=butt,line width=0.800pt]
      (76.6155,918.8774) node[above right] (text4430) {$g$};
      \path[draw=black,line join=miter,line cap=butt,even odd rule,line width=0.500pt]
      (60.0000,857.3622) rectangle (90.0000,879.8622);
      \path[draw=black,line join=miter,line cap=butt,even odd rule,line width=0.500pt]
      (82.5000,924.8622) .. controls (82.7525,940.2670) and
      (82.2475,939.7099) .. (82.5000,954.8622) node[below=0.1cm,at end] {$C$\rlap{ .}};
      \path[fill=black,line join=miter,line cap=butt,line width=0.500pt]
      (69.3680,873.4485) node[above right] (text4430-5) {$f$};
      \path[draw=black,line join=miter,line cap=butt,even odd rule,line width=0.500pt]
      (75.0000,827.3622) -- (75.0000,857.3622) node[above=0.1cm,at start] {$A$};
      \path[draw=black,line join=miter,line cap=butt,even odd rule,line width=0.500pt]
      (75.0000,879.8622) -- (75.0000,902.3622) node[left=0.1cm,midway] {$B$};
    \end{tikzpicture}
  \end{equation*}
\end{enumerate}

With a little further thought, we can derive from the intended
interpretations of the boxes a description of the associated wiring
discipline:
\begin{enumerate}[(i)]
\item In the linear sequent calculus, we can only cut
  along a single formula, so that in the corresponding graphical
  calculus, we can only plug two boxes together along a single wire
  (output to input);\vskip0.25\baselineskip
\item In the case of linear maps between vector spaces,
  we can compose maps together over \emph{multiple} tensor components,
  so that we can now plug multiple outputs of one box into multiple
  inputs of a second. We can also form the tensor product of two maps,
  corresponding to composing two boxes by placing them alongside each
  other.\vskip0.25\baselineskip
\item In the case of programs, we have the possibility
  of duplicating or discarding values; thus the corresponding
  graphical calculus will augment the rules from (ii) by
  allowing wires to split and terminate as they go down the page.
\end{enumerate}
There are other possibilities; for example, intermediate between (i)
and (ii) we have (ii)' which allows for plugging multiple inputs as
in (ii) but does not allow for placing boxes alongside each
other.

\subsection{Algebraic structures from graphical calculi}
\label{sec:algebr-struct-from}

In general, the purpose of graphical calculi is to provide a
denotation system for elements in a semantic structure. For example,
the graphical calculus in (ii) can be used to describe compound
morphisms in the category of $k$-vector spaces, but more generally, in
any symmetric monoidal category~\cite{Kelly1971Coherence}; it is
essentially the calculus of string diagrams
in~\cite{Joyal1991The-geometry}. However, the calculus in~(iii), with
its more permissive wiring discipline, cannot be interpreted into
$k$-vector spaces as there is no $k$-linear correlate to the operation
of splitting or terminating wires.

There is a particularly canonical class of semantic structures into
which a given graphical calculus can be interpreted; the structures in
this class are built out of families of sets representing the wires
and boxes of the graphical calculus, together with operations on those
sets encoding the wiring discipline. For the graphical calculus
in~(i) above, these structures are the
\emph{polycategories} of~\cite{Szabo1975Polycategories}. These were
explicitly introduced as semantic models for a two-sided propositional
sequent calculus; although originally this was the classical Genzten
calculus, it later became clear~\cite{Lambek1993From} that they encode
precisely the sequent calculus of multiplicative linear logic.

\begin{Defn}
  \label{def:33}
  A small (symmetric) polycategory $\C$ comprises a set $\ob(\C)$ of
  \emph{objects}; sets $\C(\vec A; \vec B)$ of \emph{morphisms} for
  each pair of lists $\vec A = (A_1,\ldots,A_n)$ and
  $\vec B = (B_1, \dots, B_m)$ of objects; and the following further
  data:
  \begin{itemize}
  \item \emph{Identity} morphisms $\id_A \in \C(A; A)$ for each object
    $A$.\vskip0.25\baselineskip
  \item \emph{Composition} operations giving for each
    $f \in \C(\vec A; \vec B)$ and $g \in \C(\vec C; \vec D)$ and
    indices $i,j$ with $B_i = C_j$, a morphism
    \begin{equation*}
      g \icomp{j}{i} f \in \C(\vec C_{\scriptscriptstyle < j}, \vec A,
      \vec C_{\scriptscriptstyle > j}; \vec B_{\scriptscriptstyle <i},
      \vec D, \vec B_{\scriptscriptstyle >i})\rlap{ ;}
    \end{equation*}
    here we use comma to denote concatenation of lists, and write
    $\vec C_{\scriptscriptstyle < j}$ for the list
    $(C_1, \dots, C_{j-1})$, and so on.\vskip0.25\baselineskip
  \item \emph{Exchange} operations giving for each
    $f \in \C(\vec A; \vec B)$ and permutations $\varphi \in \aut_n$
    (the symmetric group on $n$ letters) and $\psi \in \aut_m$ an
    element
    \begin{equation*}
      \psi \cdot f \cdot \varphi \in \C(\vec A_{\varphi}; {\vec B}_{\psi^{-1}})
    \end{equation*}
    where $\vec A_{\varphi}$ denotes the list
    $(A_{\varphi(1)}, \dots, A_{\varphi(n)})$ and likewise for
    $\vec B_{\psi^{-1}}$.
  \end{itemize}
  These data are required to satisfy the axioms of
  Definition~\ref{def:44} below.
\end{Defn}

If $\C$ is a polycategory, then we think of elements of $\ob(\C)$ as
wire-labels, and elements of $\C(\vec A; \vec B)$ as boxes of the
form~\eqref{eq:21}. The operations of a polycategory now correspond to
the elementary wiring operations on such boxes. The identity morphisms
can be depicted as bare wires; composition $g \icomp{j}{i} f$ as the
plugging of the $i$th output of $f$ into the $j$th input of $g$, as on
the left below; and exchange as the rearrangement of input or output
wires, as on the right below.

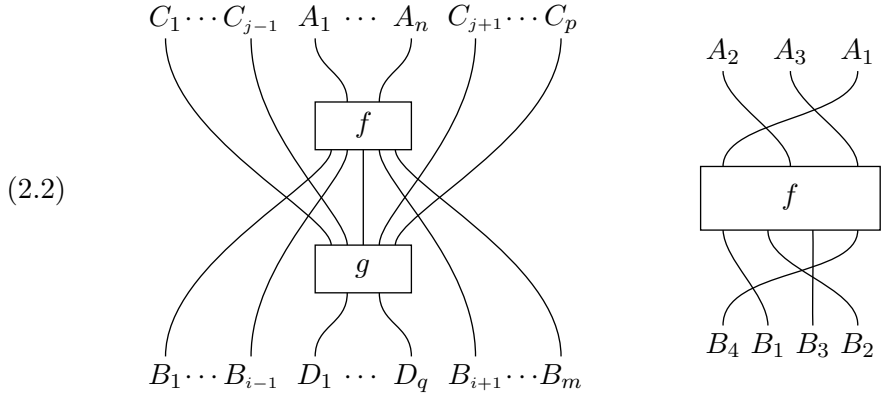
\begin{equation}
  \begin{tikzpicture}[y=0.80pt, x=0.80pt, yscale=-1.000000, xscale=1.000000, inner sep=0pt, outer sep=0pt]
    \path[draw=black,line join=miter,line cap=butt,even odd rule,line
    width=0.500pt,rounded corners=0.0000cm] (60.0000,902.3622) rectangle
    (105.0000,924.8622);
    \path[draw=black,line join=miter,line cap=butt,even odd rule,line width=0.500pt]
    (60.0000,872.3622) .. controls (60.0000,887.2619) and (75.2525,887.9675) ..
    (75.0000,902.3622) node[above=0.1cm,at start] {$A_1$};
    \path (82.500,863.3622) node {$\cdots$};
    \path (8.000,863.3622) node {$\cdots$};
    \path (158.000,863.3622) node {$\cdots$};
    \path (82.500,1032.3622) node {$\cdots$};
    \path (8.000,1032.3622) node {$\cdots$};
    \path (158.000,1032.3622) node {$\cdots$};
    \path[draw=black,line join=miter,line cap=butt,even odd rule,line width=0.500pt]
    (90.0000,902.3622) .. controls (90.0000,887.7150) and (105.0000,887.3622) ..
    (105.0000,872.3622) node[above=0.1cm,at end] {$A_n$};
    \path[draw=black,line join=miter,line cap=butt,even odd rule,line width=0.500pt]
    (67.5000,924.8622) .. controls (68.0051,939.2568) and (-10.0000,977.3621) ..
    (-10.0000,1022.3622) node[below=0.1cm,at end] {$B_1$}; 
    \path[draw=black,line join=miter,line cap=butt,even odd rule,line width=0.500pt]
    (97.5000,924.8622) .. controls (97.7525,940.2670) and (175.0000,977.2098) ..
    (175.0000,1022.3622) node[below=0.1cm,at end] {$B_{m}$};
    \path[fill=black,line join=miter,line cap=butt,line width=0.500pt]
    (82.6331,913.2857) node {$f$};
    \path[draw=black,line join=miter,line cap=butt,even odd rule,line width=0.500pt]
    (75.0000,924.8622) .. controls (75.5051,939.2568) and (30.0000,977.3621) ..
    (30.0000,1022.3622) node[below=0.1cm,at end] {$B_{\scriptscriptstyle
        i-1}$};
    \path[draw=black,line join=miter,line cap=butt,even odd rule,line width=0.500pt]
    (90.0000,924.8622) .. controls (90.2525,940.2670) and (135.0000,977.3621) ..
    (135.0000,1022.3622) node[below=0.1cm,at end] {$B_{\scriptscriptstyle i+1}$};
    \path[draw=black,line join=miter,line cap=butt,even odd rule,line width=0.500pt]
    (82.5000,924.8622) -- (82.5000,969.8622);
    \path[draw=black,line join=miter,line cap=butt,even odd rule,line
    width=0.500pt,rounded corners=0.0000cm] (60.0000,969.8622) rectangle
    (105.0000,992.3622);
    \path[draw=black,line join=miter,line cap=butt,even odd rule,line width=0.500pt]
    (-10.0000,872.3622) .. controls (-10.0000,917.3622) and (67.2475,954.9624) ..
    (67.5000,969.8622) node[above=0.1cm,at start] {$C_1$};
    \path[draw=black,line join=miter,line cap=butt,even odd rule,line width=0.500pt]
    (30.0000,872.3622) .. controls (30.0000,917.3622) and (74.7475,954.9624) ..
    (75.0000,969.8622) node[above=0.06cm,at start]
    {$C_{\scriptscriptstyle j-1}$};
    \path[draw=black,line join=miter,line cap=butt,even odd rule,line width=0.500pt]
    (135.0000,872.3622) .. controls (135.0000,917.3622) and (89.7475,954.9624) ..
    (90.0000,969.8622) node[above=0.08cm,at start]
    {$C_{\scriptscriptstyle j+1}$};
    \path[draw=black,line join=miter,line cap=butt,even odd rule,line width=0.500pt]
    (175.0000,872.3622) .. controls (175.0000,917.3622) and (97.2475,954.9624) ..
    (97.5000,969.8622) node[above=0.07cm,at start]
    {$C_p$};
    \path[draw=black,line join=miter,line cap=butt,even odd rule,line width=0.500pt]
    (75.0000,992.3622) .. controls (75.0000,1007.3622) and (60.0000,1007.3622) ..
    (60.0000,1022.3622) node[below=0.1cm,at end] {$D_1$};
    \path[draw=black,line join=miter,line cap=butt,even odd rule,line width=0.500pt]
    (90.0000,992.3622) .. controls (90.0000,1007.3622) and (105.0000,1007.3622) ..
    (105.0000,1022.3622) node[below=0.1cm,at end] {$D_q$};
    \path[fill=black,line join=miter,line cap=butt,line width=0.500pt]
    (82.2759,981.1428) node{$g$};
  \end{tikzpicture}
  \qquad \qquad
  \begin{tikzpicture}[y=0.80pt, x=0.80pt, yscale=-1.000000, xscale=1.400000, inner sep=0pt, outer sep=0pt]
    \path[draw=black,line join=miter,line cap=butt,even odd rule,line
    width=0.500pt,rounded corners=0.0000cm] (60.0000,902.3622) rectangle
    (120.0000,932.3622);
    \path[draw=black,line join=miter,line cap=butt,even odd rule,line width=0.500pt]
    (67.5000,857.3622) .. controls (67.5000,872.2619) and (90.2525,887.9675) ..
    (90.0000,902.3622) node[above=0.1cm,at start] {$A_2$};
    \path[draw=black,line join=miter,line cap=butt,even odd rule,line width=0.500pt]
    (112.5000,902.3622) .. controls (112.5000,887.7150) and (90.0000,872.3622) ..
    (90.0000,857.3622) node[above=0.1cm,at end] {$A_3$};
    \path[draw=black,line join=miter,line cap=butt,even odd rule,line width=0.500pt]
    (67.5000,932.3622) .. controls (68.0051,946.7568) and (82.2475,962.4624) ..
    (82.5000,977.3622) node[below=0.1cm,at end] {$B_1$};
    \path[draw=black,line join=miter,line cap=butt,even odd rule,line width=0.500pt]
    (97.5000,932.3622) .. controls (97.7525,947.7670) and (97.2475,962.2099) ..
    (97.5000,977.3622) node[below=0.1cm,at end] {$B_3$};
    \path[fill=black,line join=miter,line cap=butt,line width=0.500pt]
    (90,916.8571) node {$f$};
    \path[draw=black,line join=miter,line cap=butt,even odd rule,line width=0.500pt]
    (67.5000,902.3622) .. controls (67.5000,879.8622) and (112.5000,879.8622) ..
    (112.5000,857.3622) node[above=0.1cm,at end] {$A_1$};
    \path[draw=black,line join=miter,line cap=butt,even odd rule,line width=0.500pt]
    (82.5000,932.3622) .. controls (82.2575,950.2937) and (113.0051,958.9028) ..
    (112.5000,977.3622)node[below=0.1cm,at end] {$B_2$};
    \path[draw=black,line join=miter,line cap=butt,even odd rule,line width=0.500pt]
    (112.5000,932.3622) .. controls (112.5000,954.8622) and (67.2475,950.3165) ..
    (67.5000,977.3622)node[below=0.1cm,at end] {$B_4$};
  \end{tikzpicture}
  \label{eq:multicomp}
\end{equation}
[Note that the identities of a polycategory involve only a single
object rather than a list. A geometric explanation for
this is that all the graphs occurring in polycategorical composition
are \emph{connected}, whereas the identity on a list of objects would
be an unconnected graph.]

In terms of the graphical calculus, the axioms for a polycategory can
be seen simply as asserting that various ways of wiring together a
diagram of boxes coincide. We now give these axioms in full, mainly to
show how unpalatable they are when presented algebraically, and
without any real expectation that the reader should work through the
details.

\begin{Defn}
  \label{def:44}
  The axioms for a polycategory $\C$ are:
  \begin{itemize}
  \item The \emph{unit} axioms:
    \begin{equation*}
      f \icomp{i}{1} \id_{A_i} = f = \id_{B_j} \, \icomp{1}{j} f
    \end{equation*}
    for all $f \in \C(\vec A; \vec B)$ and valid indices
    $i,j$.\vskip0.25\baselineskip
  \item The \emph{associativity} axiom:
    \begin{equation}\label{eq:43}
      (h \icomp{\ell}{k} g) \icomp{\bar \jmath}{i} f = h
      \icomp{\ell}{\bar k} (g
      \icomp{j}{i} f)
    \end{equation}
    for all $f \in \C(\vec A; \vec B)$, $g \in \C(\vec C; \vec D)$ and
    $h \in \C(\vec E; \vec F)$ and all indices $i,j,k,\ell$ with
    $B_i = C_j$ and $D_k = E_\ell$. Here, $\bar \jmath = j + \ell -1$ and
    $\bar k = k + i -1$.\vskip0.25\baselineskip
  \item The \emph{left interchange} axiom:
    \begin{equation*}
      (h \icomp{k_2}{j} g) \icomp{k_1}{i} f = 
      \psi \cdot ((h \icomp{k_1}{i} f) \icomp{{\bar k}_2}{j} g)\rlap{ .}
    \end{equation*}
    for all $f \in \C(\vec A; \vec B)$, $g \in \C(\vec C; \vec D)$ and
    $h \in \C(\vec E; \vec F)$ and all indices $i, j$ and $k_1 < k_2$
    such that $B_i = E_{k_1}$ and $D_j = E_{k_2}$. Here,
    $\bar{k}_2 = k_2 + n -1$ where $n$ is the length of the list
    $\vec A$, and $\psi$ is the permutation for which
    $(\vec B_{\scriptscriptstyle < i}, \vec D_{\scriptscriptstyle < j},
    \vec F, \vec D_{\scriptscriptstyle > j}, \vec B_{\scriptscriptstyle >
      i})_\psi = (\vec D_{\scriptscriptstyle < j}, \vec
    B_{\scriptscriptstyle < i}, \vec F, \vec B_{\scriptscriptstyle > i},
    \vec D_{\scriptscriptstyle > j})$.\vskip0.25\baselineskip
  \item The \emph{right interchange} axiom:
    \begin{equation*}
      g \icomp{j}{i_1} (h \icomp{k}{i_2} f) =
      (h \icomp{k}{{\bar \imath}_2} (g \icomp{j}{i_1} f)) \cdot \varphi\rlap{ .}
    \end{equation*}
    for all $f \in \C(\vec A; \vec B)$, $g \in \C(\vec C; \vec D)$ and
    $h \in \C(\vec E; \vec F)$ and all indices $i_1 < i_2$ and $j$, $k$
    such that $B_{i_1} = C_{j}$ and $B_{i_2} = E_{k}$. Here,
    $\bar \imath = i + m - 1$ where $m$ is the length of the list
    $\vec D$, and $\varphi$ is the permutation for which
    $(\vec E_{\scriptscriptstyle < k}, \vec C_{\scriptscriptstyle < j},
    \vec A, \vec C_{\scriptscriptstyle > j}, \vec E_{\scriptscriptstyle >
      k})_\varphi = (\vec C_{\scriptscriptstyle < j}, \vec
    E_{\scriptscriptstyle < k}, \vec A, \vec E_{\scriptscriptstyle > k},
    \vec C_{\scriptscriptstyle > j})$.\vskip0.25\baselineskip
  \item The \emph{action} axioms:
    \begin{equation*}
      \id_m \cdot f \cdot \id_n = f \quad \text{and} \quad 
      (\psi_2\psi_1) \cdot f \cdot (\varphi_1 \varphi_2) = \psi_2 \cdot
      (\psi_1 \cdot f \cdot \varphi_1) \cdot \varphi_2
    \end{equation*}
    for all $f \in \C(\vec A; \vec B)$ and suitable permutations
    $\varphi_1$, $\varphi_2$, $\psi_1$ and $\psi_2$.
    \vskip0.25\baselineskip
  \item The \emph{equivariance of composition} axiom:
    \begin{equation*}
      (\psi_2 \cdot g \cdot \varphi_2) \icomp{j}{i} (\psi_1 \cdot f \cdot \varphi_1) = \bar \psi \cdot (g \icomp{\varphi_2(j)}{\psi_1^{-1}(i)} f) \cdot \bar\varphi
    \end{equation*}
    for all $f \in \C(\vec A; \vec B)$, $g \in \C(\vec C; \vec D)$, and
    suitable permutations $\varphi_1, \varphi_2, \psi_1$ and $\psi_2$.
    Here, $\bar \varphi$ is determined by
    $(\vec C_{\scriptscriptstyle < \varphi_2(j)}, \vec A, \vec
    C_{\scriptscriptstyle > \varphi_2(j)})_{\bar \varphi} = ((\vec
    C_{\varphi_2})_{\scriptscriptstyle < j}, \vec A_{\varphi_1}, (\vec
    C_{\varphi_2})_{\scriptscriptstyle > j})$ and $\bar \psi$ by
    $(\vec B_{\scriptscriptstyle < \psi_1^{-1}(i)}, \vec D, \vec
    B_{\scriptscriptstyle > \psi_1^{-1}(i)})_{\bar\psi^{-1}} = ((\vec
    B_{\psi_1^{-1}})_{\scriptscriptstyle < i}, \vec D_{\psi_2^{-1}},
    (\vec B_{\psi_1^{-1}})_{\scriptscriptstyle > i})$.
  \end{itemize}
\end{Defn}

In a similar way, we can associate algebraic structures to the other
graphical calculi described above; since these calculi extend (i) with
more permissive wiring disciplines, the associated structures extend
polycategories with more permissive composition operations. For the
calculi with wiring disclipines~(ii)',~(ii) and (iii), the structures
obtained are, respectively (coloured)
\emph{properads}~\cite{Vallette2007A-Koszul}, (coloured)
\textsc{prop}s~\cite{Mac-Lane1965Categorical}, and (many-sorted)
\emph{Lawvere theories}~\cite{Lawvere1963Functorial}. We now sketch
the definitions for (ii)' and (ii), leaving (iii) as an exercise to
the reader.

\begin{Defn}
  \label{def:45}
  A \emph{coloured properad} $\C$ is given by the same data as for a
  polycategory, except for the operation of composition, which is
  generalised as follows. Given morphisms $f \in \C(\vec A; \vec B)$
  and $g \in \C(\vec C; \vec D)$, and non-empty sequences of indices
  $I = \{i, \dots, i+k\}$ and $J = \{j, \dots, j+k\}$ such that
  $B_{i+\ell} = C_{j+\ell}$ for each $0 \leqslant \ell \leqslant k$,
  there is a composite morphism
  \begin{equation}\label{eq:38}
    g \icomp{J}{I} f \in \C(\vec C_{< j}, \vec A, \vec C_{> j+k}; \vec
    B_{< i}, \vec D, \vec B_{> i+k})\rlap{ .}
  \end{equation}
  These data satisfy axioms identical in form to
  Definition~\ref{def:44}.

  A \emph{coloured \textsc{prop}} $\C$ is a coloured properad augmented
  with a morphism $0 \in \C(\ \mathord; \ )$ (representing the empty
  string diagram) and an operation which to morphisms
  $f \in \C(\vec A; \vec B)$ and $g \in \C(\vec C; \vec D)$, associates
  a morphism
  \begin{equation}\label{eq:39}
    g \icomp{\emptyset}{\emptyset} f \in \C(\vec A, \vec C; \vec B, \vec D)\rlap{ ,}
  \end{equation}
  (representing $f$ and $g$ placed alongside each other), all subject
  to suitable axioms.
\end{Defn}

Of course, for each of these algebraic structures there is an
associated notion of structure-preserving map giving the morphisms of
a category:

\begin{Defn}
  \label{def:46}
  If $\C$ and $\D$ are polycategories, then a \emph{polyfunctor}
  $F \colon \C \rightarrow \D$ comprises an assignation on objects
  $F \colon \ob(\C) \rightarrow \ob(\D)$ and assignations
  \begin{align*}
    \C(\vec A; \vec B) &\rightarrow \D(F\vec A; F\vec
    B) \\
    f & \mapsto Ff
  \end{align*}
  on morphisms (where $F\vec A = (FA_1, \dots, FA_n)$ and similarly for
  $F \vec B$), such that for all suitable $A, f, g, i, j, \psi$ and
  $\varphi$ we have
  \begin{equation*}
    F(\id_A) = \id_{FA}\ \text, \quad F(g \icomp{j}{i} f) =
    Fg \icomp{j}{i} Ff \quad \text{and} \quad F(\psi \cdot f \cdot
    \varphi) = \psi \cdot Ff \cdot \varphi\rlap{ .}
  \end{equation*}
  We write $\cat{Polycat}$ for the category of small polycategories
  and polyfunctors. In a similar manner we have categories
  $\cat{Properad}$ and $\cat{Prop}$ of properads and \textsc{prop}s.
\end{Defn}

\subsection{Monads from graphical calculi}
\label{sec:monads-from-graph}

We now explain how the algebraic structures of the preceding section
can be represented as algebras for suitable monads on a presheaf
category. The presheaf category in question encodes the objects and
morphisms of a polycategory, properad or \textsc{prop}:

\begin{Defn}
  \label{def:36}
  Let $\PP$ be the category with object-set
  $\ens{\star} + \Nat \times \Nat$ and non-identity maps
  $\sigma^{n,m}_1, \dots, \sigma^{n,m}_n, \tau^{n,m}_1, \dots,
  \tau^{n,m}_m \colon \star \to (n,m)$ (we omit superscripts in the
  sequel for readability). A presheaf $X \in \psh[\PP]$ is called a
  \emph{polygraph}\footnote{Our usage follows~\cite{Blute1996Natural};
    note that these polygraphs are completely unrelated to those
    of~\cite{Burroni1993Higher-dimensional}.}, with elements of
  $X(\star)$ being called \emph{vertices}, and elements of $X(n,m)$
  \emph{edges}. We write $s_1, \dots, s_n$ and $t_1, \dots, t_m$ for
  $X(\sigma_1), \dots, X(\sigma_n)$ and $X(\tau_1), \dots, X(\tau_m)$,
  and call the images of $e \in X(n,m)$ under these maps its
  \emph{sources} and \emph{targets} respectively.
\end{Defn}

There are forgetful functors from $\cat{Polycat}$, $\cat{Properad}$ or
$\cat{Prop}$ to $\psh[\PP]$ sending a polycategory, properad, or
\textsc{prop} to its underlying polygraph of objects and morphisms,
and these functors are monadic, so allowing us to identify the
structures at issue with algebras for the induced monad on
$\psh[\PP]$. In fact, we may describe these monads explicitly; we now
do this in detail for polycategories, and indicate how this should be
adapted in the other cases.

The key observation is that objects of $\psh[\PP]$ can be seen as
combinatorial representations of wiring diagrams of the kind drawn
above. For instance, the diagram~\eqref{eq:21} for a box with $n$
inputs and $m$ outputs corresponds to the representable presheaf
$\yoneda_{(n,m)} = \PP(\thg, (n,m))$, while the diagram
in~\eqref{eq:multicomp} for a composite $g \icomp{j}{i} f$ corresponds
to a pushout
\begin{equation}\label{eq:36}
  \cd[@-0.3em]{
    {\yoneda_\star} \ar[r]^-{\yoneda_{\tau_i}} \ar[d]_{\yoneda_{\sigma_j}} &
    {\yoneda_{(n,m)}} \ar[d]^{u} \\
    {\yoneda_{(p,q)}} \ar[r]_-{v} & 
    \pushoutcorner {\yoneda_{(p,q)} \ifcomp{j}{i} \yoneda_{(n,m)}} 
  }
\end{equation}
in $\psh[\PP]$. Writing
$A = \yoneda_{(p,q)} \ifcomp{j}{i} \yoneda_{(n,m)}$, the further
composite $h \icomp{\ell}{\bar k} (g \icomp{j}{i} f)$ in the
associativity axiom~\eqref{eq:43} corresponds to the pushout in
$\psh[\PP]$ as to the left in:
\begin{equation*}
  \cd[@!C@C-3em@-0.3em]{
    {\yoneda_\star} \ar[r]^-{\yoneda_{\tau_k}} \ar[d]_{\yoneda_{\sigma_\ell}} &
    {\yoneda_{(p,q)}} \ar[r]^-{v} & 
    {A} \ar[d]^-{} \\
    {\yoneda_{(r,s)}} \ar[rr]_-{} & &
    \pushoutcorner {\yoneda_{(r,s)} \ifcomp{\ell}{\bar k} A}
  } \qquad \ 
  \cd[@-0.4em@C-1em]{
    & \yoneda_\star \ar[dl]_-{\yoneda_{\sigma_\ell}} \ar[dr]_-{\yoneda_{\tau_k}} & & 
    \yoneda_\star \ar[dl]^-{\yoneda_{\sigma_j}} \ar[dr]^-{\yoneda_{\tau_i}} \\
    \yoneda_{(r,s)} & & \yoneda_{(p,q)} & & \yoneda_{(n,m)}\rlap{ ,}
  }
\end{equation*}
which is isomorphic to the polygraph representing
$(h \icomp{\ell}{k} g) \icomp{\bar \jmath}{i} f$, since both are
colimits for the diagram above right. Iteratively taking pushouts of
the preceding kind yields the following class of polygraphs describing
the compound wiring operations of a polycategory; eventually, in
Section~\ref{sec:first-appl-polyc} below, we will be able to generate
all of these shapes from those for the basic wiring operations, but
for the moment we give a more hands-on construction.

\begin{Defn}
  \label{def:43}
  A \emph{finite polygraph} is one with finitely many vertices and
  edges. An \emph{$(n,m)$-labelling} of a finite polygraph is given by
  choices of vertices $\ell_1, \dots, \ell_n$ and $r_1, \dots, r_m$,
  called the \emph{leaves} and \emph{roots} respectively. An
  \emph{isomorphism} of labelled polygraphs is one respecting the
  labellings. Let $\L(n,m)$ be a set of isomorphism-class
  representatives of $(n,m)$-labelled finite polygraphs; we write
  $\abs X$ for the underlying polygraph of $X \in \L(n,m)$ and
  $\ell^X$ and $r^X$ for the labellings. Now:
  \begin{enumerate}[(a)]
  \item Let $\id \in \L(1,1)$ be $\yoneda_\star$ labelled in the
    unique possible way;\vskip0.25\baselineskip
  \item Given $X \in \L(n,m)$, $Y \in \L(p,q)$ and indices
    $1 \leqslant i \leqslant m$ and $1 \leqslant j \leqslant p$, let
    $Y \ifcomp{j}{i} X \in \L(n+p-1,m+q-1)$ be such that there is a
    pushout of underlying polygraphs
    \begin{equation}\label{eq:37}
      \cd[@-0.3em]{
        {\yoneda_\star} \ar[r]^-{r^X_i} \ar[d]_{\ell^Y_j} &
        {\abs{X}} \ar[d]^{u} \\
        {\abs{Y}} \ar[r]_-{v} & 
        \abs{Y \ifcomp{j}{i} X} 
      }
    \end{equation}
    with the labelling of the leaves and roots given respectively by:
    \begin{align*}
      &v\ell^Y_1, \dots, 
      v\ell^Y_{j-1}, \ 
      u\ell^X_1, \dots, 
      u\ell^X_n, \ 
      v\ell^Y_{j+1}, \dots, 
      v\ell^Y_{p}\\
      \text{and} \quad 
      &u r^X_1, \dots,
      u r^X_{i - 1},\ 
      v r^Y_1, \dots,
      v r^Y_q, \ 
      u r^X_{i+1}, \dots,
      u r^X_m\rlap{ .}
    \end{align*}
  \item For any $X \in \L(n,m)$, $\varphi \in \aut_n$ and
    $\psi \in \aut_m$, let $\psi \cdot X \cdot \varphi \in \L(n,m)$ be
    $\abs{X}$ labelled by $\ell_{\varphi(1)}, \dots, \ell_{\varphi(n)}$
    and
    $r_{\psi^{-1}(1)}, \dots, r_{\psi^{-1}(m)}$.\vskip0.25\baselineskip
  \item For each $n,m \in \mathbb N$, let $\spn{n,m} \in \L(n,m)$ be
    $\yoneda_{(n,m)}$ labelled by
    $\yoneda_{\sigma_1}, \dots, \yoneda_{\sigma_n}$ and
    $\yoneda_{\tau_1}, \dots, \yoneda_{\tau_m}$.
  \end{enumerate}
  Now let $\T(n,m) \subset \L(n,m)$ be the subsets obtained by closing
  the elements $\spn{n,m}$ in (d) under the operations in
  (a)--(c).\footnote{Note that implicit in these definitions are the
    assumptions that the elements $\id$ and $\spn{n,m}$ of (a) and (d)
    are the chosen representatives of their isomorphism-classes, and
    that relabelling a representative $X \in \L(n,m)$ as in (c) yields
    another such; we are clearly at liberty to make these assumptions.}
\end{Defn}

As noted above, the labelled polygraphs in the sets $\T(n,m)$
represent the compound wiring operations of a polycategory;
following~\cite{Kock2011Polynomial,Kock2016}, we may characterise them
in a direct combinatorial manner.

\begin{Defn}
  \label{def:42}
  For any polygraph $X$, we write $G_X$ for the undirected multigraph
  obtained as follows: the nodes are the disjoint union of the sets of
  edges and of vertices of $X$, and there is an arc $v \frown e$ for
  each way that $v$ is a source or a target of $e$. A polygraph $X$ is
  called a \emph{polycategorical tree} if it has finitely many edges
  and vertices, and moreover:
  \begin{itemize}
  \item Each vertex of $X$ is a source of at most one edge;
  \item Each vertex of $X$ is a target of at most one edge;
  \item $G_X$ is acyclic (in particular without multiple edges) and
    connected (in particular non-empty).
  \end{itemize}
  A \emph{labelled polycategorical tree} is a polycategorical tree
  equipped with an $(n,m)$-labelling for which $\ell_1, \dots, \ell_n$
  enumerate the \emph{inputs}, i.e., the vertices which are not the
  targets of any edge, and $r_1, \dots, r_m$ enumerate \emph{outputs},
  i.e., the vertices which are not the source of any edge.
\end{Defn}

\begin{Lemma}
  \label{lem:21}
  $\T(n,m)$ is the set of isomorphism-class representatives of
  $(n,m)$-labelled polycategorical trees.
\end{Lemma}

\begin{proof}
  Each $\spn{n,m}$ is a labelled polycategorical tree, and labelled
  polycategorical trees are closed under (a)--(c) above. Thus,
  $(n,m)$-labelled polycategorical trees contain all of $\T(n,m)$.
  Conversely, we may show that any labelled polycategorical tree $T$
  is in $\T(n,m)$ by induction on the number of edges in $T$. If $T$
  contains no edge, then it must be $\id \in \T(1,1)$. If $|T| = e$,
  then $T \in \T(n,m)$ by (c)--(d). Otherwise, by acyclicity and
  finiteness, we may find an edge $e$ whose sources are not targets of
  any edge. Removing $e$ from $T$, together with any isolated vertices
  that this creates, leaves us with a positive number of connected
  components $T_1,\ldots,T_p$ which all contain at least one edge. By
  acyclicity, $e$ is linked to each $T_j$ through exactly one vertex.
  If $p = 1$, then $T$ is a composite of $T_1$ and $e$. Otherwise, let
  $T'$ denote $T$ with $T_p$ and any consequent isolated vertices
  removed; by the inductive hypothesis, $T'$ is in $\T(n',m')$, where
  $n'$ and $m'$ respectively denote the numbers of inputs and outputs
  of $T'$, and $T$ is a composite of $T_p$ and $T'$. This completes
  the proof.
\end{proof}

The labelled polycategorical trees may now be used to provide a
concrete description of the left adjoint $F_\PP$ to
$U \colon \cat{Polycat} \rightarrow \psh[\PP]$. Given
Lemma~\ref{lem:21}, the proof of the following result is
straightforward, if tedious; the reader may reconstruct it by suitably
adapting~\cite[Proposition~1.9.2]{Markl2002Operads} or~\cite[\sec
2.2.7]{Kock2011Polynomial}.

\begin{Prop}
  \label{prop:multicat}
  The free polycategory $F_\PP X$ on a polygraph $X \in \psh[\PP]$ has
  object set $X(\star)$ and morphism sets
  $(F_\PP X)(v_1, \dots, v_n; w_1, \dots, w_m)$ given by
  \begin{equation}\label{eq:1}
    \sum_{T \in \trees(n,m)}
    \{\,f \colon \abs T \rightarrow X \text{ in } \psh[\PP] : f(\ell^T_i) = v_i \text{ and } f(r^T_j) = w_j\}\rlap{ .}
  \end{equation}
  The identity morphism in $(F_\PP X)(v; v)$ is the pair
  $(\id \in \trees(1,1), v \colon \yoneda_\star \rightarrow X)$; the
  composite $(S, f) \icomp{j}{i} (T, g)$ is given by
  $(S \ifcomp{j}{i} T, f \ifcomp{j}{i} g)$, where $f \ifcomp{j}{i} g$
  is the unique map out of the pushout~\eqref{eq:37} induced by $f$ and
  $g$; and the exchange operation is defined by
  $\psi \cdot (T, f) \cdot \varphi = (\psi \cdot T \cdot \varphi, f)$.
\end{Prop}

We may proceed in a similar manner to obtain explicit descriptions of
the monads for properads and for \textsc{prop}s on $\psh[\PP]$. In the
case of properads, we generalise the labelled polycategorical trees to
labelled \emph{properadic graphs}, by modifying clause (b) of
Definition~\ref{def:42} so as to allow for pushouts of the form
\begin{equation*}
  \cd{
    {\yoneda_\star + \dots + \yoneda_\star} \ar[r]^-{\spn{r^X_i,
        \dots, r^X_{i+k}}} \ar[d]_{\spn{\ell^X_j,
        \dots, \ell^X_{j+k}}} &
    {\abs{X}} \ar[d]^{u} \\
    {\abs{Y}} \ar[r]_-{v} & 
    \abs{Y \ifcomp{J}{I} X} \rlap{ .}
  }
\end{equation*}
We may now modify Lemma~\ref{lem:21} to characterise the properadic
graphs by requiring $G_X$ to be acyclic as a \emph{directed} graph,
though still connected as an \emph{undirected} graph. With this
modification, Proposition~\ref{prop:multicat} carries through, to give
an explicit description of the monad for properads on $\psh[\PP]$.

\textsc{Prop}s may almost be treated in the same way. We can augment
Definition~\ref{def:43} suitably to obtain the class $\G$ of
\emph{graphs for }\textsc{prop}s, and then modify Lemma~\ref{lem:21}
to obtain a combinatorial characterisation of these graphs; in this
case $G_X$ need only be acyclic as a directed graph, as
in~\cite{Joyal1991The-geometry}. However, the formula for the morphism
sets of the free \textsc{prop} is not quite given by the obvious
adaptation of~\eqref{eq:1}. This is essentially because the axioms for
a $\textsc{prop}$ are susceptible to the well-known
\emph{Eckmann--Hilton argument}~\cite{EckmannHilton}; among other
things, this implies that for any \textsc{prop} $\C$ and any
$f, g \in \C(\ ; \ )$, we have
$g \icomp{\emptyset}{\emptyset} f = f \icomp{\emptyset}{\emptyset} g$:
\begin{equation*}
  \begin{tikzpicture}
    \node[rectangle,draw,minimum height=3em,minimum width=3em] (a) {$f$} ;
    \node[rectangle,draw,minimum height=3em,minimum width=3em] (b) at (2,0) {$g$} ;
  \end{tikzpicture}
  \qquad = \qquad
  \begin{tikzpicture}
    \node[rectangle,draw,minimum height=3em,minimum width=3em] (a) {$g$} ;
    \node[rectangle,draw,minimum height=3em,minimum width=3em] (b) at (2,0) {$f$} ;
  \end{tikzpicture}\quad \rlap{ .}
\end{equation*}

Since the shape of the composition just depicted is encoded by the
graph for \textsc{prop}s
$R = \spn{0,0} \ifcomp{\emptyset}{\emptyset} \spn{0,0} \in \G(0,0)$
with underlying presheaf $\abs{R} = \yoneda_{(0,0)}+\yoneda_{(0,0)}$,
the formula for the free \textsc{prop} $F_\PP X$ on a polygraph $X$
must differ from~\eqref{eq:1} in identifying, among other things, the
pair of elements $\spn{f,g} \colon \abs R \rightarrow X$ and
$\spn{g,f} \colon \abs R \rightarrow X$ in $(F_\mathsf{P}X)(\ ;\ )$
for any $f,g \in X(\ ;\ )$. This is an instance of a more general
phenomenon: each graph for \textsc{prop}s $T \in \G(n,m)$ may admit a
non-trivial group $\aut_T$ of label-preserving automorphisms
(permuting unlabelled isomorphic connected components), and the
construction of the free \textsc{prop} must quotient out by the action
of these automorphisms.
\begin{Prop}
  The free \textsc{prop} $F_{\PP} X$ on $X \in \psh[\PP]$ has object
  set $X(\star)$ and morphism sets
  $(F_\PP X)(v_1, \dots, v_n; w_1, \dots, w_m)$ given by
  \begin{equation*}
    \sum_{T \in \G(n,m)} \{\,f \colon \abs{T} \rightarrow X \text{ in } \psh[\PP] : f(\ell^T_i) = v_i \text{, } f(r^T_j) = w_j\}_{/ \aut_{T}}
  \end{equation*}
  with remaining structure defined analogously to
  Proposition~\ref{prop:multicat} above.
\end{Prop}

\subsection{An alternative presentation}
\label{sec:an-altern-pres}

There is another way of presenting polycategories, properads and
\textsc{prop}s as the algebras for a monad on a presheaf category.
Though it is further away from the graphical intuition, it is quite a
common approach in mathematical practice, and still fits into the
general framework we will develop; it therefore seems to be worth
describing here. The idea is to incorporate the exchange operations
into the underlying presheaf:

\begin{Defn}
  \label{def:35}
  Let $\PS$ be the category obtained from $\PP$ by adjoining arrows
  $\xi_{\varphi, \psi} \colon (n,m) \rightarrow (n,m)$ for each
  $\varphi \in \aut_n$ and $\psi \in \aut_m$, subject to the
  equations:
  \begin{align*}
    \xi_{\id_n, \id_m} &= \id_{(n,m)}&
    \xi_{\varphi_1, \psi_1} \circ \xi_{\varphi_2, \psi_2} &=
    \xi_{\varphi_1\varphi_2, \psi_2\psi_1}\\ 
    \xi_{\varphi, \psi} \circ \sigma_i &= \sigma_{\varphi(i)}
    & 
    \xi_{\varphi, \psi} \circ \tau_{\psi(j)} &= \tau_{j}\rlap{ .}
  \end{align*}
  A presheaf $X \in \psh[\PS]$ is called a \emph{symmetric polygraph}.
\end{Defn}
Like before, there are monadic forgetful functors from
$\cat{Polycat}$, $\cat{Properad}$ and $\cat{Prop}$ to $\psh[\PS]$. In
the case of polycategories, the left adjoint can now be described
using labelled \emph{symmetric} polycategorical trees. Let $\L_s(n,m)$
and $\T_s(n,m)$ be defined exactly like $\L(n,m)$ and $\T(n,m)$ in
Definition~\ref{def:43} but working over the category $\psh[\PS]$ of
symmetric polygraphs. Since each symmetric polygraph $\yoneda_{(n,m)}$
has \emph{free} action by $\aut_n \times \aut_m^\mathrm{op}$, the same
will be true of the underlying symmetric polygraph of any
$T \in \T_s(n,m)$; in fact, these $T$'s are precisely the
images\footnote{Though note that non-isomorphic elements of $\T(n,m)$
  may become isomorphic in $\T_s(n,m)$.} of the non-symmetric trees in
$\T(n,m)$ under the free functor $\psh[\PP] \rightarrow \psh[\PS]$.
Just as we saw when considering graphs for \textsc{prop}s in the
previous section, symmetric labelled trees $T \in \T_s(n,m)$ may admit
non-trivial groups $\aut_T$ of label-preserving automorphisms, and the
construction of the free polycategory must quotient out by these in
the same way.

\begin{Prop}
  \label{prop:22}
  The free polycategory $F_{\PS} X$ on $X \in \psh[\PS]$ has object
  set $X(\star)$ and morphism sets
  $(F_\PS X)(v_1, \dots, v_n; w_1, \dots, w_m)$ given by
  \begin{equation*}
    \sum_{T \in \T_s(n,m)} \{\,f \colon \abs{T} \rightarrow X \text{ in } \psh[\PS] : f(\ell^T_i) = v_i \text{, } f(r^T_j) = w_j\}_{/ \aut_{T}}
  \end{equation*}
  with remaining structure defined analogously to
  Proposition~\ref{prop:multicat} above.
\end{Prop}

Exactly the same considerations apply to the cases of properads and
\textsc{prop}s; note that, in the properadic case, the description of
the free properad monad so obtained is that of~\cite{Kock2016}.

\section{Familial functors and shapeliness}
\label{sec:famil-funct-shap}

Now that we have described various ``graphically specified''
structures as algebras for monads on presheaf categories, we begin our
attempts to obtain these monads via a notion of shapeliness. As in the
introduction, our approach will be to seek on the appropriate presheaf
category a \emph{universal} shapely monad $\mathsf U$ with ``exactly
one operation of each shape'', and to generate the monad encoding the
given structure as a suitable submonad of $\mathsf{U}$. In this
section, we look for $\mathsf U$ as a terminal object among
\emph{familially representable}, or more shortly \emph{familial},
endofunctors---ones which pointwise are coproducts of representables.
While this turns out not quite to work, the techniques we develop will
be crucial to our subsequent efforts.

\subsection{Linear operations and familial functors}
\label{sec:line-oper}

The key concept underlying the notion of familial functor is that of a
\emph{linear operation}.

\begin{Defn}
  \label{def:2}
  Given a functor $F \colon \A \rightarrow \B$ and objects $A \in \A$
  and $B \in \B$, an \emph{$F$-operation of input arity $A$ at stage
    $B$} is a map $t \colon B \rightarrow FA$. An $F$-operation $t$ is
  \emph{linear} if it is initial in its connected component of the
  comma category $B \downarrow F$.
\end{Defn}

An operation $t \colon B \rightarrow TA$ of a monad $\mathsf T$ on
$\A$ corresponds to a family of interpretation functions
$\dbr{t} \colon \A(A, X) \rightarrow \A(B,X)$, one for each
$\mathsf T$-algebra $(X,x)$; maps of $B \downarrow T$ account for
reindexing such $\mathsf T$-operations so as to act only on part of
their input arity, so that linearity expresses the idea of an
operation which ``consumes all its input arity''.

\begin{Lemma}
  \label{lem:5}
  An operation $t \colon B \rightarrow FA$ is linear if and only if
  for every square of the following form, there is a unique
  $h \colon A \rightarrow A'$ with $Fh.t = u$; it then follows also
  that $fh = g$.
  \begin{equation}\label{eq:2}
    \cd[@C-0.5em]{
      {B} \ar[r]^-{u} \ar[d]_{t} &
      {FA'} \ar[d]^{Ff} \\
      {FA} \ar[r]^-{Fg} \ar@{.>}[ur]|-{Fh}
      &
      {FA''}
    }
  \end{equation}
\end{Lemma}

\begin{proof}
  This is~\cite[Proposition~0]{Diers1978Spectres}.
\end{proof}

Now a \emph{familial} functor is one whose operations are all
reindexings of linear ones. In giving the definition, we say that
\emph{$Y$ covers $X$} if there is a map $Y \rightarrow X$.
 
\begin{Defn}
  \label{def:3}
  A functor $F \colon \A \rightarrow \B$ is \emph{familial at stage
    $B \in \B$} if each operation in $B \downarrow F$ is covered by a
  linear one; a transformation $\alpha \colon F \Rightarrow G$ is
  \emph{familial at stage $B$} if $F$ and $G$ are so, and the induced
  functor $B \downarrow F \rightarrow B \downarrow G$ preserves linear
  operations. We write simply \emph{familial} to mean ``familial at
  every stage''.
\end{Defn}

Familial functors were introduced by Diers~\cite{Diers1978Spectres};
his terminology is that familial functors are those ``having a left
multiadjoint''. Our name is a shortening of the term ``familially
representable'' used---for the special case
$\B = \cat{Set}$---in~\cite{Johnson1992Algebra}.

\begin{Lemma}
  \label{lem:8}
  A functor $F \colon \A \rightarrow \B$ is familial at stage
  $B \in \B$ if and only if the functor
  $\B(B, F\thg) \colon \A \rightarrow \cat{Set}$ is a (possibly large)
  coproduct of representables.
\end{Lemma}

\begin{proof}
  $F$ is familial at stage $B$ just when $B \downarrow F$ is a
  coproduct of categories with initial objects. This is to say that
  there is an $I$-indexed family of elements
  $\{t_i \in \B(B, FA_i) : i \in I\}$ such that any $f \in \B(B,FA)$
  factors as $f = F\bar f \circ t_i$ for a unique $i \in I$ and
  $\bar f \colon A_i \rightarrow A$; or equally, that
  $\B(B,F\thg) \cong \sum_i \A(A_i, \thg)$.
\end{proof}

\subsection{Pointwise familiality}
\label{sec:pointw-famil-funct}

We will be interested in familial endofunctors of presheaf categories;
later, we will need more general familial functors with
\emph{codomain} a presheaf category. The most relevant kind of
familiality for these is:

\begin{Defn}
  \label{def:1}
  A functor $F \colon \A \rightarrow \psh$ or transformation
  $\alpha \colon F \Rightarrow G \colon \A \rightarrow \psh$ is
  \emph{pointwise familial} if it is familial at all representable
  stages $\yoneda_c \in \psh$; while $F$ is called \emph{small} if
  $\yoneda_c \downarrow F$ has a mere \emph{set} of connected
  components for each $c \in \C$. We write $\fampt(\A, \psh)$ for the
  category of small pointwise familial functors and pointwise familial
  transformations.
\end{Defn}

By Lemma~\ref{lem:8}, $F \colon \A \rightarrow \psh$ is (small)
pointwise familial just when each functor
$(F\thg)c \in [\A, \cat{Set}]$ can be expressed as
\begin{equation}
  \label{eq:41}
  (F\thg)c \cong \textstyle\sum_{t \in Sc} \A(Et, \thg)
\end{equation}
for some set $Sc$ and family of objects $(Et \in \A)_{t \in Sc}$. So,
for example, the ``free polycategory'' endofunctor on the category
$\psh[\PP]$ of polygraphs as in Proposition~\ref{prop:multicat} is
pointwise familial, but the corresponding endofunctor on the category
$\psh[\PS]$ of symmetric polygraphs is \emph{not} so, as it involves
not just coproducts of representables but also quotients by group
actions. We will be able to handle the latter example when we consider
\emph{analytic} functors in the following section.

We now explain how~\eqref{eq:41} allows us to give a compact
represention for pointwise familial functors. Such a functor $F$ is
determined to within isomorphism by the sets $Sc$ and objects
$(Et \in \A)_{t \in Sc}$ as in~\eqref{eq:41}, as $c$ ranges over
$\ob(\C)$, together with information about how these transform under
each $(F\thg)f \colon (F\thg)d \rightarrow (F\thg)c$. More precisely,
$(F\thg)f$ may be identified with a transformation
\begin{equation}\label{eq:42}
  \textstyle\sum_{t \in Sd} \A(Et, \thg) \rightarrow \sum_{s \in Sc} \A(Eu, \thg)\rlap{ ;}
\end{equation}
and, by the Yoneda lemma, to specify this is equally to specify a
function $Sf \colon Sd \rightarrow Sc$, together with a family of maps
$(E(Sf(t)) \rightarrow E(t) \in \A)_{t \in Sd}$.

Functoriality in $f$ of the maps~\eqref{eq:42} tell us that these data
determine, firstly, a presheaf $S \in \psh$; and secondly, a functor
$E \colon \el S \rightarrow \psh$. Here $\el S$ is the
\emph{category of elements} of the presheaf $S$, whose object-set is
given by $\sum_{c \in \C} Sc$ and whose maps from $s \in Xc$ to
$t \in Xd$ are maps $f \in \C(c,d)$ with $s = (Sf)(t)$.
 
In fact, giving $S$ and $E$ of this form is equivalent to giving a
small pointwise familial functor $\A \rightarrow \psh$; this can be
made precise by constructing an equivalence between $\fampt(\A, \psh)$
and the following category, which is the ``$\C\text-\cat{Fam}(\A)$''
of~\cite[Definition~2.10]{Weber2001Symmetric}.

\begin{Defn}
  \label{def:14}
  For any $\A$ and small $\C$, the category $\el_\C /\!\!/\, \A$ has
  as objects, pairs of a presheaf $S \in \psh$ and a functor
  $E \colon \el S \rightarrow \A$, and as maps
  $(S, E) \rightarrow (T, D)$, pairs of a presheaf map
  $p \colon S \rightarrow T$ and a natural isomorphism $\varphi$ of
  the form:
  \begin{equation}\label{eq:4}
    \cd[@-1em]{
      {\el S} \ar[rr]^-{\el p} \ar[dr]_-{E} & \ltwocello{d}{\varphi} &
      {\el T}\rlap{ .} \ar[dl]^-{D} \\ &
      {\A}
    }
  \end{equation}
\end{Defn}

A little care is necessary in order to extract the pair $(S,E)$ from a
small pointwise familial $F$. Observe first that \emph{choosing} an
isomorphism~\eqref{eq:41} is equivalent to \emph{choosing} a linear
operation in each connected component of $\yoneda_c \downarrow F$.
Having done so, $Sc$ can be taken to be the set of these chosen linear
operations, and $(Et \in \A)_{t \in Sc}$ to be the family of input
arities of these operations. Henceforth, we assume that each
$F \in \fampt(\A, \psh)$ is equipped with such choices of linear
operation; for any $t \in \yoneda_c \downarrow F$, we write $\tilde t$
for the chosen linear operation which covers it. In light of the
preceding discussion, it is now natural to define:

\begin{Defn}
  \label{def:5}
  The \emph{spectrum}~\cite[Definition~3]{Diers1978Spectres} of a
  small pointwise familial $F \colon \A \rightarrow \psh$ is the
  presheaf $S_F \in \psh$ given by:
  \begin{equation*}
    S_F(c) = \{ t \in \yoneda_c \downarrow F : \tilde t = t\}
    \qquad \text{and} \qquad
    S_F(f \colon d \rightarrow c) \colon t \mapsto \widetilde{t \yoneda_f}\rlap{ .}
  \end{equation*}
  The \emph{canonical diagram} of $F$ is the functor
  $D_F \colon \el S_F \rightarrow \psh \downarrow F$ with:
  \begin{align*}
    D_F(t) & = \cd{\yoneda_c \ar[d]^-{t} \\ FA} & \text{and} \quad \quad
    D_F(f \colon \widetilde{t\yoneda_f} \rightarrow t) \ \ & = \ \ 
    \cd{
      {\yoneda_c} \ar[r]^-{\yoneda_f} \ar[d]_{\widetilde{t\yoneda_f}} &
      {\yoneda_d} \ar[d]^{t} \\
      {FA'} \ar[r]^-{Fu_f} &
      {FA}\rlap{ ,}
    }
  \end{align*}
  where the lower right map is uniquely induced by linearity of
  $\widetilde{t\yoneda_f}$. The
  \emph{exponent}~\cite[Definition~7.1]{Weber2004Generic} of $F$ is the
  functor $E_F \colon \el S_F \rightarrow \A$ obtained by composing the
  canonical diagram with the second projection
  $\pi_2 \colon \psh \downarrow F \rightarrow \A$.
\end{Defn}

As elements of $S_F(c)$ are in bijection with connected components of
$\yoneda_c \downarrow F$, the presheaf $S_F$ is equally the colimit of
$F$; smallness is just what is needed to ensure this colimit exists.
In particular, smallness is vacuous when either $\A$ is small or $\A$
has a terminal object, and in the latter case, we may take $S_F = F1$.

\begin{Prop}
  \label{prop:1}
  The assignation $F \mapsto (S_F, E_F)$ is the action on objects of
  an equivalence of categories between $\fampt(\A, \psh)$ and
  $\el_\C /\!\!/\, \A$.
\end{Prop}

\begin{proof}
  This is a generalisation of~\cite[Theorem~2.18]{Weber2001Symmetric},
  and we argue largely as there. First let
  $\alpha \colon F \Rightarrow G$ in $\fampt(\A, \psh)$. To give the
  transformation $\alpha$ is to give transformations
  $\alpha_c \colon (F\thg)c \Rightarrow (G\thg)c$ naturally in $c$;
  since by Lemma~\ref{lem:8} the functor $(F\thg)c$ is a coproduct of
  representables, giving each $\alpha_c$ is equivalent to giving the
  $G$-linear operation $\alpha_A.t \colon \yoneda_c \rightarrow GA$
  obtained by acting $\alpha$ on each chosen $F$-linear operation
  $t \colon \yoneda_c \rightarrow FA$. But $\alpha_A.t$ factorises as
  on the left in:
  \begin{equation}\label{eq:8}
    \cd[@-1em]{
      & {\yoneda_c\mathstrut} \ar[dl]_-{p(t) \,\defeq\, \widetilde{\alpha_A.t}} \ar[dr]^-{\alpha_A.t} \\
      {GA'} \ar[rr]^-{G\varphi_t} & &
      {GA}
    } \qquad \qquad \cd[@-1em]{
      {\el S_F} \ar[rr]^-{\el p} \ar[dr]_-{E_F} & \ltwocello{d}{\varphi} &
      {\el S_G} \ar[dl]^-{E_G} \\ &
      {\A}
    }
  \end{equation}
  using the chosen linear operations of $G$, and linearity of
  $\alpha_A.t$ is equivalent to each $\varphi_t$ being invertible:
  indeed, linear operations are closed under isomorphism in
  $\yoneda_c \downarrow G$, and any morphism between linear operations
  is invertible. Thus, to give the pointwise familial $\alpha$ is
  equally to specify for each chosen linear
  $t \in \yoneda_c \downarrow F$ a chosen linear
  $p(t) \in \yoneda_c \downarrow G$ together with an isomorphism
  $\varphi_t \colon E_G(p(t)) \rightarrow E_F(t)$. All this must be
  done naturally in $c$ so that to give $\alpha$ is equally to give a
  pair $(p, \varphi)$ as right above with $\varphi$ invertible.

  This defines $\fampt(\A, \psh) \rightarrow \el_\C /\!\!/\, \A$ on
  morphisms and simultaneously shows that it is fully faithful
  (functoriality follows from uniqueness in Lemma~\ref{lem:5}). It
  remains to prove essential surjectivity. Given $S \in \psh$ and
  $E \colon \el S \rightarrow \A$, define a functor
  $F \colon \A \rightarrow \psh$ by taking
  $(F\thg)c = \sum_{t \in S c} \A(Et, \thg)$ and taking
  $(F\thg)(f \colon d \rightarrow c)$ to be the unique natural
  transformation rendering commutative each diagram:
  \begin{gather*}
    \cd[@C+1.7em@-0.2em]{
      {\A(Et, \thg)} \ar[r]^-{\A(Ef, \thg)} \ar[d]_{\iota} &
      {\A(E(tf), \thg)} \ar[d]^{\iota} \\
      {\sum_{t \in S c} \A(Et, \thg)} \ar[r]^-{(F\thg)f} &
      {\sum_{u \in S d} \A(Eu, \thg)}\rlap{ .}
    }
  \end{gather*}
  By Lemma~\ref{lem:8}, $F$ is pointwise familial, and is moreover
  small since the coproduct defining $(F\thg)c$ is so; now by choosing
  the linear operations in $\yoneda_c \downarrow F$ to be those
  $\gamma_t \colon \yoneda_c \rightarrow FEt$ picking out the pairs
  $(t, 1_{Et})$, we have a bijection $S \rightarrow S_F$ sending $t$ to
  $\gamma_t$, which, since $E_F(\gamma_t) = Et$, commutes
  \emph{strictly} with the functors to $\A$.
\end{proof}

\subsection{Composition of familial functors}
\label{sec:comp-famil-funct}

In seeking a universal shapely monad among the class of familial
endofunctors, we must consider both composability and existence of a
terminal object; we start with composability. The following lemma
gives the properties of linear operations necessary to
establish our results.

\begin{Lemma}
  \label{lem:1}
  Let $F \colon \A \rightarrow \B$ and $G \colon \B \rightarrow \C$ be
  functors.
  \begin{enumerate} [(i)]
  \item \label{item:linear:composite} If $s \colon C \rightarrow GB$
    is $G$-linear and $t \colon B \rightarrow FA$ is $F$-linear, then
    the composite $Gt . s \colon C \rightarrow GB \rightarrow GFA$ is
    $GF$-linear.\vskip0.25\baselineskip
  \item The full subcategory
    $\cat{Lin}(\B \downarrow F) \subset \B \downarrow F$ on the linear
    operations is closed under pointwise colimits (ones created by the
    projection $\B \downarrow F \rightarrow \B \times \A$).
  \end{enumerate}
\end{Lemma}

\begin{proof}
  A short calculation using Lemma~\ref{lem:5}.
\end{proof}

Now, in order to show that the composite of $F\colon \A \to \B$ and
$G\colon \B \to \C$ is familial at stage $C$, we should like to take
the linear operations of the composite to be of the form $Gt.s$ for
$s$ and $t$ as in~\eqref{item:linear:composite} above. In order for
this to work, we need to be able to cover any operation $C \to GFX$ by
some operation of this form. This suggests that, if $GF$ is to be
familial at stage $C$, then the input arity of each $G$-linear
operation at stage $C$ should be a stage of familiality for $F$, as in
the following definition.

\begin{Defn}
  \label{def:18}
  Given $\A' \subset \A$ and $\B' \subset \B$ full replete
  subcategories, we say that $F \colon \A \rightarrow \B$ is
  \emph{$(\A',\B')$-familial} if it is familial at each stage
  $B \in \B'$ and each linear $t \in B \downarrow F$ has input arity
  in $\A'$. A transformation $\alpha \colon F \Rightarrow G$ between
  such functors is \emph{$(\A', \B')$-familial} if it is familial at
  every stage $B \in \B'$.
\end{Defn}

In this terminology, a familial functor $F \colon \A \rightarrow \B$
is equally $(\A, \B)$-familial, while a pointwise familial functor
$\A \rightarrow \psh$ is equally an $(\A, \yoneda \C)$-familial one. The
next result improves in very mild ways
on~\cite[p.~985]{Diers1978Spectres}
and~\cite[Corollary~5.15]{Weber2004Generic}.

\begin{Prop}
  \label{prop:10}
  If $F \colon \A \rightarrow \B$ and $G \colon \B \rightarrow \C$ are
  $(\A', \B')$- and $(\B', \C')$-familial, then their composite is
  $(\A', \C')$-familial, and has as linear operations at stage
  $C \in \C'$ precisely the composites
  $Gt.s \colon C \rightarrow GB \rightarrow GFA$ of $G$- and
  $F$-linear operations. The correspondingly familial transformations
  between these functors are likewise composable; in particular, there
  is a $2$-category $\cat{FAM}$ of categories, familial functors and
  familial transformations.
\end{Prop}

\begin{proof}
  Because all linear coverings of an operation are isomorphic and
  $\A'$ is replete, to show that $GF$ is $(\A', \C')$-familial it
  suffices to show that any $s \colon C \rightarrow GFX$ with
  $C \in \C'$ is covered in $C \downarrow GF$ by some linear operation
  with input arity in $\A'$. But we have successive factorisations
  \begin{equation*}
    \cd[@!C@C-2em@-0.5em]{
      & {C} \ar[dl]_-{\tilde s} \ar[dr]^-{s} & &
      & & {B} \ar[dl]_-{\tilde t} \ar[dr]^-{t} \\
      {GB} \ar[rr]^-{Gt} & & 
      {GFX} & &
      {FA} \ar[rr]^-{Ff} & &
      {FX}
    }
  \end{equation*}
  with $B \in \B'$ and $A \in \A'$ by applying familiality of $G$ to
  $s$ and of $F$ to $t$. By Lemma~\ref{lem:1}(i), the composite
  $G \tilde t.\tilde s$ is $GF$-linear so that
  $f \colon G\tilde t.\tilde s \rightarrow s$ is the required cover.
  That all linear operations have this shape follows by
  Lemma~\ref{lem:5}. Stability under composition follows from the
  previous points and naturality.
\end{proof}

Since we are really interested in pointwise familial functors between
presheaf categories, we should like to know that these, too, are
closed under composition. The key to showing this is the following
result.

\begin{Prop}
  \label{prop:18}
  Let $\A$ be cocomplete. For any functor $F \colon \A \rightarrow \B$
  or transformation $\alpha \colon F \Rightarrow G$, the full
  subcategory $\B' \subset \B$ whose objects are those stages
  $B \in \B$ at which $F$ (respectively $\alpha$) is familial is
  closed in $\B$ under colimits.
\end{Prop}

\begin{proof}
  Suppose given $F \colon \A \rightarrow \B$, a diagram
  $D \colon \I \rightarrow \B$ such that $F$ is familial at each
  $DI \in \B$, and a colimiting cocone
  $(p_I \colon DI \rightarrow V)_{I \in \I}$; we must show that $F$ is
  also familial at $V$. So let $t \colon V \rightarrow FA$, and
  consider the diagram of linear operations
  $D_t \colon \I \rightarrow \B \downarrow F$ defined by:
  \begin{align*}
    D_t(I) & = \cd[@-0.25em]{{DI} \ar[d]_{\widetilde{tp_I}} \\ {FA_{I}}}
    & \quad D_t(f \colon I \rightarrow J) \ \ & = \ \ 
    \cd[@-0.25em]{
      {DI} \ar[r]^-{Df} \ar[d]_{\widetilde{tp_I}} &
      {DJ} \ar[d]^{\widetilde{tp_J}} \\
      {FA_{I}} \ar[r]^-{Fu_f} &
      {FA_{J}}
    }
  \end{align*}
  where the map $u_f$ is the unique one induced by linearity of
  $\widetilde{tp_I}$. Since $\A$ is cocomplete, the diagram $D_t$
  admits a pointwise colimit $u \colon V \rightarrow FW$, which by
  Lemma~\ref{lem:1}(ii) is itself linear. There is a cocone
  $D_t \Rightarrow \Delta t$ with components
  $(p_I, q_I) \colon \widetilde{tp_I} \rightarrow t$ where the maps
  $q_I$ are, again, induced by linearity of $\widetilde{tp_I}$, and
  this now induces a map $u \rightarrow t$ in $V \downarrow F$
  providing the desired linear cover of $t$. This shows $F$ is
  familial at $V$, and also that $t \in V \downarrow F$ is linear just
  when its cocone $(p,q) \colon D_t \Rightarrow \Delta t$ is
  colimiting; this last fact entails the part of the proposition
  concerned with transformations $\alpha$.
\end{proof}

As every presheaf is a colimit of representables, we immediately
conclude from the preceding two results that:

\begin{Cor}
  \label{cor:3}
  If $\A$ is cocomplete, then each pointwise familial functor or
  transformation in $\fampt(\A, \psh)$ is familial; whence there
  is a $2$-category $\fampt$ of presheaf categories and pointwise
  familial functors and transformations.
\end{Cor}

In the next section, size considerations will force us to bound the
input arities of the pointwise familial functors we consider. As we
would still like such functors to compose, we introduce the relevant
notions and prove composability here.

\begin{Defn}
  \label{def:19}
  We write $\F \C \subset \psh$ for the full, replete subcategory of
  \emph{finitely presentable} presheaves: those expressible as a
  finite colimit of representables. A pointwise familial functor
  $\psh \rightarrow \psh[\D]$ is called \emph{finitary} if it is
  $(\F\C, \yoneda\D)$-familial.
\end{Defn}

The modifier ``finitary'' typically refers to a functor which
preserves filtered colimits; that our usage agrees with this follows
from Lemma~\ref{lem:8} and the fact that a representable
$\psh(A, \thg) \colon \psh \rightarrow \cat{Set}$ is finitary just
when $A$ is in $\F\C$.

To see that finitary pointwise familial functors and transformations
compose, we appeal to Proposition~\ref{prop:10} and the following
result:

\begin{Prop}
  \label{prop:11}
  The pointwise familial $F \colon \psh \rightarrow \psh[\D]$ is
  finitary if and only if it is $(\F\C, \F\D)$-familial; whence there
  is a $2$-category $\fampt^\omega$ of presheaf categories, finitary
  pointwise familial functors and pointwise familial transformations.
\end{Prop}

\begin{proof}
  For the non-trivial direction, let $t \colon B \rightarrow FA$ with
  $B \in \F \D$. On expressing $B$ as a finite colimit of
  representables, the proof of Proposition~\ref{prop:18} yields a
  cover of $t$ by a linear operation $u \colon B \rightarrow FA$
  obtained as a finite colimit in $\B \downarrow F$ of linear
  operations of the form $\yoneda_{d_I} \rightarrow FA_I$. By
  assumption, each $A_I$ is in $\F \C$, whence $A = \colim_{I} A_I$ is
  too.
\end{proof}

\subsection{Universal familial endofunctors}
\label{sec:univ-famil-endof}

We now have all the ingredients we require for our first attempt at
constructing a universal shapely monad $\mathsf U$ on $\psh$. As
anticipated in the previous section, a naive attempt to construct it
as a terminal object in the monoidal category
$\fampt(\psh, \psh)$ fails for size reasons.

\begin{Prop}
  \label{prop:12}
  If $\C \neq 0$ and the category $\A$ has a proper class of
  non-isomorphic objects, then $\fampt(\A, \psh)$ has no terminal
  object; in particular, if $\C \neq 0$ then $\fampt(\psh, \psh)$ has
  no terminal object.
\end{Prop}

\begin{proof}
  By Proposition~\ref{prop:1}, it suffices to show that
  $\el_\C /\!\!/\, \A$ has no terminal object. Suppose that $(S, E)$
  were terminal; fixing some $c \in \C$, we would then have for each
  $A \in \A$ a unique map
  \begin{equation*}
    \cd[@-1em]{
      {\el \yoneda_c} \ar[rr]^-{\el t_A} \ar[dr]_-{\Delta A} & \ltwocello{d}{\varphi_A} &
      {\el S} \ar[dl]^-{E} \\ &
      {\A}
    }
  \end{equation*}
  where $\Delta A$ is the constant functor at $A$. Note that
  $t_A \in S c$ satisfies $Et_A \cong A$; since there are a proper
  class of non-isomorphic $A$'s, there must be a proper class of
  distinct $t_A$'s, contradicting the fact that $S c$ is a set.
\end{proof}

What permits the above negative argument is the fact that a pointwise
familial functor may have linear operations of arbitrarily large input
arity; this suggests restricting attention to the \emph{finitary}
pointwise familial functors whose linear input arities lie in the
essentially small\footnote{A category is \emph{essentially small} if
  it is equivalent to a small category.} $\F\C$. We first note that:

\begin{Lemma}
  \label{lem:2}
  Precomposition with the inclusion $J \colon \F \C \rightarrow \psh$
  induces an equivalence between the categories
  $\fampt^\omega(\psh, \psh[\D])$ and $\fampt(\F \C, \psh[\D])$.
\end{Lemma}

\begin{proof}
  Precomposing the equivalence
  $\fampt(\F \C, \psh[\D]) \rightarrow {\el_\D /\!\!/\, \F \C}$ of
  Proposition~\ref{prop:1} by
  $(\thg) \circ J \colon \fampt^\omega(\psh, \psh[\D]) \rightarrow
  \fampt(\F \C, \psh[\D])$ evidently yields another equivalence;
  whence, by two-out-of-three, $(\thg) \circ J$ is an equivalence.
\end{proof}

Unfortunately, even with the finitariness restriction we are still
unable to construct a strictly terminal familial endofunctor:

\begin{Prop}
  \label{prop:2}
  If $\C \neq 0$ and $\A$ is essentially small, then
  $\fampt(\A, \P\C)$ always has a weakly terminal object, but has a
  terminal object if and only if $\A$ has no non-identity
  automorphisms; consequently, if $\C \neq 0$, then
  $\fampt^\omega(\psh, \psh)$ has a weakly terminal object, but not a
  terminal object.
\end{Prop}

\begin{proof}
  By Proposition~\ref{prop:1} we may prove the stated properties for
  the equivalent category $\el_\C /\!\!/\, \A$; but as $\A \simeq \A'$
  with $\A'$ small and now
  $\el_\C /\!\!/\, \A \simeq \el_\C /\!\!/\, \A'$, we may assume
  without loss of generality that $\A$ is itself small. We construct a
  weakly terminal $(S,E)$ in $\el_\C /\!\!/\, \A$ as follows. The
  presheaf $S \in \P\C$ has:
  \begin{equation*}
    Sc = \{\,\text{functors } F \colon \C/c \rightarrow \A \,\} \qquad \text{and} \qquad S(f \colon d \rightarrow c) \colon F \mapsto F(f \circ \thg)\rlap{ ,}
  \end{equation*}
  while $E \colon \el S \rightarrow \A$ is given by
  $E(F \colon \C / c \rightarrow \A) = F(1_c)$ on objects, and by:
  \begin{align*}
    f \colon F(f\circ \thg) \rightarrow F \qquad &\mapsto \qquad F(f \colon f \rightarrow 1_c) \colon Ff \rightarrow F1_c
  \end{align*}
  on morphisms. To see weak terminality of $(S,E)$, consider some other
  $(T, D)$ in $\el_\C /\!\!/\, \A$. We define
  $p \colon T \rightarrow S$ in $\psh$ by sending $t \in Tc$ to the
  element
  \begin{equation*}
    p(t) \colon \C / c = \el \yoneda_c \xrightarrow{\el t} \el T \xrightarrow{D} \A
  \end{equation*}
  of $Sc$. Naturality of $p$ in $c$ follows because
  $\el (t) \rond \el(\yoneda_f) = \el (t \rond \yoneda_f) = \el(t \cdot
  f)$ for all $t \in T(c')$ and $f\colon c \to c'$ in $\C$.
  Furthermore, from the equality $p(t)(1_c) = Dt$, we deduce
  $E \circ \el p = D \colon \el T \rightarrow \A$ and so
  $(p,1_D) \colon (T,D) \rightarrow (S,E)$ in $\el_\C /\!\!/\, \A$.

  Now let $\A$ have no non-identity automorphisms; replacing it by its
  (equivalent) skeleton, we may assume that in fact it has no
  non-identity isomorphisms, and so that each map~\eqref{eq:4} of
  $\el_\C /\!\!/\, \A$ has $\varphi$ an identity. In this case, we
  claim the weakly terminal $(S,E)$ given above is terminal. Indeed, if
  $(q, 1_D) \colon (T,D) \rightarrow (S,E)$ is any map in
  $\el_\C /\!\!/\, \A$, then for each $t \in Tc$, the functor
  $q(t) \colon \C / c \rightarrow \A$ satisfies
  $q(t)(1_c) = Dt = p(t)(1_c)$; but then
  $q(t)(h) = q(th)(1_d) = p(th)(1_d) = p(t)(h)$ for all
  $h \colon d \rightarrow c$, whence $q = p$ as required.

  Next let $\A$ admit the non-identity automorphism $a \in \A(A,A)$,
  and assume that there is a terminal object $(T,D)$ in
  $\el_\C /\!\!/\, \A$; we derive a contradiction. By terminality of
  $(T,D)$, there is for any $c \in \C$ a unique pair as on the left in
  \begin{equation*}
    \cd[@-1em]{
      {\el \yoneda_c} \ar[rr]^-{\el p} \ar[dr]_-{\Delta A} & \ltwocello{d}{\varphi} &
      {\el T} \ar[dl]^-{D} \\ &
      {\A}
    } \qquad \qquad 
    \cd[@-1em]{
      {\el \yoneda_c} \ar[rr]^-{\el p} \ar[dr]_-{\Delta A} & \ltwocello{d}{\varphi.\Delta a} &
      {\el T} \ar[dl]^-{D} \\ &
      {\A}
    }
  \end{equation*}
  where here $\Delta A$ is the constant functor at $A$. But now the
  triangle on the right also describes a morphism
  $(\yoneda_c, \Delta A) \rightarrow (T,D)$; so we must have
  $\varphi.\Delta a = \varphi$ and so, by invertibility of $\varphi$,
  that $\Delta a = \id_{\Delta_A}$, contradicting $a \neq \id_A$.

  For the final claim, note that we have
  $\fampt^\omega(\psh, \psh) \simeq \fampt(\F\C, \psh)$ by
  Lemma~\ref{lem:2}; now if $\C \neq 0$, then the essentially small
  $\F \C$ certainly contains non-identity automorphisms---for instance,
  the switch map
  $\yoneda_c + \yoneda_c \rightarrow \yoneda_c + \yoneda_c$---and so
  $\fampt^\omega(\psh, \psh)$ has a weakly terminal object, but no
  terminal object.
\end{proof}

\section{Analytic functors and shapeliness}
\label{sec:analyt-funct-shap}

The underlying reason that there is no terminal object among finitary
familial endofunctors of a presheaf category is that linear operations
cannot be fixed by automorphisms of their input arities; this means
that such automorphisms may be propagated up to the level of familial
functors, so obstructing the existence of a terminal object. The next
step in our pursuit of a universal shapely monad will attempt to
resolve this problem by introducing \emph{analytic} functors, whose
generating operations can be fixed by input arity automorphisms.

\subsection{Generic operations and analytic functors}
\label{sec:generic-operations}

The fundamental step in moving from familial to analytic functors is
to generalise from linear to \emph{generic} operations. In what
follows, we write $\aut_X$ for the automorphism group of any object
$X \in \C$.

\begin{Defn}
  \label{def:4}
  An object $X \in \C$ is \emph{Galois} if for each $Y \in \C$, the
  composition action makes $\C(X,Y)$ into a connected $\aut_X$-set. An
  operation $t \colon B \rightarrow FA$ of a functor
  $F \colon \A \rightarrow \B$ is \emph{generic} if it is Galois in
  its connected component of $B \downarrow F$.
\end{Defn}

(Our nomenclature draws on one of the basic examples of a Galois
object: if $k \subset K$ is a Galois field extension, and $\A$ is the
category of intermediate field extensions, then $K$ is Galois in
$\A^\mathrm{op}$.)

An object $X$ is Galois when it admits a map to every other object
(\emph{weak initiality}) and, for any pair of maps
$f, f' \colon X \rightrightarrows Y$, there is an automorphism
$\sigma \in \aut_X$ with $f' = f \sigma$ (\emph{transitivity}); thus,
Galois objects are initial ``up to a group of automorphisms''. In
these terms, a generic $F$-operation can be understood as one which,
like a linear operation, consumes all of its input arity, but which
may now be invariant under certain automorphisms of that arity.

The next result identifies our generic operations with those
of~\cite[Definition~5.2]{Weber2004Generic}, which when
$\A = \B = \cat{Set}$ and $B = 1$ are equally those
of~\cite{Joyal1986Foncteurs}.

\begin{Lemma}
  \label{lem:12}
  An operation $t \colon B \rightarrow FA$ is generic if and only if
  for every square of the following form there exists some
  $\ell \colon A \rightarrow Y$ with $F\ell.t = u$ and $h\ell=k$:
  \begin{equation}\label{eq:25}
    \cd[@C-0.5em]{
      {B} \ar[r]^-{u} \ar[d]_{t} &
      {FY} \ar[d]^{Fh} \\
      {FA} \ar[r]^-{Fk} \ar@{.>}[ur]|-{F\ell}
      &
      {FZ}\rlap{ .}
    }
  \end{equation}
\end{Lemma}

\begin{proof}
  We claim that $X$ is Galois if and only if it is weakly initial and
  every diagram as in the solid part of
  \begin{equation}\label{eq:44}
    \cd{
      & Y \ar[d]^-{g} \\
      X \ar[r]^-{f} \ar@{-->}[ur]^-{h} & Z
    }
  \end{equation}
  can be completed to a commuting diagram as displayed. Indeed, if $X$
  is Galois then in the situation of~\eqref{eq:44}, weak initiality
  gives a map $k \colon X \rightarrow Y$, and transitivity gives some
  $\sigma \in \Aut_X$ such that $f = gk\sigma$, so that we may take
  $h = k\sigma$. Conversely, if $X$ satisfies the displayed condition,
  then taking $Y=Z=X$ and $f = 1_X$ shows that each $g \in \C(X,X)$ is
  split epimorphic; whence each $g \in \C(X,X)$ is invertible;
  whereupon taking $Y = X$ in~\eqref{eq:44} gives transitivity.

  Now the condition on $t$ above says that any cospan
  $t \rightarrow v \leftarrow u$ in $B \downarrow F$ can be completed
  to a commuting triangle, which thus says that $t$ is Galois in its
  connected component, as desired.
\end{proof}

\begin{Cor}
  \label{cor:2}
  Any map $h \colon u \rightarrow t$ in $B \downarrow F$ with generic
  codomain is a split epimorphism; in particular, any map between
  generic operations is an isomorphism.
\end{Cor}

\begin{proof}
  Take $k = 1_A$ in~\eqref{eq:25}.
\end{proof}

Replacing linear operations with generic ones in the definition of
familial functor yields the notion of \emph{analytic} functor.

\begin{Defn}
  \label{def:10}
  A functor $F \colon \A \rightarrow \B$ is \emph{analytic at stage
    $B \in \B$} if each operation in $B \downarrow F$ is covered by a
  generic one; a transformation $\alpha \colon F \Rightarrow G$ is
  \emph{analytic at stage $B$} if $F$ and $G$ are so, and the induced
  functor $B \downarrow F \rightarrow B \downarrow G$ preserves
  generic operations. We write simply \emph{analytic} to mean
  ``analytic at every stage''.
\end{Defn}

\begin{Rk}
  \label{rk:6}
  Analytic endofunctors of $\cat{Set}$ and \emph{weakly cartesian}
  transformations were introduced by Joyal
  in~\cite{Joyal1986Foncteurs}; by \cite[Theorems~10.10 \&
  10.11]{Weber2004Generic}, they are precisely the filtered-colimit
  preserving analytic endofunctors and transformations of $\cat{Set}$
  in our sense. However, as noted in the introduction, there are other
  possible ways to extend Joyal's notion of analyticity to general
  presheaf categories; two which exist in the literature are the
  \emph{quotient containers} of~\cite{Abbott2004Constructing}, and the
  \emph{generalised species} of~\cite{Fiore2008The-cartesian} (also
  studied in~\cite{Fiore2014Analytic}). Neither of these
  generalisations are adequate for our purposes, since neither have
  familial functors as a special case.
\end{Rk}

Just as familial functors are obtained from coproducts of
representables, so analytic functors arise from coproducts of
\emph{near-representables} in the sense of~\cite{Tambara2015Finite}:

\begin{Defn}
  \label{def:9}
  Let $A \in \A$ and $G \leqslant \aut_A$. A \emph{coinvariant} for
  $G$ is a joint coequaliser $q \colon A \twoheadrightarrow A_{/G}$
  for the set of morphisms
  $\{\sigma \colon A \rightarrow A \mid \sigma \in G\}$; dually, an
  \emph{invariant} for $G$ is a joint equaliser
  $\iota \colon A^{\setminus G} \rightarrowtail A$ for the maps in
  $G$. A functor $F \colon \A \rightarrow \cat{Set}$ is
  \emph{near-representable} if $F \cong \A(A, \thg)_{/G}$ for some
  $A \in \A$ and $G \leqslant \aut_A = \aut_{\A(A, \thg)}$.
\end{Defn}

\begin{Lemma}
  \label{lem:13}
  A functor $F \colon \A \rightarrow \B$ is analytic at stage
  $B \in \B$ if and only if the functor
  $\B(B, F\thg) \colon \A \rightarrow \cat{Set}$ is a (possibly large)
  coproduct of near-representables.
\end{Lemma}

\begin{proof}
  This will follow as in Lemma~\ref{lem:8} once we have proved that:
  $F \in [\A, \cat{Set}]$ is near-representable just when $\el F$
  contains a Galois object. For any $(x, A)$ in $\el F$, let
  $G = \aut_{(x,A)} \leqslant \aut_A$; now
  $x \colon \yoneda_A \rightarrow F$ coequalises
  $\yoneda_\sigma \colon \yoneda_A \rightarrow \yoneda_A$ for each
  $\sigma \in G$, and so descends to a map
  $\bar x \colon {\yoneda_A}_{/G} \rightarrow F$. It suffices to show
  that $\bar x$ is an isomorphism just when $(x,A)$ is Galois.
  Surjectivity of $\bar x$ corresponds to weak initiality of $(x,A)$;
  injectivity requires that, for any
  $f, f' \colon A \rightrightarrows B$ with
  $\bar x(f) = \bar x(f') = y$, we have $f' = f\sigma$ for some
  $\sigma \in G$, or in other words, that for any
  $f, f' \colon (x,A) \rightrightarrows (y,B)$ in $\el F$, there is
  some $\sigma \colon (x,A) \rightarrow (x,A)$ in $G$ with
  $f' = f\sigma$: which is transitivity of $(x,A)$.
\end{proof}

\subsection{Pointwise analyticity}
\label{sec:pointw-analyt}

As before, when we consider endofunctors of presheaf categories, or
more generally functors \emph{into} a presheaf category, the most
appropriate kind of analyticity is pointwise:

\begin{Defn}
  \label{def:17}
  A functor $F \colon \A \rightarrow \psh$ or transformation
  {${\alpha \colon F \Rightarrow G \colon \A \rightarrow \psh}$} is
  \emph{pointwise analytic} if it is analytic at all representable
  stages; $F$ is called \emph{small} if $\yoneda_c \downarrow F$ has a
  mere set of connected components for each $c \in \C$. We write
  $\anpt(\A, \psh)$ for the category of small pointwise analytic
  functors and pointwise analytic transformations.
\end{Defn}

In particular, by Lemma~\ref{lem:13}, a functor
$F \colon \A \rightarrow \psh$ is small pointwise analytic just when
each $(F\thg)c$ is a small coproduct of near-representables; so, for
example, comparing with the formula of Proposition~\ref{prop:22}, we
find---as promised above---that the ``free polycategory'' endofunctor
on the category of \emph{symmetric} polygraphs is pointwise analytic,
though it is not pointwise familial.

As in the familial case, a small pointwise analytic $F$ is determined
by the near-representable summands of each $(F\thg)c$ and how these
transform under maps $(F\thg)f \colon (F\thg)d \rightarrow (F\thg)c$.
We wish to give a representation of these data analogous to
Definition~\ref{def:5}; the new aspect is that, in encoding a
near-representable summand $\A(A, \thg)_{/G}$, we must record not just
the arity $A$ but also the group $G$ of automorphisms which fix it. We do this
using the notion of \emph{orbit category}.

\begin{Defn}
  \label{def:6}
  The \emph{orbit category} $\O(\A)$ of a category $\A$ has as
  objects, pairs $(A,G)$ where $A \in \A$ and $G \leqslant \aut_A$,
  and as morphisms $[f] \colon (A, G) \rightarrow (B, H)$, equivalence
  classes of maps $f \colon A \rightarrow B$ in $\A$ with the property
  that
  \begin{equation}\label{eq:3}
    \text{for all $\tau \in H$, there exists $\sigma \in G$ with
      $\tau f = f \sigma$,}
  \end{equation}
  where $[f] = [f']$ when there exists $\sigma \in G$ with
  $f' = f\sigma$. We write $J \colon \A \rightarrow \O(\A)$ for the
  full embedding sending $A$ to $(A,1)$.
\end{Defn}

Intuitively, we regard the generating operations of
$F \in \anpt(\A, \psh)$ as having input arities drawn not from $\A$
but from $\O(\A)$; we will make this precise by equating such functors
$F$ with pointwise familial ones $F' \in \fampt(\O(\A), \psh)$. First
we describe the passage between functors with domains $\A$ and
$\O(\A)$.

\begin{Prop}
  \label{prop:3}
  $\O(\A)$ admits all group invariants, and for any category $\B$
  admitting group invariants, composition with $J$ induces an
  equivalence
  \begin{equation}\label{eq:7}
    \cat{INVAR}(\O(\A), \B) \xrightarrow{(\thg) \circ J} \cat{CAT}(\A, \B)
  \end{equation}
  with domain the category of invariant-preserving functors and
  transformations.
\end{Prop}

\begin{proof}
  The statement says that $\O(\A)$ is the free completion of $\A$
  under group invariants, and by~\cite[Theorem~5.35]{Kelly1982Basic},
  this completion may be found as the full subcategory of
  $[\A, \cat{Set}]^\mathrm{op}$ obtained by closing the representables
  under group invariants. So it suffices to identify $\O(\A)$ with
  this full subcategory. Direct calculation using the Yoneda lemma
  shows that maps in $\O(\A)$ from $(A,G)$ to $(B,H)$ are in bijection
  with maps $\A(B, \thg)_{/H} \rightarrow \A(A, \thg)_{/G}$ in
  $[\A, \cat{Set}]$; so there is a fully faithful
  $K \colon \O(\A) \rightarrow [\A, \cat{Set}]^\mathrm{op}$ with
  $K(A,G) = \A(A, \thg)_{/G}$. By definition, each $\A(A, \thg)_{/G}$
  in the image of $K$ lies in the closure of the representables in
  $[\A, \cat{Set}]^\mathrm{op}$ under group invariants, and so it
  suffices to show that this subcategory in fact has all group
  invariants---which is~\cite[Proposition~2.2]{Tambara2015Finite}.
\end{proof}

Explicitly, if $\B$ admits group invariants and
$H \colon \A \rightarrow \B$, then the invariant-preserving extension
$H' \colon \O(\A) \rightarrow \B$ is defined by
$H'(A,G) = HA^{\setminus HG}$, where here
$HG = \{H \sigma : \sigma \in G\} \leqslant \aut_{HA}$. In particular,
if $F \colon \A \rightarrow \B$ is any functor between categories,
then applying this construction to $JF \colon \A \rightarrow \O(\B)$
yields an invariant-preserving
$\O(F) \colon \O(\A) \rightarrow \O(\B)$ given by
$\O(F)(A,G) = (FA,FG)$.

We will now show that, when $\B = \psh$, the equivalence~\eqref{eq:7}
restricts back to one between pointwise analytic functors out of $\A$
and pointwise familial ones out of $\O(\A)$. However, under this
equivalence, pointwise analytic transformations correspond not to
familial ones but to \emph{near-familial} ones in the following sense:

\begin{Defn}
  \label{def:8}
  A morphism $[f] \colon (A,G) \rightarrow (B,H)$ in $\O(\A)$ is
  called \emph{vertical} if the underlying map
  $f \colon A \rightarrow B$ is invertible in $\A$. For any
  $F' \colon \O(\A) \rightarrow \B$, an operation $t$ in
  $B \downarrow F'$ is called \emph{near-linear} if it is covered by a
  linear operation via a map which is vertical in $\O(\A)$. If
  $F, G \colon \O(\A) \rightarrow \B$ are familial at stage $B$, then
  a transformation $\alpha \colon F \Rightarrow G$ is
  \emph{near-familial} at stage $B$ if it preserves near-linear
  operations.
\end{Defn}

We now give our equivalence result, after a preparatory lemma:
in the statement of the lemma, we call an object of an orbit category
$\O(\A)$ \emph{near-initial} if it admits a vertical map from an
initial object.

\begin{Lemma}
  \label{lem:15}
  \begin{enumerate}[(i)]
  \item $\A$ has a Galois object if and only if $\O(\A)$ has an
    initial one. $F \colon \A \rightarrow \B$ preserves Galois objects
    if and only if $\O(F) \colon \O(\A) \rightarrow \O(\B)$ preserves
    near-initial objects.\vskip0.25\baselineskip
  \item Let $\B$ admit group invariants and let
    $F \colon \A \rightarrow \B$ have invariant-preserving extension
    $F' \colon \O(\A) \rightarrow \B$. We have
    $B \downarrow F' \cong \O(B \downarrow F)$ naturally in $F$.
  \end{enumerate}
\end{Lemma}

\begin{proof}
  For (i), $A \in \A$ is Galois just when each $\A(A,B)$ is a
  connected $\aut_A$-set, which is equivalent to the existence of a
  unique $[u_B] \colon (A, \aut_A) \rightarrow (B,1)$ in $\O(\A)$ for
  each $B \in \A$. Now, $[u_B]$ factors through each
  $[1] \colon (B, H) \rightarrowtail (B,1)$, for, by definition of a
  Galois object, any $h \in H$ yields an automorphism $g \in \aut_A$
  such that $u_Bg = hu_B$. So $A$ is Galois just when $(A, \aut_A)$ is
  initial in $\O(\A)$. The second claim is immediate on observing that
  $(A, G) \in \O(\A)$ is near-initial if and only if
  $(A, \aut_A) \in \O(\A)$ is initial, if and only if $A \in \A$ is
  Galois.

  For (ii), an object of $B \downarrow F'$ comprises
  $(A,G) \in \O(\A)$ and $u \colon B \rightarrow FA^{\setminus FG}$ in
  $\B$. Now, to give $u$ is equally to give a map
  $t \colon B \rightarrow FA$ satisfying $F\sigma . t = t$ for all
  $\sigma \in G$. This condition says that each $\sigma \in G$ lies in
  $\aut_t \leqslant \aut_A$, and so an object of $B \downarrow F'$ is
  equally a pair $(t \in \B \downarrow F,\, G \leqslant \aut_t)$.
  Arguing similarly on morphisms, we conclude that
  $B \downarrow F' \cong \O(B \downarrow F)$; naturality in $F$ is
  straightforward.
\end{proof}

\begin{Prop}
  \label{prop:7}
  Let $\B$ admit group invariants. Under the equivalence~\eqref{eq:7},
  functors and transformations $\A \rightarrow \B$ which are analytic
  at stage $B$ correspond to functors and transformations
  $\O(\A) \rightarrow \B$ which are familial, respectively
  near-familial at stage $B$. When $\B = \psh$, the
  equivalence~\eqref{eq:7} restricts to one
  \begin{equation}
    \label{eq:9}
    \nfampt(\O(\A), \psh) \xrightarrow{(\thg) \circ J} \anpt(\A, \psh) 
  \end{equation}
  with as domain the category of small pointwise familial functors and
  pointwise near-familial transformations $\O(\A) \rightarrow \psh$.
\end{Prop}

\begin{proof}
  Let $F \colon \A \rightarrow \B$ have invariant-preserving extension
  $F' \colon \O(\A) \rightarrow \B$. Because any object $(A,G)$ of
  $\O(\A)$ admits a morphism to $(A,1)$, and because all morphisms of
  $\O(\A)$ have an underlying morphism in $\A$, $\O(\thg)$ preserves
  connected components. So, by Lemma~\ref{lem:15}, each
  $B \downarrow F$ is a coproduct of categories with Galois objects
  just when each $\O(B \downarrow F) \cong B \downarrow F'$ is a
  coproduct of categories with initial objects. Moreover, if
  $\alpha \colon F \Rightarrow G$ is a transformation between functors
  analytic at stage $B$, with invariant-preserving extension
  $\alpha' \colon F' \Rightarrow G'$, then by Lemma~\ref{lem:15}, each
  functor
  $B \downarrow \alpha \colon B \downarrow F \rightarrow B \downarrow
  G$ preserves Galois objects just when each
  $\O(B \downarrow \alpha) \cong B \downarrow \alpha'$ preserves
  near-initial objects. This proves the first claim.

  Now suppose that $\B = \psh$. It is immediate that smallness is
  preserved under the preceding equivalences, and so the only
  additional point to verify is that
  $\nfampt(\O(\A), \psh) \subset \cat{INVAR}(\O(\A), \psh)$. But if
  $G \colon \O(\A) \rightarrow \psh$ is pointwise familial, then each
  $(G\thg)c \colon \O(\A) \rightarrow \cat{Set}$, being a coproduct of
  representables, preserves connected limits and in particular group
  invariants; whence $G$ preserves group invariants, as limits in
  $\psh$ are pointwise.
\end{proof}

Using this result, we may now give the promised analytic analogue of
Definition~\ref{def:5}, describing each small pointwise analytic
$F \colon \A \rightarrow \psh$ in terms of the near-representable
summands of each $(F\thg)c$.

\begin{Defn}
  \label{def:11}
  Let $F \colon \A \rightarrow \psh$ be small pointwise analytic. The
  \emph{spectrum} $S_{F} \in \psh$ and \emph{exponent}
  $E_F \colon \el S_F \rightarrow \O(\A)$ of $F$ are the spectrum and
  exponent of the small pointwise familial
  $F' \colon \O(\A) \rightarrow \psh$ corresponding to $F$
  under~\eqref{eq:9}.
\end{Defn}

\begin{Rk}
  \label{rk:1}
  Let us unpack this definition. Given $F \in \anpt(\A, \psh)$, we
  choose like before a generic operation in each connected component
  of $\yoneda_c \downarrow F$, and write $\tilde t$ for the chosen
  generic cover of $t \in \yoneda_c \downarrow F$. The spectrum of $F$
  is now exactly as in Definition~\ref{def:5}, while the exponent
  $E_F \colon \el S_F \rightarrow \O(\A)$ is given on objects by
  $E_F(t \colon \yoneda_c \rightarrow TA) = (A, \aut_t)$, where
  $\aut_t$ is the automorphism group of
  $t \in \yoneda_c \downarrow T$, or equally the set of all
  $\sigma \in \aut_A$ such that $(T\sigma)(t) = t$. To define $E_F$ on
  a map $f \colon \widetilde{t\yoneda_f} \rightarrow t$ of $\el S_F$,
  we form the square
  \begin{align*}
    \cd{
      {\yoneda_c} \ar[r]^-{\yoneda_f} \ar[d]_{t' =\widetilde{t\yoneda_f}} &
      {\yoneda_d} \ar[d]^{t} \\
      {FA'} \ar[r]^-{Fu_f} &
      {FA}
    }
  \end{align*}
  whose lower edge is \emph{any} map induced by weak initiality of
  $\widetilde{t\yoneda_f}$ in $\yoneda_c \downarrow F$, and take
  $E_F(f) = [u_f] \colon (A', \aut_{t'}) \rightarrow (A, \aut_t)$. Note
  that the mapping $f \mapsto u_f$ is only functorial ``up to
  automorphism groups'', so that $E_F$ may \emph{not} exist as a
  functor $\el S_F \rightarrow \A$.
\end{Rk}

Just as in the familial case, a small pointwise analytic
$\A \rightarrow \psh$ can be recovered from its spectrum and exponent.
We express this in terms of an equivalence between $\anpt(\A, \psh)$
and the following category:

\begin{Defn}
  \label{def:16}
  For any $\A$ and small $\C$, the category
  $\el_\C /\!\!/_{\!v}\,\, \O(\A)$ has as objects, pairs
  ($S \in \psh$, $E \colon \el S \rightarrow \O(\A)$), and as maps
  $(S, E) \rightarrow (T, D)$, pairs of a presheaf map
  $p \colon S \rightarrow T$ and a pointwise vertical transformation
  $\varphi$ of the form:
  \begin{equation}\label{eq:10}
    \cd[@-1em]{
      {\el S} \ar[rr]^-{\el p} \ar[dr]_-{E} & \ltwocello{d}{\varphi} &
      {\el T} \ar[dl]^-{D} \\ &
      {\O(\A)}\rlap{ .}
    }
  \end{equation}
\end{Defn}

\begin{Prop}
  \label{prop:4}
  The assignation $F \mapsto (S_F, E_F)$ is the action on objects of
  an equivalence of categories between $\anpt(\A, \psh)$ and
  $\el_\C /\!\!/_{\!v}\,\, \O(\A)$.
\end{Prop}

\begin{proof}
  By Proposition~\ref{prop:7}, it suffices to show that
  $F' \mapsto (S_{F'}, E_{F'})$ underlies an equivalence of categories
  $\nfampt(\O(\A), \psh) \rightarrow \el_\C /\!\!/_{\!v}\,\, \O(\A)$.
  This is almost exactly as in Proposition~\ref{prop:1}, with the only
  difference arising on morphisms. By definition, a transformation
  $\alpha' \colon F' \Rightarrow G'$ between pointwise familial
  functors $F', G' \colon \O(\A) \rightarrow \psh$ is near-familial
  just when each triangle as to the left of~\eqref{eq:8} has
  $\varphi_t$ \emph{vertical}, rather than invertible: this accounts
  for the differing $2$-cell data between~\eqref{eq:4}
  and~\eqref{eq:10}.
\end{proof}

\subsection{Universal analytic endofunctors}
\label{sec:univ-analyt-endof}

Now in seeking a universal shapely monad among analytic endofunctors,
we must as before consider both composability and existence of a
terminal object. This time we deal with terminality first. As in
Proposition~\ref{prop:12}, there is a size obstruction to constructing
a terminal object of the category $\anpt(\psh, \psh)$, and so we must
impose size restrictions.

\begin{Defn}
  \label{def:12}
  If $\A' \subset \A$ and $\B' \subset \B$ are full replete
  subcategories, we say that $F \colon \A \rightarrow \B$ is
  \emph{$(\A',\B')$-analytic} it is analytic at each $B \in \B'$, and
  each generic $t \in B \downarrow F$ has input arity in $\A'$. A
  transformation $\alpha \colon F \Rightarrow G$ between such functors
  is \emph{$(\A', \B')$-analytic} if it is analytic at every stage
  $B \in \B'$.
\end{Defn}

\begin{Defn}
  \label{def:13}
  A pointwise analytic $F \colon \psh \rightarrow \psh[\D]$ is called
  \emph{finitary} if it is $(\F\C, \yoneda\D)$-analytic. We write
  $\anpt^\omega(\psh, \psh[\D])$ for the category of finitary
  pointwise analytic functors and pointwise analytic transformations.
\end{Defn}


This restriction is in fact enough: $\anpt^\omega(\psh, \psh)$---and
more generally, $\anpt^\omega(\psh[\D], \psh)$---\emph{does} have a
terminal object. To see this, we first argue as in Lemma~\ref{lem:2}
to establish an equivalence between $\anpt^\omega(\psh[\D], \psh)$ and
$\anpt(\F \D, \psh)$; the claim will now follow once we show more
generally that $\anpt(\A, \psh)$ has a terminal object whenever $\A$
is essentially small. The key to proving this is the following lemma;
in it we write $[\I, \O(\A)]_v$ for the category of functors $\I
\rightarrow \O(\A)$ and pointwise vertical transformations, where $\I$
is any small category.

\begin{Lemma}
  \label{lem:14}
  Each connected component of $[\I, \O(\A)]_v$ has an initial object.
\end{Lemma}

\begin{proof}
  Let $(T,G) \colon \I \rightarrow \O(\A)$ be given on objects by
  $I \mapsto (TI, G_I)$ and on morphisms by
  $f \mapsto [Tf] \colon (TI,G_I) \rightarrow (TJ,G_J)$. Note that the
  family of subgroups $(G_I \leqslant \aut_{TI})_{I \in \I}$ satisfies
  the condition that
  \begin{equation}\tag{$\star$}\label{eq:11}
    \text{for all $f \colon I \rightarrow J$ and $\sigma \in G_J$, there exists $\tau \in G_I$ with $Tf \circ \tau = \sigma \circ Tf$.}
  \end{equation}
  Call a family of subgroups $H = (H_I \leqslant \aut_{TI})_{I \in \I}$
  \emph{suitable} if it satisfies~\eqref{eq:11} with $H_I$ and $H_J$ in
  place of $G_I$ and $G_J$. We claim that, if $H^x$ is a suitable
  family of subgroups for each $x \in X$, then the family of subgroups
  $\bigvee_x H^x = (\bigvee_{x} H^x_I)_{I \in \I}$ is again suitable
  (here the join $\bigvee_{x} H^x_I$ is taken in the lattice of
  subgroups of $\aut_{TI}$). Indeed, if $f \colon I \rightarrow J$ in
  $\I$ and $\sigma \in \bigvee_{x} H^x_J$, then
  $\sigma = \sigma_1 \cdots \sigma_n$ for some
  $\sigma_i \in H_J^{x_i}$; now by suitability of each $H^{x_i}$, there
  are elements $\tau_i \in H_I^{x_i}$ with $Tf. \tau_i = \sigma_i. Tf$
  for each $i$, and so
  $\tau = \tau_1 \cdots \tau_n \in \bigvee_x H_I^x$ is an element with
  $Tf . \tau = \sigma. Tf$.

  It follows that there is a largest suitable family of subgroups given
  by:
  \begin{equation*}
    \tilde G = \textstyle\bigvee \{H : H \text{ is a suitable family of subgroups}\}\rlap{ .}
  \end{equation*}
  By suitability,
  $[Tf] \colon (TI, \tilde G_I) \rightarrow (TJ, \tilde G_J)$ is
  well-defined for each $f \colon I \rightarrow J$; as
  $G_I \leqslant \tilde G_I$, this assignation is functorial in $f$ and
  so we obtain $(T, \tilde G) \colon \I \rightarrow \O(\A)$ and a
  vertical transformation $\xi \colon (T, \tilde G) \rightarrow (T,G)$
  with components $\xi_I = [1_{TI}]$. We claim that $(T,\tilde G)$ is
  in fact initial in its connected component.

  First we show that any pair of vertical transformations
  $\alpha, \beta \colon (T,\tilde G) \rightrightarrows (S, H)$ are
  equal. Each component $\alpha_I$ or $\beta_I$ is an equivalence class
  of maps $TI \rightarrow SI$, and so we may consider the family of
  subgroups
  \begin{equation*}
    (K_I = \spn{a^{-1}b \mid a \in \alpha_I, b \in \beta_I} \leqslant \aut_{TI})_{I \in \I}\rlap{ .}
  \end{equation*}
  We claim this family is suitable: for then $K_I \leqslant \tilde G_I$
  so that $a^{-1}b \in \tilde G_I$ for all $a \in \alpha_I$ and
  $b \in \beta_I$, whence $\alpha_I = \beta_I$ as required. For
  suitability, it suffices to show that, if
  $(a, b) \in \alpha_J \times \beta_J$ and $f \colon I \rightarrow J$,
  then there exists $(c,d) \in \alpha_I \times \beta_I$ with
  $Tf.c^{-1}d = a^{-1}b.Tf$. For \emph{any} $c \in \alpha_I$ we have by
  naturality of $\alpha$ that $[a.Tf] = [Sf.c]$; but then
  $a.Tf = Sf.(c\sigma)$ for some $\sigma \in \tilde G_I$, and so on
  replacing $c$ by $c\sigma \in \alpha_I$ we may take it that in fact
  $a.Tf = Sf.c$. Similarly, we can find $d \in \beta_I$ such that
  $b.Tf = Sf.d$, and now $Tf.c^{-1}d = a^{-1}b.Tf$ as required.

  To show initiality of $(T, \tilde G)$ in its connected component, it
  now suffices to show that, for all cospans
  $\alpha \colon (T, \tilde G) \rightarrow (S, H) \leftarrow (R, K)
  \colon \beta$ in $[\I, \O(\A)]_v$ there is \emph{some} map
  $\gamma \colon (T, \tilde G) \rightarrow (R,K)$. To this end,
  consider the family of subgroups
  \begin{equation*}
    (L_I = \spn{a^{-1}bc^{-1}d \mid a,d \in \alpha_I,\, b,c \in \beta_I} \leqslant \aut_{TI})_{I \in \I}\rlap{ .}
  \end{equation*}
  Repeating the above argument shows this family is suitable, and so
  $L_I \leqslant \tilde G_I$ for all $I \in \I$. Now, choosing any
  $a \in \alpha_I$ and $b \in \beta_I$, we have for each $I \in \I$ a
  well-defined map
  $\gamma_I = [b^{-1}a] \colon (TI, \tilde G_I) \rightarrow (RI, K_I)$.
  Indeed, if $\sigma \in K_I$ then $b \sigma \in \beta_I$ and so
  $\tau = a^{-1}b\sigma b^{-1}a \in L_I \leqslant \tilde G_I$ satisfies
  $b^{-1}a.\tau = \sigma b^{-1} a$ as required. Clearly
  $\beta_I.\gamma_I = \alpha_I$ for each $I$; since each $\beta_I$ is
  vertical, hence monic in $\O(\A)$, we conclude by naturality of
  $\alpha$ that $\gamma \colon (T, \tilde G) \rightarrow (R, K)$ is
  also natural as required.
\end{proof}

Using this, we are finally able to prove:

\begin{Prop}
  \label{prop:9}
  If $\A$ is essentially small, then the category
  $\anpt(\A, \psh)$ has a terminal object; in particular, any
  $\anpt^\omega(\psh[\D], \psh)$ has a terminal object.
\end{Prop}

\begin{proof}
  It suffices by Proposition~\ref{prop:4} to show that the equivalent
  category $\el_\C /\!\!/_{\!v}\,\, \O(\A)$ has a terminal object, and
  as before, we may assume without loss of generality that $\A$ is in
  fact small. For any $c \in \C$, we know by Lemma~\ref{lem:14} that
  each connected component of $[\C / c, \O(\A)]_v$ has an initial
  object; make a choice of such, and for each
  $F \in [\C / c, \O(\A)]_v$, write $\tilde F$ for the chosen initial
  object in its connected component, and
  $u \colon \tilde F \rightarrow F$ for the unique vertical
  transformation. The required terminal
  $(S, E) \in \el_\C /\!\!/_{\!v}\,\, \O(\A)$ now has:
  \begin{equation*}
    Sc = \{ F \in [\C/c, \O(\A)]_v : \tilde F = F \} \qquad \text{and} \qquad S(f \colon d \rightarrow c) \colon F \mapsto \widetilde{F(f \circ \thg)}\rlap{ ,}
  \end{equation*}
  and has $E \colon \el S \rightarrow \O(\A)$ given by
  $E(F, c) = F(1_c)$ on objects, and
  \begin{align*}
    f \colon (\widetilde{F(f\circ \thg)}, d) \rightarrow (F, c) \quad &\mapsto \quad 
    \widetilde{F(f \circ \thg)}(1_d) \xrightarrow{u_{1_d}} F(f \circ \thg)(1_d) = Ff \xrightarrow{Ff} F1_c
  \end{align*}
  on morphisms. To see terminality of $(S,E)$, let $(T, D)$ be another
  object of $\el_\C /\!\!/_{\!v}\,\, \O(\A)$. To define a map
  $f \colon T \rightarrow S$, we form for each $t \in Tc$ the composite
  \begin{equation}\label{eq:12}
    F_t \colon \C / c = \el \yoneda_c \xrightarrow{\el t} \el T \xrightarrow{D} \O(\A)
  \end{equation}
  and now define $p(t) = \widetilde{F_t} \in Sc$. For any
  $f\colon d \to c$, we have
  $$p(t) \cdot f = \widetilde{\widetilde{F_t} \circ \el f} =
  \widetilde{F_t \circ \el f} = \widetilde{F_{t \cdot f}} = p (t \cdot
  f),$$ so that $p$ is natural in $c$. Moreover, we have a pointwise
  vertical transformation
  \begin{equation*}
    \cd[@-1em@C-0.5em]{
      {\el T} \ar[rr]^-{\el p} \ar[dr]_-{D} & \ltwocello{d}{\varphi} &
      {\el S} \ar[dl]^-{E} \\ &
      {\O(\A)}
    }
  \end{equation*}
  whose component at $t \in Tc$ is the map
  $u_{1_c} \colon E(p(t)) = \widetilde{F_t}(1_c) \rightarrow F_t(1_c)
  = Dt$.
  So we have a map $(p, \varphi) \colon (T,D) \rightarrow (S,E)$ and
  to conclude the proof, we must show that any
  $(q, \psi) \colon (T,D) \rightarrow (S, E)$ is equal to
  $(p,\varphi)$. For each $t \in Tc$, consider the composite functor
  \begin{equation*}
    G_t \colon \C / c = \el \yoneda_c \xrightarrow{\el t} \el T \xrightarrow{\el q} \el S \xrightarrow{E} \O(\A)\rlap{ .}
  \end{equation*}
  By naturality of $q$, this functor sends $f \colon d \rightarrow c$
  to $\widetilde{q(t)(f \circ \thg)}(1_d) \in \O(\A)$, and there is now a
  vertical transformation
  $\xi \colon G_t \Rightarrow q(t) \colon \C / c \rightarrow \O(\A)$
  with component
  \begin{equation*}
    G_t(f) = \widetilde{q(t)(f \circ \thg)}(1_d)
    \xrightarrow{u_{1_d}} q(t)(f \circ \thg)(1_d) = q(t)(f)
  \end{equation*}
  at $f \colon d \rightarrow c \in \C / c$. Since $q(t)$ is a (chosen)
  initial object in its connected component of $[\C/c, \O(\A)]_v$, the
  map $\xi \colon G_t \Rightarrow q(t)$ must be a split epimorphism;
  since every map of $[\C/c, \O(\A)]_v$ is (pointwise monomorphic and
  hence) monomorphic, $\xi$ is thus invertible, so that
  $G_t \cong q(t)$. Since the composite vertical transformation
  \begin{equation}\label{eq:5}
    \cd[@-1em@C-0.5em]{
      {\C / c} \ar[rr]^-{\el qt} \ar[dr]_-{F_t = D.\el t} & \ltwocello{d}{\psi.\el t} &
      {\el S} \ar[dl]^-{E} \\ &
      {\O(\A)}
    }
  \end{equation}
  exhibits $G_t$ (the upper composite) as connected to $F_t$ in
  $[\C / c, \O(\A)]_v$, this determines $q(t)$ uniquely as being
  $\tilde F_t = p(t)$; since this holds for all $t \in \el T$, we
  conclude that $p = q$. Moreover, as $G_t \cong q(t)$ is initial in
  its connected component of $[\C / c, \O(\A)]_v$, the $2$-cell
  in~\eqref{eq:5} must be equal to
  $\varphi.\el t \colon G_t \Rightarrow F_t$; as this holds for all
  $t \in \el T$, we have $\varphi = \psi$ as required.
\end{proof}

\begin{Rk}
  \label{rk:2}
  For any small category $\C$, we may view the terminal object $U$ of
  $\anpt^\omega(\psh, \psh)$ as an object in $\anpt(\psh, \psh)$. By
  Proposition~\ref{prop:7} and by near-familiality, any
  $F \in \anpt(\psh, \psh)$ which admits a map to $U$ must itself be
  finitary, so that the map $F \rightarrow U$ is unique if it exists.
  In other words, $U$ is a \emph{subterminal} object in
  $\anpt(\psh, \psh)$; it follows that we can identify
  $\anpt^\omega(\psh, \psh)$ with the slice category
  $\anpt(\psh, \psh) / U$. We will revisit this point in
  Remark~\ref{rk:3} below.
\end{Rk}

\subsection{Composition of analytic functors}
\label{sec:comp-analyt-funct-1}

The passage from familial to analytic functors has thus fixed the
problem we had previously, namely the lack of a terminal object among
such functors. However, we are not in the clear yet, as we must still
show that pointwise analytic functors compose. By modifying
Lemma~\ref{lem:1}(i) to use Lemma~\ref{lem:12} in place of
Lemma~\ref{lem:5} we may show that generic morphisms compose; now
arguing as in Proposition~\ref{prop:10} yields:

\begin{Prop}
  \label{prop:16}
  If $F \colon \A \rightarrow \B$ and $G \colon \B \rightarrow \C$ are
  $(\A', \B')$- and $(\B', \C')$-analytic, then their composite is
  $(\A', \C')$-analytic, and has as generic operations at stage
  $C \in \C'$ precisely the composites
  $Gt.s \colon C \rightarrow GB \rightarrow GFA$ of $G$- and
  $F$-generic operations. The correspondingly analytic transformations
  between these functors are likewise composable; in particular, there
  is a $2$-category $\cat{AN}$ of categories, analytic functors and
  analytic transformations.
\end{Prop}

However, this does not imply that pointwise analytic functors between
presheaf categories are composable, since we do not know that
pointwise analytic functors are necessarily analytic. In fact, this is
not true, by virtue of:

\begin{Prop}
  \label{prop:6}
  Pointwise analytic functors between presheaf categories are not
  closed under composition.
\end{Prop}

\begin{proof}
  Consider the following two functors:
  \begin{equation}
    \begin{aligned}
      F \colon \cat{Set} &\rightarrow \cat{Set}^{\cat{2}} & \qquad \qquad G \colon \cat{Set}^\cat{2} & \rightarrow \cat{Set} \\
      X &\mapsto (X^2 \rightarrow X^2/\aut_2) & (A \rightarrow B) & \mapsto A \times_B A\rlap{ .}
    \end{aligned}\label{eq:16}
  \end{equation}
  $G$ is representable at $W = (2 \rightarrow 1)$, and so pointwise
  analytic; $F$ is pointwise analytic with spectrum
  $1 \in \cat{Set}^\cat{2}$ and exponent
  $\el 1 = \cat{2} \rightarrow \O(\cat{Set})$ picking out the arrow
  $(2,\aut_2) \rightarrow (2,1)$. The composite
  $GF \colon \cat{Set} \rightarrow \cat{Set}$ sends a set $X$ to
  \begin{equation*}
    X^2 \times_{X^2 / \aut_2} X^2 = \{(a,b,c,d) \in X^4 : (a,b) = (c,d) \text{ or } (a,b) = (d,c)\}\rlap{ .}
  \end{equation*}
  Now, no operation $(a,a,a,a) \colon 1 \rightarrow GFX$ can be
  generic, because the square left below has no filler
  $GFX \to GF\{0,1\}$; while if $a \neq b \in X$, then no $(a,b,a,b)$
  or $(a,b,b,a) \colon 1 \rightarrow GFX$ can be generic because the
  square below right has no filler in either direction.
  \begin{equation*}
    \cd[@-0.25em]{
      {1} \ar[r]^-{(0,1,0,1)} \ar[d]_{(a,a,a,a)} &
      {GF\{0,1\}} \ar[d]^{GF!} & &
      {1} \ar[r]^-{(a,b,b,a)} \ar[d]_{(a,b,a,b)} &
      {GFX} \ar[d]^{GF!} \\
      {GFX} \ar[r]^-{GF!} &
      {GF1} & &
      {GFX} \ar[r]^-{GF!} &
      {GF1}
    } 
  \end{equation*}
  So $GF$ is not pointwise analytic, as there are no generic operations
  in $1 \downarrow GF$.
\end{proof}

\begin{Cor}
  \label{cor:1}
  $\anpt^\omega(\psh, \psh)$ need not be monoidal under composition.
\end{Cor}

\begin{proof}
  If $F$ and $G$ are as in the preceding proof, then $F\pi_1$ and
  $\Delta G$ lie in $\anpt^\omega(\cat{Set}^\atwo, \cat{Set}^\atwo)$.
  But if their composite $\Delta GF \pi_1$ were pointwise analytic,
  then so too would be $\pi_1(\Delta GF \pi_1)\Delta = GF$.
\end{proof}

\begin{Rk}
  \label{rk:5}
  The preceding argument does not rule out the possibility that the
  composition-powers of the terminal finitary analytic endofunctor $U$
  of a presheaf category happen to be again analytic---which would
  allow for the construction of a monad structure on $U$. However, at
  least for the presheaf categories of our examples, the preceding
  argument may be adapted to show that this is not so.
\end{Rk}

\section{Cellular functors and shapeliness}
\label{sec:cell-analyt-funct}

We have now failed to construct a universal shapely monad on a
presheaf category $\psh$ in two different ways: there was no universal
\emph{familial} monad due to the lack of a terminal familial
endofunctor, while there was no universal \emph{analytic} monad due to
the failure of pointwise analytic functors to be composition-closed.

Our next attempt to produce a universal shapely monad will focus on a
special class of pointwise analytic functors, which we term
\emph{cellular}, that are closed under composition. This is
achieved by way of an additional condition which allows their
pointwise analyticity to be built up to analyticity at more complex
stages, so that Proposition~\ref{prop:16} can then be applied.

Building up this analyticity will require an analogue of
Proposition~\ref{prop:18}, which showed that the stages of familiality
of a functor $\A \rightarrow \B$ are closed under colimits. The reason
this does not carry over unchanged to the analytic setting is that the
analogue of Lemma~\ref{lem:1}(ii) fails to hold, as the following
explicit counterexample shows.

\begin{Prop}
  \label{prop:32}
  The generic operations of a functor $F \colon \A \rightarrow \B$
  need not be closed under pointwise colimits in $\B \downarrow F$.
\end{Prop}

\begin{proof}
  Consider the pointwise analytic functor
  $F \colon \cat{Set} \rightarrow \cat{Set}^\atwo$ from~\eqref{eq:16}.
  By examination of its spectrum, this admits generic operations
  $t \colon \yoneda_0 \rightarrow F2$ and
  $u \colon \yoneda_1 \rightarrow F2$ fitting into a span
  \begin{equation*}
    \cd[@!@-2em]{ 
      & \yoneda_0 \ar[dl]_{} \ar[dd]_-{t} \ar[rr]^{} &&
      \yoneda_1 \ar[dd]^{u} \\ 
      \yoneda_1 \ar[dd]_{u} \\
      & F2 \ar@{=}[rr] \ar@{=}[dl] && 
      F2 \\
      F2
    }
  \end{equation*}
  in $\cat{Set}^\atwo \downarrow F$. We claim that the pushout
  $u +_t u \colon \yoneda_1 +_{\yoneda_0} \yoneda_1 \rightarrow F2$ of
  this span is not generic; in fact, we claim that there are no generic
  operations at all in $\yoneda_1 +_{\yoneda_0} \yoneda_1 \downarrow F$. Indeed,
  the functor $G$ in~\eqref{eq:16} is representable at
  $\yoneda_1 +_{\yoneda_0} \yoneda_1$, and so
  $1 \downarrow GF \cong \yoneda_1 +_{\yoneda_0} \yoneda_1 \downarrow F$;
  but since by the proof of Proposition~\ref{prop:6}, the former
  category contains no generic operations, neither does the latter.
\end{proof}

In light of this negative result, our first objective in this section
will be to describe certain good colimit types under which generic
operations \emph{are} closed, and to show that for these colimit
types, we do have an analogue of Proposition~\ref{prop:18}. We then
introduce the notion of cellular functor, this being a pointwise
analytic functor whose generic operations have input arities that can
be constructed from representables using only these good
colimit types. With this in place, it is then reasonably
straightforward to show that cellular functors are closed under
composition.

\subsection{Arrow-genericity and arrow-analyticity}
\label{sec:analyt-at-morph}

By the same argument as for linear operations, generic operations are
closed under coproducts in $\B\downarrow F$, and at first this may
appear to be all that we can salvage from Proposition~\ref{prop:18} in
the analytic case. But in fact, there is a class of morphisms in
$\B \downarrow F$ along which generic operations are closed under
pushout; we now introduce this class.

\begin{Defn}
  \label{def:20}
  Let $F \colon \A \rightarrow \B$. We say that a map
  $(b,a) \colon t_1 \rightarrow t_2$ in $\B \downarrow F$ as below is
  \emph{arrow-generic} if $t_1$ and $t_2$ are generic for $F$ and
  $(t_1, t_2) \colon b \rightarrow F^\atwo(a)$ is generic for the
  functor $F^\atwo \colon \A^\atwo \rightarrow \B^\atwo$.
  \begin{equation}\label{eq:24}
    \cd{
      {B_1} \ar[r]^-{b} \ar[d]_{t_1} &
      {B_2} \ar[d]^{t_2} \\
      {FA_1} \ar[r]^-{Fa} &
      {FA_2}
    }
  \end{equation}
\end{Defn}

Just as with linear and generic operations, there is a
characterisation of arrow-genericity as a diagonal filling property:

\begin{Lemma}
  \label{lem:4}
  Let $t_1, t_2$ be $F$-generic operations. A map
  $(b,a) \colon t_1 \rightarrow t_2$ in $\B \downarrow F$ is
  arrow-generic just when for every commuting diagram as below (with
  $hj = ka$), there exists $\ell$ as shown with $h\ell = k$ and
  $F\ell.t_2 = u$ and $\ell a = j$.
  \begin{equation}\label{eq:15}
    \cd[@R-1em]{
      {B_1} \ar[r]^-{b} \ar[dd]_{t_1} &
      {B_2} \ar[dd]_(0.27){t_2} \ar[r]^-{u} &
      {FY} \ar[dd]^-{Fh}\\ & {} \\
      {FA_1} \ar[r]_-{Fa} \ar'[ur]^-{Fj}[uurr] &
      {FA_2} \ar[r]_-{Fk} \ar@{-->}[uur]_-{F\ell}& {FZ}
    }
  \end{equation}
  It follows that the class of arrow-generic maps in $\B \downarrow F$
  contains the isomorphisms and is composition-closed.
\end{Lemma}

\begin{proof}
  First assume the condition in the statement. We must show that for
  any cube as below left, there are diagonal fillers
  $j_1 \colon A_1 \rightarrow Y_1$ and
  $j_2 \colon A_2 \rightarrow Y_2$ with $h_i j_i = k_i$ and
  $Fj_i.t_i = u_i$ and $j_2 a = yj_1$. Applying genericity of $t_1$ to
  the front face yields the required $j_1$; now the left and back
  faces give the solid part of a diagram as in~\eqref{eq:15} with the
  composite $yj_1$ as its diagonal, and applying the stated condition
  to this yields a filler $j_2$ for the back face satisfying the
  required equations.
  \begin{equation*}
    \cd[@!@-1.4em@R-0.8em]{ 
      & B_2 \ar'[d][dd]_(0.35){t_2} \ar[rr]^{u_2} && 
      FY_2 \ar[dd]^{Fh_2} \\ 
      B_1 \ar[ur]^{b} \ar[dd]_{t_1} \ar[rr]^(0.65){u_1} && 
      FY_1 \ar[ur]_-{Fy} \ar[dd]^(0.3){Fh_1} \\
      & FA_2 \ar'[r]^(0.6){Fk_2}[rr] && 
      FZ_2 \\
      FA_1 \ar[rr]_{Fk_1} \ar[ur]^{Fa} & & 
      FZ_1 \ar[ur]_{Fz}
    } \qquad \qquad 
    \cd[@!@-1.4em@R-0.8em]{ 
      & B_2 \ar'[d][dd]_(0.35){t_2} \ar[rr]^{u} && 
      FY \ar[dd]^{Fh} \\ 
      B_1 \ar[ur]^{b} \ar[dd]_{t_1} \ar[rr]^(0.65){ub} && 
      FY \ar[ur]_-{F1} \ar[dd]^(0.3){F1} \\
      & FA_2 \ar'[r]^(0.6){Fk}[rr] && 
      FZ \\
      FA_1 \ar[rr]_{Fj} \ar[ur]^{Fa} & & 
      FY \ar[ur]_{Fh}
    }
  \end{equation*}
  Suppose conversely that $(b,a) \colon t_1 \rightarrow t_2$ is
  arrow-generic. Given a diagram as in the solid part of~\eqref{eq:15},
  we apply arrow-genericity to the cube above right to obtain fillers
  for the front and back faces making everything commute. The front
  filler is necessarily $j$, and so the back filler is the
  $\ell \colon A_2 \rightarrow Y$ required for~\eqref{eq:15}.
\end{proof}

As mentioned above, the reason for introducing arrow-generic maps is
that generic operations in $\B \downarrow F$ are closed under pushout
along them; we show this in the next section, but first let us introduce
the associated notion of analyticity.

\begin{Defn}
  \label{def:24}
  A functor $F \colon \A \rightarrow \B$ is \emph{arrow-analytic at
    stage $b \in \B(B_1,B_2)$} if $F$ is analytic at stages
  $B_1, B_2 \in \B$ and $F^\atwo$ is analytic at stage
  $b \in \B^\atwo$; we define arrow-analyticity of a transformation
  $\alpha \colon F \Rightarrow G$ correspondingly.
\end{Defn}

It should not yet be clear whether arrow-analyticity is a property
that will be fulfilled in examples of interest. We will see that this
is the case in Lemma~\ref{lem:16} below, where we characterise
arrow-generic morphisms $t_1 \rightarrow t_2$ in terms of an
easily-satisfied relation between the automorphism groups of
$t_1 \in B_1 \downarrow F$ and $t_2 \in B_2 \downarrow F$. Combining
this with the following lemma will allow us to find many examples of
arrow-analytic functors and transformations.

\begin{Lemma}
  \label{lem:6}
  A functor $F \colon \A \rightarrow \B$ is arrow-analytic at 
  $b \in \B(B_1, B_2)$ if and only if it is analytic at
  $B_1, B_2 \in \B$ and every $(b,a) \colon t_1 \rightarrow t_2$ in
  $\B \downarrow F$ between generic operations is arrow-generic. A
  transformation $\alpha \colon F \Rightarrow G$ between two such
  functors is arrow-analytic at $b$ if and only if it is
  analytic at $B_1, B_2 \in \B$.
\end{Lemma}

\begin{proof}
  If $F$ is analytic at stages $B_1, B_2 \in \B$, then every square as
  to the front of the diagram below left can be factorised through the
  back faces as displayed:
  \begin{equation*}
    \cd[@!C@-0.8em@C-0.5em]{
      & B_1 \ar[rr]^-{b} \ar[dd]_(0.7){t_1} \ar[dl]_-{\tilde t_1} & & B_2 \ar[dd]^-{t_2} \ar[dl]_-{\tilde t_2} \\
      FA_1 \ar'[r][rr]^-{Fa} \ar[dr]_-{Fu_1} & & FA_2 \ar[dr]_-{Fu_2} \\
      & FX_1 \ar[rr]^(0.4){Fx} & & FX_2
    } \qquad \qquad 
    \cd{
      {B_1} \ar[r]^-{b} \ar[d]_{s_1} &
      {B_2} \ar[d]^{s_2} \\
      {FY_1} \ar[r]_-{Fy} &
      {FY_2}\rlap{ .}
    }
  \end{equation*}
  Here, $u_1 \colon \widetilde t_1 \rightarrow t_1$ and
  $u_2 \colon \widetilde t_2 \rightarrow t_2$ are generic covers
  obtained from analyticity at $B_1$ and $B_2$, and $a$ is induced by
  applying Lemma~\ref{lem:12} to the generic $\widetilde t_1$. Now if
  the hypotheses in the statement hold, then
  $(\widetilde t_1, \widetilde t_2)$ is $F^\atwo$-generic and so each
  $(t_1, t_2)$ in $b \downarrow F^\atwo$ admits a generic cover, as
  required for $F^\atwo$ to be analytic at $b$. Suppose conversely that
  $F^\atwo$ is analytic at $b$, and consider a square as right above
  with generic sides; we must show that it is arrow-generic. So
  construct covers
  \begin{equation*}
    (\widetilde t_1, \widetilde t_2) \xrightarrow{(u_1, u_2)} (t_1,
    t_2) \xrightarrow{(v_1, v_2)} (s_1, s_2)
  \end{equation*}
  in $b \downarrow F^\atwo$, where $(t_1, t_2)$ is an $F^\atwo$-generic
  cover, and where $(\widetilde t_1, \widetilde t_2)$ is obtained as
  above left using analyticity of $F$ at $B_1, B_2$. Since
  $\widetilde t_1$ and $s_1$ are both $F$-generic operations at stage
  $B_1$, $v_1u_1$ is invertible by Corollary~\ref{cor:2} and so $u_1$
  is a split monomorphism; similarly $u_2$ is split monic. On the other
  hand, since $(t_1, t_2)$ is $F^\atwo$-generic, the map $(u_1, u_2)$
  must---by Corollary~\ref{cor:2} again---be a pointwise split
  epimorphism: whence $u_1, u_2, v_1$ and $v_2$ are invertible, so that
  $(s_1, s_2)$, like $(t_1, t_2)$, is arrow-generic as required. It
  follows that, if $F$ is arrow-analytic at $b$, then the generic
  operations in $b \downarrow F^\atwo$ are precisely the squares with
  generic sides; the statement about arrow-analytic transformations
  follows directly from this.
\end{proof}

In the sequel, we will make use of this characterisation of
arrow-analyticity without further comment. We conclude this section by
recording the analogue of Definitions~\ref{def:18} and~\ref{def:12}
for arrow-analytic functors:

\begin{Defn}
  \label{def:29}
  If $\I \subset \A^\atwo$ and $\J \subset \B^\atwo$ are full replete
  subcategories, we say that $F \colon \A \rightarrow \B$ is
  \emph{$(\I,\J)$-arrow-analytic} if it is arrow-analytic at each
  $b \in \J$, and each arrow-generic
  $(t_1, t_2) \in b \downarrow F^\atwo$ as in~\eqref{eq:24} has its
  input arity $a$ in $\I$. A transformation
  $\alpha \colon F \Rightarrow G$ between such functors is
  \emph{$(\I, \J)$-arrow-analytic} if it is arrow-analytic at every
  stage $b \in \J$.
\end{Defn}

\begin{Prop}
  \label{prop:19}
  If $F \colon \A \rightarrow \B$ and $G \colon \B \rightarrow \C$ are
  $(\I, \J)$- and $(\J, \K)$-arrow-analytic, then their composite is
  $(\I, \K)$-arrow-analytic, and correspondingly for the
  transformations between such functors.
\end{Prop}

\begin{proof}
  Direct from the definitions and Proposition~\ref{prop:16}.
\end{proof}

\subsection{Building up stages of analyticity}
\label{sec:building-up-stages}

We are now ready to see what the purpose of arrow-genericity and
arrow-analyticity really is. We begin with a lemma which provides an
analytic analogue of Lemma~\ref{lem:1}(ii) above.
 
\begin{Lemma}
  \label{lem:7}
  Let $F \colon \A \rightarrow \B$ and consider a pointwise pushout
  square in $\B \downarrow F$ as below. If $s_1$,
  $s_2$ and $t_1$ are generic and $(d,c)$ is arrow-generic, then $t_2$ is also
  generic and $(b,a)$ is also arrow-generic.
  \begin{equation}\label{eq:26}
    \cd[@C+1em@!]{
      {s_1} \ar[r]^-{(g,f)} \ar[d]_{(d,c)} &
      {t_1} \ar[d]^{(b,a)} \\
      {s_2} \ar[r]^-{(m,n)} &
      {t_2}
    }
  \end{equation}
\end{Lemma}

\begin{proof}
  We will show that every diagram as in the solid part
  of~\eqref{eq:15} admits a dashed filler; since $t_1$ is generic,
  this immediately implies that $t_2$ is generic, and so by
  Lemma~\ref{lem:4} that $(b,a)$ is arrow-generic. To prove the claim,
  observe that the stated filling condition can be described as a
  \emph{left lifting property}: it says that, for each
  $h \colon Y \rightarrow Z$ in $\A$, each square in $\B \downarrow F$
  as left below has a diagonal filler.
  \begin{equation}\label{eq:18}
    \cd[@+0.5em]{
      {t_1} \ar[r]^-{(ub,j)} \ar[d]_{(b,a)} &
      {1_{FY}} \ar[d]^{(1, h)} \\
      {t_2} \ar[r]_-{(u,k)} \ar@{-->}[ur]|-{(u,\ell)} &
      {Fh}
    } \qquad \qquad 
    \cd[@+0.5em]{
      {s_1} \ar[r]^-{(ubg,jf)} \ar[d]_{(d,c)} &
      {1_{FY}} \ar[d]^{(1, h)} \\
      {s_2} \ar[r]_-{(um,kn)} \ar@{-->}[ur]|-{(um,\ell')} &
      {Fh}
    }
  \end{equation}
  Pasting the given square with the pushout~\eqref{eq:26} gives a
  square as right above; since $(d,c)$ is arrow-generic, we induce a
  filler for this square as indicated and so by the universal property
  of pushout the required filler $(u,\ell)$ as left above.
\end{proof}

In fact, we can do better than this: the characterisation of
arrow-generic maps by a left lifting property allows us to show that
they are also closed under \emph{transfinite
  composition}~\cite[Definition~2.1.1]{Hovey1999Model}. As we do not
need this further fact, we leave its verification to the interested
reader.

We now use the preceding lemma to give the promised analytic analogue
of Proposition~\ref{prop:18}; we reiterate that, in light of
Proposition~\ref{prop:32}, the assumption of arrow-analyticity of $b$
in part (ii) cannot be dropped.

\begin{Prop}
  \label{prop:17}
  Let $\A$ be cocomplete and let $F \colon \A \rightarrow \B$.
  \begin{enumerate}[(i)]
  \item If $F$ is analytic at $B$, then it is arrow-analytic at $1_B$;
    if $F$ is arrow-analytic at composable maps $b$ and $c$, then it
    is also arrow-analytic at $cb$.\vskip0.25\baselineskip
  \item \label{item:arrow:analyticity:pushout} For any pushout as
    below in $\B$, if $F$ is analytic at $B_1$, $B_2$ and $C_1$ and
    arrow-analytic at $b$, then it is also analytic at $C_2$ and
    arrow-analytic at $c$.
    \begin{equation}\label{eq:27}
      \cd{
        {B_1} \ar[r]^-{f_1} \ar[d]_{b} &
        {C_1} \ar[d]^{c} \\
        {B_2} \ar[r]^-{f_2} &
        {C_2}
      }
    \end{equation}
  \end{enumerate}
  The analogous results hold for natural transformations
  $\alpha \colon F \Rightarrow G \colon \A \rightarrow \B$.
\end{Prop}

Again, we could add an additional clause to this proposition showing
closure of stages of arrow-analyticity under transfinite composition,
but we refrain from doing so as we have no use for it in what follows.

\begin{proof}
  Let $F$ be analytic at $B$. Any map $(1_B, a) \colon t_1 \rightarrow
  t_2$ between generic operations in $\B \downarrow F$ is a map in $B
  \downarrow F$, so that $f$ is invertible by Corollary~\ref{cor:2}
  and $(1_B, f)$ is arrow-generic by Lemma~\ref{lem:4}. This shows
  that $B$ is arrow-analytic at $1_B$ by Lemma~\ref{lem:6}. Suppose
  next that $F$ is arrow-analytic at $b \colon B_1 \rightarrow B_2$
  and $c \colon B_2 \rightarrow B_3$.  Given a square as below left
  with $t_1$ and $t_3$ generic, let $e \colon t_2 \rightarrow t_3c$ be
  a generic cover in $B_2 \downarrow F$ and let $d$ be induced by
  genericity as centre below; this yields a factorisation of the left
  square as to the far right.
  \begin{equation*}
    \cd{
      {B_1} \ar[r]^-{cb} \ar[d]_{t_1} &
      {B_3} \ar[d]^{t_3} & 
      {B_1} \ar[r]^-{t_2b} \ar[d]_{t_1} &
      {FA_2} \ar[d]^{Fe} & 
      {B_1} \ar[r]^-{b} \ar[d]_{t_1} &
      {B_2} \ar[r]^-{c} \ar[d]_{t_2} &
      {B_3} \ar[d]^{t_3} \\
      {FA_1} \ar[r]^-{Fa} &
      {FA_3} &
      {FA_1} \ar[r]_-{Fa} \ar@{-->}[ur]|-{Fd}&
      {FA_3} &
      {FA_1} \ar[r]^-{Fd} &
      {FA_2} \ar[r]^-{Fe} &
      {FA_3}
    }
  \end{equation*}
  By arrow-analyticity at $b$ and $c$, both small squares are
  arrow-generic, whence also their composite by Lemma~\ref{lem:4}; this
  shows that $B$ is arrow-analytic at $cb$ as required for (i). Now
  suppose the hypotheses of (ii). We first show that $F$ is analytic at
  $C_2$. Let $z \colon C_2 \rightarrow FZ$, and consider the left cube
  in:
  \begin{equation}\label{eq:19}
    \cd[@!@-1.4em@R-0.8em]{ 
      & B_1 \ar[dl]_{b} \ar'[d][dd]_(0.35){t_1} \ar[rr]^{f_1} && 
      C_1 \ar[dl]^-c \ar[dd]^{u_1} \\ 
      B_2 \ar[dd]_{t_2} \ar[rr]^(0.6){f_2} && 
      C_2 \ar[dd]^(0.3){z} \\
      & FA_1 \ar'[r]^(0.6){Fg_1}[rr] \ar[dl]_{Fa\!} && 
      FY_1\ar[dl]^{Fn} \\
      FA_2 \ar[rr]_{Fm} & & 
      FZ } \qquad \qquad 
    \cd[@!@-1.4em@R-0.8em]{ 
      & B_1 \ar[dl]_{b} \ar'[d][dd]_(0.35){t_1} \ar[rr]^{f_1} && 
      C_1 \ar[dl]^-c \ar[dd]^{u_1} \\ 
      B_2 \ar[dd]_{t_2} \ar[rr]^(0.6){f_2} && 
      C_2 \ar[dd]^(0.3){u_2} \\
      & FA_1 \ar'[r]^(0.6){Fg_1}[rr] \ar[dl]_{Fa\!} && 
      FY_1\ar[dl]^{Fy} \\
      FA_2 \ar[rr]_{Fg_2} & & 
      FY_2 
    }
  \end{equation}
  The front, left and right faces arise from generic covers
  $m \colon t_2 \rightarrow zf_2$, $a \colon t_1 \rightarrow t_2b$ and
  $n \colon u_1 \rightarrow zc$, while the map $g_1$ across the back
  face is obtained as in (i) using genericity of $t_1$. Since the top
  face is a pushout, and $\A$ is cocomplete, the back and left faces
  admit a pushout in $\B \downarrow F$ which may be taken to be as
  right above. Since $t_1$ and $t_2$ are generic and $F$ is
  arrow-analytic at $b$, the map $(b,a) \colon t_1 \rightarrow t_2$ is
  arrow-generic; since $u_1$ is also generic, we conclude by
  Lemma~\ref{lem:7} that $u_2$ is generic and
  $(c,y) \colon u_1 \rightarrow u_2$ is arrow-generic.

  Now taking $w \colon Y_2 \rightarrow Z$ to be the unique map with
  $wy = n$ and $wg_2 = m$, we see that $w \colon u_2 \rightarrow z$
  provides a generic cover of $z$ in $C_2 \downarrow F$, so that $F$ is
  analytic at $C_2$ as required. For arrow-analyticity at $c$, suppose
  that $(c,n) \colon u_1 \rightarrow z$ is a map between generic
  operations in $\B \downarrow F$. We may complete this to a cube as
  left above and form the generic pushout $u_2$ as to the right. Now
  since $z$ is generic, the induced map $u_2 \rightarrow z$ in
  $C_2 \downarrow F$ is invertible by Corollary~\ref{cor:2}. So the
  left cube above is also a pushout; as
  $(b,a) \colon t_1 \rightarrow t_2$ is arrow-generic, so too is
  $(c,n) \colon u_1 \rightarrow z$ by Lemma~\ref{lem:7}.
\end{proof}

\subsection{Cellular analytic functors}
\label{sec:cell-analyt-funct-1}

By using Proposition~\ref{prop:17}, we can now build up the
analyticity of a pointwise analytic functor between presheaf
categories to analyticity at more complex stages by assuming suitable
instances of arrow-analyticity. In order to specify what these more
complex stages are, we borrow some ideas from algebraic topology, in
particular the theory of \emph{cell complexes} in model categories;
see~\cite[\sec 2.1.2]{Hovey1999Model}, for example.

\begin{Defn}
  \label{def:26}
  Let $I$ be a class of maps in a category $\C$ with an initial
  object. A map $f \colon X \rightarrow Y$ is called a \emph{finite
    relative $I$-complex} if either it is an isomorphism, or it can be
  written as a finite composite
  \begin{equation}\label{eq:14}
    X = X_0 \xrightarrow{f_1} X_1 \rightarrow \cdots \xrightarrow{f_m} X_m = Y
  \end{equation}
  where each $f_i$ is a pushout of a map in $I$. An object $Y \in \C$
  is called a \emph{finite $I$-complex} if the unique map
  $0 \rightarrow Y$ is a finite relative $I$-complex. We write
  $\mathbf{Cx}(I) \subset \C$ for the full subcategory on the finite
  $I$-complexes, and $\mathbf{Cx}^\atwo(I) \subset \C^\atwo$ for the
  full subcategory on the relative finite $I$-complexes between finite
  $I$-complexes.
\end{Defn}

The modifier ``finite'' here comes from the fact that in~\eqref{eq:14}
we allow only finite compositions; the general notion of cell complex
in topology also allows for \emph{transfinite} ones, and everything
that follows could be adapted to this greater generality; however, like
before, we refrain from giving this as we will not need it.

\begin{Defn}
  \label{def:22}
  Let $\C$ be a small category. A \emph{bordage} on $\psh$ is a
  filtered family
  $\emptyset = I_0 \subset I_1 \subset \cdots \subset \bigcup_n I_n =
  I$ of maps in $\psh$ such that:
  \begin{enumerate}[(i)]
  \item Each $g \in I_{n+1}$ has representable codomain and domain a
    finite $I_{n}$-complex;
  \item Each representable object of $\psh$ is a finite $I$-complex.
  \end{enumerate}
\end{Defn}

Note that condition (i) for a bordage ensures that each map in $I$ has
domain a finite $I$-complex and, as such, is an object of
$\mathbf{Cx}^\atwo(I)$; this is something which need not be true for a
general class of maps $I$. Condition (ii) is much less important than
(i) and will only play a role in Proposition~\ref{prop:13} below.

\begin{Exs}
  \label{ex:2}
  \begin{enumerate}[(i)]
  \item Any presheaf category $\psh$ has a bordage given by
    $I = I_1 = \{\,0 \rightarrow \yoneda_c : c \in \C\,\}$. The finite
    $I$-complexes are the finite coproducts of representables, and the
    finite relative $I$-complexes are coproduct injections with
    complement a finite $I$-complex.\vskip0.25\baselineskip
  \item Let $\atwo$ be the arrow category $f \colon 0 \rightarrow 1$.
    The presheaf category $\psh[\atwo]$ has a bordage given by
    $I_1 = \{\, 0 \rightarrow \yoneda_0 \}$ and
    $I \setminus I_1 = \{\yoneda_f \colon \yoneda_0 \rightarrow
    \yoneda_1\}$. The finite $I$-complexes are all finitely presentable
    presheaves, and the finite relative $I$-complexes are the
    monomorphisms with cofinite image.\vskip0.25\baselineskip
  \item Let $\G$ be the category $s,t\colon 0 \rightrightarrows 1$. The
    presheaf category $\psh[\G]$ has a bordage given by
    $I_1 = \{\, 0 \rightarrow \yoneda_0 \}$ and
    $I \setminus I_1 = \{\spn{\yoneda_s, \yoneda_t} \colon \yoneda_0 +
    \yoneda_0 \rightarrow \yoneda_1\}$, whose finite $I$-complexes and
    finite relative $I$-complexes are as in (ii).\vskip0.25\baselineskip
  \item Changing $I \setminus I_1$ in the preceding example to be
    $\{\yoneda_t \colon \yoneda_0 \rightarrow \yoneda_1\}$ yields
    another bordage on $\psh[\G]$ whose finite $I$-complexes are now
    finite \emph{forests} whose edges are all directed towards the
    roots. Changing $I\setminus I_1$ to be
    $\{\yoneda_s \colon \yoneda_0 \rightarrow \yoneda_1\}$ yields finite
    forests with edges directed away from the roots, while taking
    $I\setminus I_1$ to be $\{\yoneda_s, \yoneda_t\}$ yields finite
    forests whose edges may be oriented arbitrarily.
  \end{enumerate}
\end{Exs}

We will see further examples of bordages when we revisit the
motivating examples of polycategories, properads and \textsc{prop}s in
Section~\ref{sec:free-shapely-monads} below.

\begin{Defn}
  \label{def:28}
  Let $I$ be a bordage on $\psh$. A pointwise analytic
  $F \colon \A \rightarrow \psh$ is \emph{$I$-cellular} if any square
  as below with $b \in I$ and $t_1, t_2$ generic is arrow-generic.
  \begin{equation}\label{eq:28}
    \cd{
      {S} \ar[r]^-{b} \ar[d]_{t_1} &
      {\yoneda_{c}} \ar[d]^{t_2} \\
      {FA_1} \ar[r]^-{Fa} &
      {FA_2}
    }
  \end{equation}
\end{Defn}

We will see in the following section that the cellularity condition is
very easy to check in practice. Note that cellularity \emph{almost}
says that $F$ is arrow-analytic at $b \colon S \rightarrow \yoneda_c$
for each $b \in I$, except that we do not assume that $F$ is analytic
at the domain object $S$. In fact, this is true by virtue of:

\begin{Prop}
  \label{prop:15}
  Let $I$ be a bordage on $\psh$ and let $\A$ be a cocomplete
  category. Any $I$-cellular $F \colon \A \rightarrow \psh$ is
  analytic at all $B \in \mathbf{Cx}(I)$ and arrow-analytic at all
  $b \in \mathbf{Cx}^\atwo(I)$. The same holds for pointwise analytic
  transformations $\alpha \colon F \Rightarrow G$ between $I$-cellular
  functors.
\end{Prop}

\begin{proof}
  Let $F$ be $I$-cellular. We prove by induction on $n$ that $F$ is
  analytic at every $B \in \mathbf{Cx}(I_n)$ and
  $b \in \mathbf{Cx}^\atwo(I_n)$. For the base case $n = 0$, every
  $B \in \mathbf{Cx}(I_0)$ is initial: thus $B \downarrow F \cong \A$,
  and so as $\A$ has an initial object, $F$ is analytic at $B$. Since
  any $b \in \mathbf{Cx}^\atwo(I_0)$ is invertible, $F$ is analytic at
  $b$ by Proposition~\ref{prop:17}.

  Now assume the result for $n$. Each map
  $b \colon S \rightarrow \yoneda_c$ in $I_{n+1}$ has domain in
  $\mathbf{Cx}(I_n)$, and so $F$ is analytic at $S$; thus
  $I$-cellularity implies that $F$ is arrow-analytic at every
  $b \in I_{n+1}$. Applying Proposition~\ref{prop:17} finitely many
  times shows that, if $f \colon X \rightarrow Y$ is a finite relative
  $I_{n+1}$-cell complex for which $F$ is analytic at $X$, then $F$ is
  also analytic at $Y$ and arrow-analytic at $f$. Taking $X$ to be
  initial and using the base case, shows that $F$ is analytic at every
  $B \in \mathbf{Cx}(I_{n+1})$; while taking $X$ to be an arbitrary
  finite $I_{n+1}$-complex shows that $F$ is arrow-analytic at every
  $b \in \mathbf{Cx}^\atwo(I_{n+1})$. The case of transformations is
  similar, and so omitted.
\end{proof}

The preceding proposition shows us that the pointwise analyticity of
functors $F \colon \psh \rightarrow \psh[\D]$ and
$G \colon \psh[\D] \rightarrow \psh[\E]$ is stable under composition
if there is a bordage $I$ on $\psh[\D]$ such that $F$ is
$I$-cellular and the input arities of $G$'s generic operations are  
$I$-cell complexes. However, $GF$ need not then satisfy any
cellularity conditions allowing it to compose further; the following
definition ensures this.

\begin{Defn}
  \label{def:25}
  Let $I$ and $J$ be bordages on $\psh$ and $\psh[\D]$. A pointwise
  analytic functor $F \colon \psh \rightarrow \psh[\D]$ is called
  \emph{$(I, J)$-cellular} if every square~\eqref{eq:28} with
  $b \in J$ and $t_1, t_2$ generic is arrow-generic and has $a$ a
  finite relative $I$-complex. We write $\cat{CELL}((\C, I), (\D, J))$
  for the category of $(I,J)$-cellular functors and pointwise analytic
  transformations.
\end{Defn}

\begin{Prop}
  \label{prop:13}
  Let $I$ and $J$ be bordages on $\psh$ and $\psh[\D]$. A functor
  $F \colon \psh \rightarrow \psh[\D]$ is $(I, J)$-cellular if and
  only if it is
  $(\mathbf{Cx}^\atwo(I), \mathbf{Cx}^\atwo(J))$-arrow-analytic. In
  particular, there is a $2$-category $\mathbf{CELL}$ of presheaf
  categories equipped with bordages, cellular analytic functors, and
  pointwise analytic transformations.
\end{Prop}

\begin{proof}
  By condition (ii) for a bordage, each representable in $\psh[\D]$ is
  a finite $J$-complex, whence any $(\mathbf{Cx}^\atwo(I),
  \mathbf{Cx}^\atwo(J))$-arrow-analytic functor is pointwise
  analytic. By condition (i), we have $J \subset
  \mathbf{Cx}^\atwo(J)$, and so any $(\mathbf{Cx}^\atwo(I),
  \mathbf{Cx}^\atwo(J))$-arrow-analytic functor is
  $(I,J)$-cellular. This proves the ``if'' direction. For the ``only
  if'', if $F$ is $(I, J)$-cellular, then it is certainly
  $J$-cellular, and so arrow-analytic at each $b \in
  \mathbf{Cx}^\atwo(J)$; it remains to show that each
  square~\eqref{eq:24} with $b \in \mathbf{Cx}^\atwo(J)$ and $t_1,
  t_2$ generic has $a \in \mathbf{Cx}^\atwo(I)$. Let $K$ denote the
  collection of all arrows $b \in \mathbf{Cx}^\atwo(J)$ for which each
  square~\eqref{eq:24} with $t_1, t_2$ generic has $a$ a finite
  relative $I$-complex. By assumption $J \subset K$, and by examining
  the proof of Proposition~\ref{prop:17} we see that $K$ is stable
  under pushouts along maps into $J$-complexes and closed under
  composition, and so must comprise all of $\mathbf{Cx}^\atwo(J)$. So
  each~\eqref{eq:24} with $b \in \mathbf{Cx}^\atwo(J)$ and $t_1, t_2$
  generic has $a$ a relative finite $I$-complex, and it remains to
  show that the domain and codomain of $a$ are in fact finite
  $I$-complexes. But whenever $B \in \mathbf{Cx}(J)$ and $t \colon B
  \rightarrow FA$ is generic, the following square has top edge in
  $\mathbf{Cx}^\atwo(J)$ and generic sides:
  \begin{equation*}
    \cd{
      {0} \ar[r]^-{!} \ar[d]_{!} &
      {B} \ar[d]^{t} \\
      {F0} \ar[r]^-{F!} &
      {FA}
    }
  \end{equation*}
  and so has bottom map a finite relative $I$-complex; whence
  $A \in \mathbf{Cx}(I)$ as required. The existence of the $2$-category
  $\cat{CELL}$ now follows from this together with
  Proposition~\ref{prop:19}.
\end{proof}

\subsection{A combinatorial characterisation of arrow-genericity}
\label{sec:comb-char-cell}

We now know that cellular pointwise analytic functors are closed under
composition; what we have not yet seen is that there \emph{are} any
cellular functors. In this section, we give a result which will allow
us to verify that a given pointwise analytic functor is indeed
cellular. The key concept required is that of a \emph{minimal
  extension}.

\begin{Defn}
  \label{def:15}
  Let $f \colon A \rightarrow B$ and let $\sigma \in \aut_A$. An
  \emph{extension} of $\sigma$ along $f$ is some $\tau \in \aut_B$ for
  which $\tau f = f \sigma$. An extension is called \emph{minimal} if
  whenever $g \colon B \rightarrow C$ satisfies $gf\sigma = gf$, also
  $g\tau = g$; equivalently, if
  $f, f\sigma \colon A \rightrightarrows B$ admit a coequaliser $q$,
  then $\tau$ is minimal just when $q \tau = q$.
\end{Defn}

The name is motivated by the case of $f \colon A \rightarrowtail B$ an
injection in $\cat{Set}$; for such an $f$, any $\sigma \in \aut_A$ has
a unique minimal extension $\tau \in \aut_B$ with
\begin{equation}\label{eq:17}
  \tau(x) =
  \begin{cases}
    \sigma(x) & \text{ if $x \in \im(f)$;}\\
    x & \text{otherwise.}
  \end{cases}
\end{equation}
So $\tau$ is minimal among extensions of $\sigma$ in that it permutes
the smallest possible part of $B$. This intuition works for
monomorphisms in any presheaf category:

\begin{Prop}
  \label{prop:8}
  If $f \colon A \rightarrowtail B$ is monic in $\psh$ then each
  $\sigma \in \aut_A$ admits at most one minimal extension along $f$;
  this extension exists just when
  \begin{equation}\label{eq:13}
    x \notin f(A) \text{ and } xh \in f(A) \qquad \Longrightarrow \qquad xh \in f(A^{\setminus \sigma})
  \end{equation}
  for all $x \in Bc$ and $h \colon d \rightarrow c$, and is then given
  componentwise as in~\eqref{eq:17}.
\end{Prop}

\begin{proof}
  If $q$ coequalises $f$ and $f \sigma$, then $\tau \in \aut_B$ is a
  minimal extension of $\sigma$ just when $q \tau = \tau$, just when
  $q_c \tau_c = \tau_c$ for all $c$, just when each $\tau_c$ is a
  minimal extension of $\sigma_c$ (as colimits in $\psh$ are
  pointwise). So any minimal extension of $\sigma$ must be given
  componentwise by~\eqref{eq:17}, with \eqref{eq:13} being just what
  is needed to ensure naturality of this definition in $c$.
\end{proof}

We now use the notion of minimal extension to give a combinatorial
characterisation of arrow-generic morphisms. As in Remark~\ref{rk:1},
we use $\aut_{t_1}$ and $\aut_{t_2}$ to denote the automorphism groups
of $t_1 \in B_1 \downarrow F$ and $t_2 \in B_2 \downarrow F$.

\begin{Lemma}
  \label{lem:16}
  Let $\A$ have coequalisers and $F \colon \A \rightarrow \B$. A map
  $(b,a) \colon t_1 \rightarrow t_2$ between generic operations in
  $\B \downarrow F$ is arrow-generic if and only if each
  $\sigma \in \aut_{t_1}$ admits a minimal extension
  $\tau \in \aut_{t_2}$ along $a \colon A_1 \rightarrow A_2$.
\end{Lemma}

\begin{proof}
  We use the alternate characterisation of arrow-genericity of
  Lemma~\ref{lem:4}. For the ``if'' direction, suppose given a diagram
  as in the solid part of~\eqref{eq:15}. As $t_2$ is generic, there
  exists a map $g \colon A_2 \rightarrow Y$ with $hg = k$ and
  $Fg.t_2 = u$; now both $ga$ and $j$ are maps $t_1 \rightarrow ub$ in
  $B_1 \downarrow F$ and so by genericity of $t_1$ there is some
  $\sigma \in \aut_{t_1}$ with $j = ga\sigma$. By the assumption on
  $a$, there is some minimal $\tau \in \aut_{t_2}$ with
  $\tau a = a \sigma$; letting $\ell = g \tau$, we have
  $\ell a = g\tau a = ga \sigma = j$ and
  $F\ell.t_2 = Fg.F\tau.t_2 = Fg.t_2 = u$. Now
  $ka\sigma = hga\sigma = hj = ka$, whence $k \tau = k$ by minimality
  of $\tau$, and so $h\ell = hg\tau = k\tau = k$ as required for
  $\ell$ to be a filler.

  For the ``only if'' direction, let
  $(b,a) \colon t_1 \rightarrow t_2$ be arrow-generic and let
  $\sigma \in \aut_{t_1}$; we must find a minimal extension
  $\tau \in \aut_{t_2}$ along $a$. Let $q \colon A_2 \rightarrow Q$ be
  a coequaliser of $a$ and $a\sigma$, and consider the diagram
  \begin{equation*}
    \cd[@R-1em]{
      {B_1} \ar[r]^-{b} \ar[dd]_{t_1} &
      {B_2} \ar[dd]_(0.27){t_2} \ar[r]^-{t_2} &
      {FA_2} \ar[dd]^-{Fq}\\ & {} \\
      {FA_1} \ar[r]_-{Fa} \ar'[ur]^-{F(a\sigma)}[uurr] &
      {FA_2} \ar[r]_-{Fq} \ar@{-->}[uur]_-{F\tau}& {FQ}\rlap{ .}
    }
  \end{equation*}
  The solid part clearly commutes, and so we induce a map $\tau$ as
  displayed making everything commute. Since $t_2$ is generic, $\tau$
  is invertible by Corollary~\ref{cor:2}; since $F\tau . t_2 = t_2$, we
  have $\tau \in \aut_{t_2}$. Moreover, $\tau a = a \sigma$, so $\tau$
  is an extension of $\sigma$; while $q \tau = q$ so that $\tau$ is
  minimal.
\end{proof}

This result allows us to check in a concrete fashion the
$(I,J)$-cellularity of a pointwise analytic
$F \colon \psh \rightarrow \psh[\D]$. For this, it suffices to check
$(I, J_n)$-cellularity for each $n$. This is trivial for $n=0$ since
$J_0 = \emptyset$; so suppose now that we have verified it up to $n$.
To check $(I, J_{n+1})$-cellularity, we must show that, for any
$b \colon S \rightarrow \yoneda_d$ in $J_{n+1} \setminus J_n$, each
square~\eqref{eq:24} with generic sides is arrow-generic with $a$ a
finite relative $I$-complex. We can do this using the previous
result so long as we can compute all such squares~\eqref{eq:24}.

Now, as $F$ is $(I, J_n)$-cellular, it is by Proposition~\ref{prop:15}
analytic at the finite $J_n$-complex $S$. Thus, each
square~\eqref{eq:24} is obtained from a generic
$t_2 \colon \yoneda_d \rightarrow FA_2$---which we can classify by
pointwise analyticity of $F$---upon  forming a generic cover
$a \colon t_1 \rightarrow t_2 b$ in $S \downarrow F$, which can be calculated
explicitly by applying the algorithm of
Proposition~\ref{prop:17} to some presentation of $S$ as a finite
$J_n$-complex.

\begin{Ex}
  \label{ex:1}
  Let $J$ be the bordage of Examples~\ref{ex:2}(ii) on $\psh[\atwo]$,
  let $I$ be any bordage on $\psh$, and let
  $F \colon \psh \rightarrow \psh[\atwo]$ be pointwise analytic.
  Recalling that $J_1 = \{0 \rightarrow \yoneda_0\}$, the condition
  for $F$ to be $(I,J_1)$-cellular is that, for each generic
  $t \colon \yoneda_0 \rightarrow FA$, the induced square
  \begin{equation*}
    \cd{
      {0} \ar[r]^-{!} \ar[d]_{!} &
      {\yoneda_0} \ar[d]^{t} \\
      {F0} \ar[r]^-{Fa} &
      {FA}
    }
  \end{equation*}
  with generic sides is arrow-generic with $a$ a finite relative
  $I$-complex. Arrow-genericity is trivial in this case, and so the
  condition is simply that $A$ is a finite $I$-complex. Now, since
  $J_2 \setminus J_1 = \{\yoneda_f \colon \yoneda_0 \rightarrow
  \yoneda_1\}$, we see that $F$ will be $(I, J_2) = (I,J)$-cellular
  when, for each generic $t_2 \colon \yoneda_0 \rightarrow FB$, the
  induced square
  \begin{equation*}
    \cd{
      {\yoneda_0} \ar[r]^-{\yoneda_f} \ar[d]_{t_1} &
      {\yoneda_1} \ar[d]^{t_2} \\
      {FA} \ar[r]^-{Fa} &
      {FB}
    }
  \end{equation*}
  with generic sides is arrow-generic and has $a$ a finite relative
  $J$-complex. Let us see what this says relative to an explicit
  presentation of $F$ as a pointwise coproduct of near-representables:
  \begin{equation*}
    F  \qquad = \qquad \textstyle\sum_{u \in U} \psh(B_u, \thg)_{/ G_u}
    \xrightarrow{\alpha} \sum_{v \in V} \psh(A_v, \thg) _{/ H_v}\rlap{ .}
  \end{equation*}
  Here, the map $\alpha$ is determined by a function
  $h \colon U \rightarrow V$ together with natural transformations
  $\alpha_u \colon \psh(B_u, \thg) _{/G_u} \rightarrow \psh(A_{hu},
  \thg)_{/H_{hu}}$---which, as in the proof of
  Proposition~\ref{prop:3}, correspond to maps
  $[a_u] \colon (A_{hu}, H_{hu}) \rightarrow (B_u, G_u)$ in $\O(\psh)$
  (note that this is really just an explicit description of the
  exponent of $F$). In these terms, the necessary conditions for $F$ to
  be $(I, J)$-cellular are that: each $A_v$ should be a finite
  $I$-complex; each $a_u \colon A_{hu} \rightarrow B_u$ should be a
  finite relative $I$-complex; and each $\sigma \in H_{hu}$ should have
  a minimal extension $\tau \in G_u$ along $a_u$.
\end{Ex}

\subsection{Universal cellular functors}
\label{sec:univ-cell-funct}

We have now achieved what we set out to do in this section, by
exhibiting a class of pointwise analytic functors which is closed
under composition. Our broader objective, recall, is to construct a
universal shapely monad as a terminal object among a suitable class of
composable endofunctors; and so it remains to check the existence of a
terminal object among cellular functors. Unfortunately, we have:

\begin{Prop}\label{prop:no-terminal-object}
  $\cat{CELL}((\C, I), (\D, J))$ need not admit a terminal object.
\end{Prop}

\begin{proof}
  Let $\C = \D = \atwo$ and let both $I$ and $J$ be the bordage
  $\{0 \rightarrow \yoneda_0, \yoneda_0 \rightarrow \yoneda_1\}$ of
  Examples~\ref{ex:2}(ii). Consider the endofunctor
  $F \colon \psh[\atwo] \rightarrow \psh[\atwo]$ sending
  $f \colon A \rightarrow B$ to
  $f \times f \colon A \times A \rightarrow B \times B$; this is
  pointwise analytic, with explicit presentation
  \begin{equation*}
    F \qquad = \qquad \psh[\atwo](\yoneda_1 + \yoneda_1, \thg)_{/1}
    \xrightarrow{\psh[\atwo](\yoneda_f +\yoneda_f , \thg)} \psh[\atwo](\yoneda_0 + \yoneda_0, \thg)_{/1}\rlap{ .}
  \end{equation*}
  To see that $F$ is $(I, I)$-cellular, we observe that
  $\yoneda_0 + \yoneda_0$ is a finite $I$-complex (= finitely
  presentable presheaf), that $\yoneda_f + \yoneda_f$ is a relative
  finite $I$-complex (= monomorphism with cofinite image), and that the
  minimal extension condition is trivially satisfied: this verifies the
  three conditions of Example~\ref{ex:1}, as required.

  In the terminology of Definition~\ref{def:11}, $F$ has spectrum
  $S_F = 1 \in \psh[\atwo]$---so that $\el S_F \cong \atwo$---and
  exponent $E_F \colon \atwo \rightarrow \O(\psh[\atwo])$ picking out
  the arrow
  \begin{equation*}
    [\yoneda_f + \yoneda_f] \colon (\yoneda_0 + \yoneda_0, 1)
    \rightarrow (\yoneda_1 + \yoneda_1, 1)
  \end{equation*}
  of $\O(\psh[\atwo])$. It follows using Proposition~\ref{prop:4} that
  for each $G \in \anpt(\psh[\atwo], \psh[\atwo])$, pointwise analytic
  transformations $\alpha \colon F \Rightarrow G$ correspond
  bijectively with squares
  \begin{equation}\label{eq:29}
    \cd[@C+1.5em]{
      {\yoneda_0} \ar[r]^-{\yoneda_f} \ar[d]_{t_1} &
      {\yoneda_1} \ar[d]^{t_2} \\
      {G(\yoneda_0 + \yoneda_0)} \ar[r]^-{G(\yoneda_f + \yoneda_f)} &
      {G(\yoneda_1 + \yoneda_1)}
    }
  \end{equation}
  in $\psh[\atwo]$ with generic sides. We claim that whenever $G$ is
  $(I, I)$-cellular, there are always two distinct such squares, so
  that $G$ cannot be terminal in $\cat{CELL}((\atwo, I), (\atwo, I))$.
  Since $G$ was arbitrary, this concludes the proof. Now, if $G$ is
  $(I,I)$-cellular, then any square as displayed above must be
  arrow-generic; by Lemma~\ref{lem:16}, this means that each
  $\sigma \in \aut_{t_1}$ admits a minimal extension
  $\tau \in \aut_{t_2}$ along $\yoneda_f + \yoneda_f$. Using
  Proposition~\ref{prop:8}, we see that the switch isomorphism
  $\sigma_{12} \colon \yoneda_0 + \yoneda_0 \rightarrow \yoneda_0 +
  \yoneda_0$ has no minimal extension along $\yoneda_f + \yoneda_f$,
  and so we must have $\aut_{t_1} = 1$. Since
  $ [\yoneda_f + \yoneda_f] \colon (\yoneda_0 + \yoneda_0, \aut_{t_1})
  \rightarrow (\yoneda_1 + \yoneda_1, \aut_{t_2})$ in
  $\O(\psh[\atwo])$, it follows that $\aut_{t_2} = 1$ too; whence the
  square
  \begin{equation*}
    \cd[@C+1.5em]{
      {\yoneda_0} \ar[r]^-{\yoneda_f} \ar[d]_{G\sigma_{12}.t_1} &
      {\yoneda_1} \ar[d]^{G\sigma_{12}.t_2} \\
      {G(\yoneda_0 + \yoneda_0)} \ar[r]^-{G(\yoneda_f + \yoneda_f)} &
      {G(\yoneda_1 + \yoneda_1)}
    }
  \end{equation*}
  is a second, \emph{distinct} instance of~\eqref{eq:29}. This proves
  the claim.
\end{proof}

\section{Shapeliness in context}
\label{sec:free-shapely-monads}

We have now failed for a third time to exhibit a notion of universal
shapely monad: the imposition of cellularity, which fixed the failure
of general analytic functors to compose, did so at the cost of
destroying the terminal object existing among them. At this point, we
prefer to leave for future work the problem of finding a general
notion of shapeliness, and concentrate instead on giving a solution
for the particular motivating examples from
Section~\ref{sec:motivating-examples}.

For these examples, the notion of cellularity turns out to be almost
sufficient: some simple \emph{ad hoc} additional conditions will be
enough to obtain the desired universal shapely monad $\mathsf U$. With
this in place, we can define a general shapely monad to be any
pointwise analytic submonad of $\mathsf U$, and then provide an
inductive construction of the free shapely monad on a generating set
of operations. Finally, we apply this construction to our motivating
examples, and thereby realise the main goal of this paper by
exhibiting the monads for polycategories, properads and \textsc{prop}s
as free shapely monads on the basic wiring operations.

\subsection{Universal shapely monads on (symmetric) polygraphs}
\label{sec:two-level-categories}

Our first goal is to construct universal shapely monads on the
presheaf categories of polygraphs and symmetric polygraphs from
Section~\ref{sec:motivating-examples} which are suitable for analysing
polycategories, properads and \textsc{prop}s. We will find these
universal monads among the class of cellular analytic endofunctors of
the previous section; but in order to neutralise the counterargument
of Proposition~\ref{prop:no-terminal-object}, we must further restrict
the functors under consideration. We build on the observation that the
monads in our examples act trivially on the set of objects of a
(symmetric) polygraph.

\begin{Defn}\label{def:37}
  An endofunctor $F$ of either $\psh[\PP]$ or $\psh[\PS]$ is called
  \emph{framed} if $FX(\star) \cong A \times X(\star)$ for some fixed
  set $A$.
\end{Defn}

Of course, ``acting trivially'' on objects is only the special case
$A = 1$ of this definition. The reason for allowing the more general
definition is to ensure that any pointwise analytic functor admitting
a pointwise analytic transformation to a framed one is itself framed;
see Remark~\ref{rk:3} below.

Since framed functors are clearly stable under composition, we have
for any bordage $I$ on $\psh[\PP]$ or $\psh[\PS]$ a monoidal category
of framed $(I,I)$-cellular endo\-functors. In both the symmetric and
non-symmetric cases, we are free to choose $I$ in any way which
ensures that the monads for polycategories, properads and
\textsc{prop}s are all in fact $(I,I)$-cellular; we now describe
suitable such choices, and check by hand that corresponding monoidal
categories of framed $(I,I)$-cellular endofunctors have a terminal
object---so giving the desired universal shapely monads. Let us begin
with the case of $\psh[\PP]$. The bordage $I_\PP$ we take has
$(I_\PP)_1 = \{0 \rightarrow \yoneda_\star\}$ and
$I_\PP \setminus (I_\PP)_1$ the set of the following maps for all
$n,m \in \mathbb{N}$:
\begin{equation*}
  \textstyle\spn{\yoneda_{\sigma_1}, \dots, \yoneda_{\sigma_n}} \colon 
  \yoneda_\star + \dots + \yoneda_\star \rightarrow \yoneda_{(n,m)}
  \text{ and } 
  \spn{\yoneda_{\tau_1}, \dots, \yoneda_{\tau_m}} \colon 
  \yoneda_\star + \dots + \yoneda_\star \rightarrow \yoneda_{(n,m)} \text{ .}
\end{equation*}

We now show that there is a universal framed $(I_\PP, I_\PP)$-cellular
endofunctor. We make use of the sets $\L(n,m)$ of $(n,m)$-labelled
finite polygraphs from Definition~\ref{def:43} above. We will call
$X \in \L(n,m)$ \emph{well-labelled} if the maps
\begin{equation*}
  \spn{\ell^X_1, \dots, \ell^X_n} \colon 
  \yoneda_\star + \dots + \yoneda_\star \rightarrow \abs{X} \quad
  \text{and} \quad 
  \spn{r^X_1, \dots, r^X_m} \colon 
  \yoneda_\star + \dots + \yoneda_\star \rightarrow \abs{X}
\end{equation*}
are both relative finite $I_\PP$-complexes; and, extending the
notation of Proposition~\ref{prop:22}, we write $\aut_X$ for the set
of label-preserving automorphisms of $\abs X$.

\begin{Prop}
  \label{prop:23}
  The monoidal category of framed $(I_\PP, I_\PP)$-cellular
  endofunctors of $\psh[\PP]$ has a terminal object $U_\PP$, which
  thus underlies a monad on $\psh[\PP]$, the \emph{universal shapely
    monad}. The spectrum $S \in \psh[\PP]$ of $U_\PP $ may be taken to
  be:
  \begin{equation*}
    S(\star) = \{u\} \quad \text{and} \quad S(n,m) = \{X \in \L(n,m) : \text{$X$ is well-labelled}\}\rlap{ ,}
  \end{equation*}
  and the exponent $E \colon \el S \rightarrow \O(\psh[\PP])$ to have
  $E(u) = \yoneda_\star$ and
  \begin{equation*}
    E(X) = (\abs X,
    \aut_X)\text,\qquad E(\sigma_i \colon u \rightarrow X) = [\ell^X_i]\text,\qquad 
    E(\tau_j \colon u \rightarrow X) = [r^X_j]\rlap{ .}
  \end{equation*}
\end{Prop}

\begin{proof}
  By Proposition~\ref{prop:4}, we have
  $\anpt(\psh[\PP], \psh[\PP]) \simeq \el_\PP /\!\!/_{\!v}\,\,
  \O(\psh[\PP])$; if we define
  $\left(\el_\PP /\!\!/_{\!v}\,\, \O(\psh[\PP])\right)' \subset
  \el_\PP /\!\!/_{\!v}\,\, \O(\psh[\PP])$ to be the full subcategory
  corresponding under this equivalence to the full subcategory of
  framed $(I_\PP, I_\PP)$-cellular endofunctors, then it suffices to
  show that $(S,E)$ as defined above is terminal in this category.
  
  First, let us call a functor
  $\PP / \star \cong 1 \rightarrow \O(\psh[\PP])$ \emph{acceptable} if
  it picks out the object $(\yoneda_\star, 1)$, and a functor
  $\PP / (n,m) \rightarrow \O(\psh[\PP])$ \emph{acceptable} if it
  takes the form
  \begin{equation}\label{eq:23}
    \cd[@C-2.3em]{
      \sigma_1 \ar[drrr]_-{\sigma_1} & \cdots & \sigma_n
      \ar[dr]|-{\sigma_n} & &
      \tau_1 \ar[dl]|-{\tau_1} & \cdots & \tau_m
      \ar[dlll]^-{\tau_m} \\
      & & & \id_{(n,m)}
    } \ \ \mapsto \quad 
    \cd[@C-2.9em]{
      (\yoneda_\star, 1) \ar[drrr]_-{[\ell_1]} & \sh{r}{0.25em}{\cdots} & \sh{r}{0.5em}{(\yoneda_\star, 1)}
      \ar[dr]|-{[\ell_n]} & &
      \sh{l}{0.5em}{(\yoneda_\star, 1)} \ar[dl]|-{[r_1]} & \sh{l}{0.25em}{\cdots} & (\yoneda_\star, 1)
      \ar[dlll]^-{[r_m]} \\
      & & & (X,G)
    }  \end{equation}
  with $(X, \ell, r)$ a well-labelled polygraph. By an argument
  like Example~\ref{ex:1} above, 
  a pointwise analytic $F \colon \psh[\PP] \rightarrow \psh[\PP]$
  is framed $(I_\PP, I_\PP)$-cellular just when, for each element $t \in S_F(x)$ of its
  spectrum, the composite
  \begin{equation*}
    F_t \colon \PP / x = \el
    \yoneda_x \xrightarrow{\el t} \el S_F \xrightarrow{E_F} \O(\psh[\PP])
  \end{equation*}
  is acceptable. Noting that this $F_t$ is the same as~\eqref{eq:12}
  appearing in the proof of Proposition~\ref{prop:9}, we thus continue
  by emulating the rest of that proof.

  Recall the key Lemma~\ref{lem:14} stating that, for each $x \in \PP$,
  the connected component of any $F \in [\PP/x, \O(\psh[\PP])]_v$
  contains an initial object $\tilde F$. We claim that, if $F$ is
  acceptable, then so too is $\tilde F$. This is trivial when
  $x = \star$, while if $x = (n,m)$, then $\tilde F$ is obtained from
  $F$ as in~\eqref{eq:23} simply by changing its value at $\id_{(n,m)}$
  from $(X,G)$ to $(X, \aut_X)$. So the analogue of Lemma~\ref{lem:14}
  holds for acceptable functors; it follows that we can define a
  terminal object $(S', E')$ for
  $\left(\el_\PP /\!\!/_{\!v}\,\, \O(\psh[\PP])\right)'$ by taking
  \begin{equation}\label{eq:30}
    S'(x) = \{\,F \in [\PP/x, \O(\psh[\PP])]_v : \tilde F = F \text{
      is acceptable}\, \}\rlap{ ,}
  \end{equation}
  with the remaining data defined exactly as in
  Proposition~\ref{prop:9} above. All that remains is to identify this
  $(S', E')$ with the $(S, E)$ in the statement. Once again, this is
  trivial at stage $\star$, while at stage $(n,m)$, any acceptable $F$
  by definition has the form~\eqref{eq:23}; but the further requirement
  that $F = \tilde F$ means that $G = \aut_X$, so that $F$ determines
  and is determined by the well-labelled polygraph $(X, \ell, r)$.
\end{proof}

\begin{Rk}
  \label{rk:3}
  As in Remark~\ref{rk:2}, if we view the terminal framed cellular
  endofunctor of $\psh[\PP]$ as an object
  $U \in \anpt(\psh[\PP], \psh[\PP])$, then any pointwise analytic $F$
  which admits a map to $U$ in this category must itself be framed
  $(I_\PP, I_\PP)$-cellular. So $U$ is subterminal in
  $\anpt(\psh[\PP], \psh[\PP])$, and the slice category
  $\anpt(\psh[\PP], \psh[\PP]) / U$ may be identified with the
  monoidal category of framed $(I_\PP, I_\PP)$-cellular endofunctors.
\end{Rk}

The case of the presheaf category $\psh[\PS]$ of symmetric polygraphs
is very similar: the maps in the bordage $I_{\PS}$ are identical in
form to those of $I_\PP$---though now living on a different
category---and we now obtain:

\begin{Prop}\label{prop:31}
  The monoidal category of framed $(I_{\PS}, I_{\PS})$-cellular
  endofunctors of $\psh[\PS]$ has a terminal object $U_{\PS}$, which
  thus underlies a monad on $\psh[\PS]$, the \emph{universal shapely
    monad}. The spectrum $S \in \psh[\PS]$ of $U_{\PS}$ may be taken
  to be:
  \begin{equation*}
    S(\star) = \{u\} \quad \text{and} \quad S(n,m) = \{X \in \L_s(n,m) : \text{$X$ is well-labelled}\}\rlap{ ,}
  \end{equation*}
  with symmetric actions on $S(n,m)$ given by
  $X \mapsto \psi \cdot X \cdot \varphi$ as in
  Definition~\ref{def:43}(c); the exponent
  $E \colon \el S \rightarrow \O(\psh[\PS])$ now has
  $E(u) = \yoneda_\star$, $E(X) = (\abs X, \aut_X)$ and
  \begin{equation*}
    E(\sigma_i \colon u \rightarrow X) = [\ell^X_i]\text,\quad 
    E(\tau_j \colon u \rightarrow X) = [r^X_i]\text,\quad 
    E(\xi_{\varphi, \psi} \colon \psi \cdot X \cdot \varphi
    \rightarrow X) = [1_{\abs X}]\text{ .}
  \end{equation*}
\end{Prop}

\subsection{Free shapely monads}

Now that we have universal shapely monads on the presheaf categories
of polygraphs and symmetric polygraphs, we are finally in a position
to define more general shapely monads. It will be convenient to abstract
away from the particularities of our examples as follows.

\begin{Defn}
  \label{def:rshapely}
  Let $U \in \anpt (\psh,\psh)$ be subterminal. We write $\uan$ for
  the full subcategory of $\anpt (\psh,\psh)$ on the
  \emph{$U$-analytic} endofunctors: those admitting a map to $U$. We
  call $U$ \emph{nice} if $\uan$ is closed in $\cat{CAT}(\psh,\psh)$
  under the composition monoidal structure, and in this case we write
  $\umd$ for the category of \emph{$U$-analytic monads}: monoids in
  $\uan$.
\end{Defn}

Clearly, the subterminal $U$ in $\anpt(\psh, \psh)$ becomes terminal
in $\uan$; when $U$ is nice, this terminal object has a unique monoid
structure making it into a terminal object $\mathsf{U}$ in $\umd$. The
universal shapely monads of Propositions~\ref{prop:23} and
\ref{prop:31} arise in this way from the nice subterminal objects
$U_\PP$ and $U_{\PS}$ in the categories of pointwise analytic
endofunctors of $\psh[\PP]$ and $\psh[\PS]$; here ``$U_\PP$-analytic''
means ``framed $(I_\PP, I_\PP)$-cellular'' and likewise for $\PS$.

\begin{Defn}
  \label{def:40}
  Let $U \in \anpt(\psh, \psh)$ be subterminal. A $U$-analytic
  endofunctor $F$ is \emph{shapely} if the unique pointwise analytic
  $F \rightarrow U$ is pointwise monic. If $U$ is nice, then a
  $U$-analytic monad is called \emph{shapely} if its underlying
  endofunctor is so. We write $\shp \subset \uan$ and
  $\shpmon \subset \umd$ for the full subcategories on the shapely
  endofunctors and monads.
\end{Defn}

By the \emph{free shapely monad} on a shapely endofunctor $F$, we mean
the value at $F$ of a left adjoint to the forgetful functor
$\shpmon \rightarrow \shp$. To construct free shapely monads we will
first need to analyse more closely the structure of shapely
endofunctors. The following two results are the key to doing so.

\begin{Prop}
  \label{prop:29}
  For any $\A$ and $\C$ (with $\C$ small), the category
  $\anpt(\A, \psh)$ admits a factorisation system (pointwise epi,
  pointwise mono).
\end{Prop}

Here, and subsequently, the term ``factorisation system'' refers to an
orthogonal factorisation system in the sense of
Freyd--Kelly\cite{Freyd1972Categories}.

\begin{proof}
  First we show that pointwise epimorphic and pointwise monomorphic
  transformations are orthogonal in $\anpt(\A, \psh)$: this says that
  any square
  \begin{equation*}
    \cd{
      {F} \ar[r]^-{\alpha} \ar[d]_{\gamma} &
      {G} \ar[d]^{\delta} \ar@{-->}[dl]_-{\varepsilon} \\
      {H} \ar[r]_-{\beta} &
      {K}
    }
  \end{equation*}
  in $\anpt(\A, \psh)$ with $\alpha$ pointwise epimorphic and $\beta$
  pointwise monomorphic admits a unique diagonal filler $\varepsilon$
  as displayed. As pointwise epimorphic and monomorphic transformations
  are orthogonal in $\cat{CAT}(\A, \psh)$, there is certainly a unique
  transformation $\varepsilon \colon G \Rightarrow H$; we must show it
  is pointwise analytic. For each $c \in \C$ we have the factorisation
  \begin{equation*}
    \yoneda_c \downarrow \delta \qquad = \qquad 
    \yoneda_c \downarrow G
    \xrightarrow{\yoneda_c \downarrow \varepsilon} \yoneda_c \downarrow H
    \xrightarrow{\yoneda_c \downarrow \beta} \yoneda_c \downarrow K\rlap{ .}
  \end{equation*}
  Now since $\beta$ is pointwise monomorphic,
  $\yoneda_c \downarrow \beta$ is fully faithful and so reflects
  generic operations; since $\yoneda_c \downarrow \delta$ preserves
  them, we conclude that $\yoneda_c \downarrow \varepsilon$ preserves
  generics, whence $\varepsilon$ is pointwise analytic as required.

  It remains to show that any $\delta \colon G \Rightarrow K$ in
  $\anpt(\psh, \psh)$ has a pointwise (epi, mono) factorisation. Let
  $\delta = \beta \varepsilon \colon G \Rightarrow H \Rightarrow K$ be
  such a factorisation in $\cat{CAT}(\A, \psh)$; we must show that $H$,
  $\beta$ and $\varepsilon$ are pointwise analytic. We argue as before
  to see that each $\yoneda_c \downarrow \varepsilon$ preserves
  generics, but since $\varepsilon$ is pointwise epimorphic,
  $\yoneda_c \downarrow \varepsilon$ is also surjective on objects;
  whence each $t \in \yoneda_c \downarrow H$ has a generic cover
  obtained by lifting along $\yoneda_c \downarrow \varepsilon$, taking
  a generic cover there, and then applying the generic-preserving
  $\yoneda_c \downarrow \varepsilon$. This shows that $H$ and
  $\varepsilon$ are pointwise analytic. Moreover, as genericity is
  stable under isomorphism, each generic operation in
  $\yoneda_c \downarrow H$ is the image of a generic operation in
  $\yoneda_c \downarrow G$. Since $\yoneda_c \downarrow \delta$
  preserves generics, so does $\yoneda_c \downarrow \beta$, and so
  $\beta$ is also pointwise analytic.
\end{proof}

\begin{Prop}
  \label{prop:30}
  The (pointwise epi, pointwise mono) factorisation system on
  $\anpt(\A, \psh)$ corresponds under Proposition~\ref{prop:4} to the
  factorisation system $(\E, \M)$ on $\el_\C /\!\!/_{\!v}\,\, \O(\A)$
  for which $\E$ and $\M$ comprise those maps $(p, \varphi)$ as
  in~\eqref{eq:10} for which $p$ is epimorphic, respectively $p$ is
  monic and $\varphi$ is invertible.
\end{Prop}

\begin{proof}
  We begin by showing that $(\E, \M)$ is a factorisation system on
  $\el_\C /\!\!/_{\!v}\,\, \O(\A)$. First, any map
  $(p, \varphi) \colon (S,E) \rightarrow (T,D)$ therein admits the
  $(\E, \M)$-factorisation
  \begin{equation*}
    (S,E) \xrightarrow{(e,\varphi)} (R, E.\el\ m) \xrightarrow{(m,1)} (T,D)
  \end{equation*}
  where $p = me \colon S \rightarrow R \rightarrow T$ is a (pointwise
  epi, pointwise mono) factorisation in $\psh$. It remains to show that
  any square
  \begin{equation*}
    \cd{
      {(S,E)} \ar[r]^-{(f,\varphi)} \ar[d]_{(h, \theta)} &
      {(T,D)} \ar[d]^{(k, \kappa)} \ar@{-->}[dl]_-{(j, \delta)} \\
      {(V,C)} \ar[r]_-{(g, \gamma)} &
      {(W,B)}
    }
  \end{equation*}
  with $f$ epimorphic, $g$ monomorphic, and $\gamma$ invertible, admits
  a unique diagonal filler $(j, \delta)$ as indicated. By the
  orthogonality of epimorphic and monomorphic maps in $\psh$, there is
  a unique $j \colon T \rightarrow V$ such that $jf = h$ and $gj = k$.
  But as $\gamma$ is invertible, the unique $\delta$ making the lower
  triangle commute is given by $\kappa \circ (\gamma^{-1} \el\,j)$; a
  short calculation shows that the top triangle then also commutes.

  So $(\E, \M)$ is a factorisation system, and to complete the proof,
  it suffices to show that under the equivalence
  $\anpt(\A, \psh) \rightarrow \el_\C /\!\!/_{\!v}\,\, \O(\A)$ of
  Proposition~\ref{prop:4}, pointwise epimorphic maps correspond to
  maps in $\E$; it then follows by orthogonality that pointwise
  monomorphic maps correspond to ones in $\M$. Now,
  $\alpha \colon F \Rightarrow G$ in $\anpt(\A, \psh)$ is pointwise
  epimorphic just when each functor
  \begin{equation}\label{eq:45}
    \yoneda_c \downarrow \alpha \colon \yoneda_c \downarrow F
    \rightarrow \yoneda_c \downarrow G
  \end{equation}
  is surjective on objects. On the other hand, the $c$-component of the
  induced transformation $S_\alpha \colon S_F \Rightarrow S_G$ on
  spectra can be identified with the action on connected components
  of~\eqref{eq:45}, so that $\alpha$ corresponds to a map in $\E$ just
  when each~\eqref{eq:45} is surjective on connected components. This
  is certainly so if it is surjective on objects; it remains to show
  that, conversely, if~\eqref{eq:45} is surjective on connected
  components, then it is surjective on objects. By our assumption, for
  each $t \colon \yoneda_c \rightarrow GA$, we can find a generic
  operation $s \colon \yoneda_c \rightarrow FB$ such that $\alpha_B. s$
  is in the connected component of $t$. Since $\alpha$ is pointwise
  generic, $\alpha_B. s$ is generic in $\yoneda_c \downarrow G$, and so
  admits a map $f \colon \alpha_B.s \rightarrow t$. This says that
  $t = Gf.\alpha_B.s = \alpha_A.Ff.s$ so that $t$ is the image
  under~\eqref{eq:45} of $Ff.s$, as required.
\end{proof}

\begin{Cor}
  \label{cor:8}
  Let $U \in \anpt(\psh, \psh)$ be subterminal. The category $\shp$ of
  shapely $U$-analytic functors is equivalent to the poset of
  subfunctors of $U$'s spectrum $S_U \in \psh$; in particular, $\shp$
  is a complete preorder, whose joins are given by unions of
  subfunctors of $U \in \cat{CAT}(\psh, \psh)$.
\end{Cor}

\begin{proof}
  By the preceding result, the shapely $U$-analytic endofunctors
  correspond under Proposition~\ref{prop:4} to the $\M$-subobjects of
  $(S_U, E_U)$ in $\el_\C /\!\!/_{\!v}\,\, \O(\psh)$. Any such
  subobject has a unique representative of the form
  $(p, 1) \colon (S,E_U.\el p) \rightarrow (S_U, E_U)$ for
  $p \colon S \hookrightarrow S_U$ a subfunctor inclusion. This proves
  the first claim; the stated form of joins in $\shp$ follows by
  transporting across the equivalence
  $\el_\C /\!\!/_{\!v}\,\, \O(\psh) \simeq \anpt(\psh, \psh)$.
\end{proof}

\begin{Ex}
  \label{ex:4}
  Consider the subterminal $U_\PP \in \anpt(\psh[\PP], \psh[\PP])$
  which classifies framed $(I_\PP, I_\PP)$-cellular endofunctors.
  By~Proposition~\ref{prop:23}, the spectrum $S$ of $U_\PP$ has
  $S(\star) = \{u\}$ and $S(n,m)$ the set of well-labelled elements in
  $\L(n,m)$. We will say that a subpresheaf of $S$ is
  \emph{non-degenerate} if it contains $u \in S(\star)$. Clearly, a
  non-degenerate subpresheaf is given by selecting arbitrary subsets
  $\F(n,m)$ of well-labelled elements from each $\L(n,m)$; the
  corresponding shapely endofunctor
  $F \colon \psh[\PP] \rightarrow \psh[\PP]$---which we also call
  non-degenerate---satisfies $FA(\star) = A(\star)$ and
  \begin{equation*}
    FA(n,m) = \textstyle \sum_{X \in \F(n,m)} \psh[\PP](\abs X, A)_{/ \aut_X}\rlap{ .}
  \end{equation*}
  We may express this subsequently by saying that the non-degenerate
  $F$ \emph{contains} the well-labelled polygraphs in each $\F(n,m)$.
  For example, the identity endofunctor of $\psh[\PP]$ contains
  precisely each of the well-labelled polygraphs $\spn{n,m}$ of
  Definition~\ref{def:43}(d).
\end{Ex}

Returning to the general situation, when $U$ is a nice subterminal
object of $\anpt(\psh, \psh)$, the composition monoidal
structure on $\uan$ induces by way of Proposition~\ref{prop:29} 
the following binary operation on $\shp$.

\begin{Defn}
  \label{def:39}
  Let $U \in \anpt(\psh, \psh)$ be nice. For any $F,G \in \shp$, we
  let $F \cdot G \in \shp$ be the pointwise monic image of the unique
  $u \colon FG \rightarrow U$ in $\uan$:
  \begin{equation*}
    \cd[@!@-1.5em]{
      & {F \cdot G} \ar@{<<-}[dl]_-{} \ar@{ >->}[dr]^-{} \\ 
      {FG} \ar[rr]_-{u} & &
      {U}\rlap{ .}
    }
  \end{equation*}
\end{Defn}

The following lemma describes the basic properties of this operation.

\begin{Lemma}
  \label{lem:20}
  Let $U \in \anpt(\psh, \psh)$ be nice. The assignation
  $F, G \mapsto F \cdot G$ defines a monotone map
  $\shp \times \shp \rightarrow \shp$ which satisfies:
  \begin{equation*}
    F \cdot \id \cong F \qquad \id \cdot G \cong G
    \quad\text{and}\quad (F \cdot G) \cdot H
    \leqslant F \cdot (G \cdot H)\rlap{ .}
  \end{equation*}
  Moreover, each $(\thg) \cdot G \colon \shp \rightarrow \shp$
  preserves joins, and if $U$ is finitary, then each $F \cdot (\thg)$
  preserves directed joins.
\end{Lemma}

\begin{proof}
  Monotonicity and the first two displayed equations are obvious. For
  the third, consider the hexagon left below in $\anpt(\psh, \psh)$;
  the indicated arrows are pointwise epimorphic or pointwise
  monomorphic, whence by orthogonality there is a filler as displayed:
  \begin{equation*}
    \cd[@-0.7em@R-0.3em]{
      & FGH \ar@{->>}[dl]_-{} \ar[dr]^-{} \\ 
      (F \cdot G)H \ar@{->>}[d]_-{} & & 
      F(G \cdot H) \ar@{->>}[d]_-{} \\
      (F \cdot G) \cdot H \ar@{>->}[dr]^-{} \ar@{-->}[rr] & & 
      F \cdot (G \cdot H) \ar@{>->}[dl]^-{} \\
      & U
    }\quad 
    \cd[@-0.8em@R-0.3em]{
      & {\Sigma_i F_iG} \ar@{->>}[dl]_-{} \ar@{->>}[dr]^-{} \\ 
      (\bigvee_i F_i)G \ar@{->>}[d]_-{} & & 
      \Sigma_i(F_i \cdot G) \ar@{->>}[d]_-{} \\
      (\bigvee_i F_i) \cdot G \ar@{>->}[dr]^-{} \ar@{-->}[rr]^{\cong} & & 
      \bigvee_i(F_i \cdot G)\rlap{ .} \ar@{>->}[dl]^-{} \\
      & U
    }
  \end{equation*}

  Next we show that $(\thg) \cdot G$ preserves joins of shapely
  functors. Let $\bigvee_i F_i$ be any such join; since it is computed
  as a union of subfunctors of $U$, the induced transformation
  $\Sigma_i F_i \rightarrow \bigvee_i F_i$ in $\cat{CAT}(\psh, \psh)$
  is epimorphic, whence also its precomposition
  $\Sigma_i F_i G \rightarrow (\bigvee_i F_i)G$. Thus in the hexagon
  right above, each edge is pointwise epi or mono as indicated, so that
  by orthogonality we induce an isomorphism
  $(\bigvee_i F_i) \cdot G \cong \bigvee_i (F_i \cdot G)$ as indicated.

  Suppose now that $U$ is finitary; by Remark~\ref{rk:2}, any
  $F \in \shp$ is then also finitary. Now any directed join
  $\bigvee_i G_i$ in $\shp$, being a union of subfunctors of $U$, may
  be computed as the colimit in $\cat{CAT}(\psh, \psh)$ of the filtered
  diagram of subfunctor inclusions. Because any $F \in \shp$ is
  finitary, it will preserve this colimit, so that the induced map
  $\Sigma_i FG_i \rightarrow F(\bigvee_i G_i)$ in
  $\cat{CAT}(\psh, \psh)$ is pointwise epimorphic. The argument of the
  previous paragraph now carries over \emph{mutatis mutandis} to show
  that $F \cdot (\bigvee_i G_i) \cong \bigvee_i (F \cdot G_i)$ as
  required.
\end{proof}

\begin{Prop}
  \label{prop:28}
  Let $U \in \anpt(\psh, \psh)$ be finitary and nice. The
  forgetful $\shpmon \rightarrow \shp$ is a reflective inclusion of
  preorders, whose image comprises those $F \in \shp$ with
  $\id \leqslant F$ and $F \cdot F \leqslant F$. The left adjoint,
  giving the free shapely monad on $F \in \shp$ is defined by:
  \begin{equation*}
    F \qquad \mapsto \qquad \bar F = \textstyle\bigvee_{n \in \mathbb N} (\id \vee F)^{\cdot n}
  \end{equation*}
  where here $F^{\cdot 0} = \id$ and $F^{\cdot n+1} = F \cdot F^{\cdot n}$.
\end{Prop}

\begin{proof}
  The only non-trivial point is the verification that $\bar F$ is
  indeed a reflection of $F$ into $\shpmon$. First, we have
  $\id = (\id \vee F)^{\cdot 0} \leqslant \bar F$ and
  $\bar F \cdot \bar F \leqslant \bar F$, since
  \begin{equation*}
    \bar F \cdot \bar F \cong \bigvee_n (\id \vee F)^{\cdot n} \cdot \bar F \cong \bigvee_{n,m} (\id \vee F)^{\cdot n} \cdot (\id
    \vee F)^{\cdot m} \leqslant \bigvee_{n,m} (\id \vee F)^{\cdot
      (n+m)} \leqslant \bar F
  \end{equation*}
  where the first two equalities use cocontinuity of
  $(\thg) \cdot \bar F$ and directed cocontinuity of each
  $(\id \vee F)^{\cdot n} \cdot (\thg)$ (noting that the join defining
  $\bar F$ is indeed directed) and the third inequality uses repeatedly
  $(F \cdot G) \cdot H \leqslant F \cdot (G \cdot H)$. So
  $\bar F \in \shpmon$; moreover, if $G \in \shpmon$ satisfies
  $F \leqslant G$, then since $\id \leqslant G$ we have
  $(\id \vee F) \leqslant G$; furthermore, if
  $(\id \vee F)^{\cdot n} \leq G$, then
  $$(\id \vee F)^{\cdot (n+1)} = (\id \vee F) \cdot (\id \vee F)^{\cdot n} \leq G \cdot G \leq G$$
  so that by induction on $n$ we have $(\id \vee F)^{\cdot n} \leq G$ for
  all $n$ and so, finally, that $\bar F \leqslant G$. This proves that
  $\bar F$ is a reflection of $F$ into $\shpmon$ as desired.
\end{proof}

\subsection{Polycategories, properads and PROPs}
\label{sec:first-appl-polyc}

We are now ready to apply the preceding theory to our motivating
examples. We concentrate on exhibiting the ``free polycategory''
monads on $\psh[\PP]$ and $\psh[\PS]$ as free shapely monads, but also
indicate how this extends to the cases of properads and
\textsc{prop}s.

We begin in the non-symmetric case $\psh[\PP]$ by describing a
non-degenerate shapely endofunctor $\Sigma_\PP$ which encodes the basic
polycategorical wiring operations; for this, it suffices by
Example~\ref{ex:4} to describe which well-labelled polygraphs
$\Sigma_\PP$ will contain. We make use of the operations on polygraphs
of Definition~\ref{def:43} above. The elements $\id \in \L(1,1)$ and
$\spn{n,m} \in \L(n,m)$ in parts (a) and (d) of this definition are
well-labelled, and that the operations $(\thg) \ifcomp{j}{i} (\thg)$
and $\psi \cdot (\thg) \cdot \varphi$ of parts (b) and (c) preserve
well-labelledness; so it makes sense to give:

\begin{Defn}\label{def:41}
  Let $\Sigma_\PP$ be the non-degenerate shapely $U_\PP$-analytic
  endofunctor of $\psh[\PP]$ which contains the following
  well-labelled polygraphs:
  \begin{enumerate}[(i)]
  \item  $\id \in \L(1,1)$;
  \item $\psi \cdot \spn{n,m} \cdot \varphi$ for each $n,m$ and
    permutations $\varphi \in \aut_n$ and $\psi \in \aut_m$;
  \item  $\spn{p,q} \ifcomp{j}{i} \spn{n,m} \in \L(n+p-1,m+q-1)$ for
    all $n,m,p,q$ and all suitable indices $i,j$.
  \end{enumerate}
\end{Defn}

\begin{Thm}\label{thm:polycats}
  The free shapely monad on the shapely $U_\PP$-analytic endofunctor
  $\Sigma_\PP$ is the ``free polycategory'' monad on $\psh[\PP]$;
  similarly, the monads for properads and \textsc{prop}s are free
  shapely monads on $\psh[\PP]$.
\end{Thm}
\begin{proof}
  Since $\Sigma_\PP$ contains each of the shapes $\spn{n, m}$, we have
  by Example~\ref{ex:4} that $\id \subset \Sigma_\PP$; so by the
  formula of Proposition~\ref{prop:28}, the free shapely monad on
  $\Sigma_\PP$ is given by $\bigvee_n (\Sigma_\PP)^{\cdot n}$. To
  compute this, we first calculate for any non-degenerate shapely
  $U_\PP$-analytic endofunctor $F$ the composite $\Sigma_\PP \cdot F$.
  Since $F$ is non-degenerate, it is by Example~\ref{ex:4} specified
  by families of well-labelled polygraphs $\F(n,m) \subset \L(n,m)$;
  $\Sigma_\PP \cdot F$ is then also non-degenerate, and so it will
  suffice to determine the well-labelled polygraphs which it contains.
  These polygraphs correspond to generic operations of
  $\Sigma_\PP \cdot F$ at stage $\yoneda_{(n,m)}$, and by
  Definition~\ref{def:39} and Proposition~\ref{prop:30}, such
  operations are precisely the images of the generic operations of
  $\Sigma_\PP F$ at stage $\yoneda_{(n,m)}$ under the unique pointwise
  analytic $\Sigma_\PP F \rightarrow U_\PP$; so it will suffice to
  compute these.

  Now, by Proposition~\ref{prop:16}, any generic operation
  $v \colon \yoneda_{(n,m)} \rightarrow \Sigma_\PP FA$ is the
  composite of a $\Sigma_\PP$-generic operation
  $s \colon \yoneda_{(n,m)} \rightarrow \Sigma_\PP B$ and an
  $F$-generic operation $t \colon B \rightarrow FA$. The first
  possibility is that
  \begin{equation*}
    v \qquad = \qquad \yoneda_{(1,1)} \xrightarrow{s} \Sigma_\PP(\yoneda_\star) \xrightarrow{\Sigma_\PP(t)}
    \Sigma_\PP F(\yoneda_\star)
  \end{equation*}
  where $s$ corresponds to $\id \in \L(1,1)$. This $v$ is sent by
  $\Sigma_\PP F \rightarrow U_\PP$ to a well-labelled polygraph
  $X \in \L(1,1)$ with $\abs X = \yoneda_\star$, which clearly forces
  $X = \id$. The next possibility is that
  \begin{equation*}
    v \qquad = \qquad \yoneda_{(n,m)} \xrightarrow{s} \Sigma_\PP(\yoneda_{(n,m)}) \xrightarrow{\Sigma_\PP(t)}
    \Sigma_\PP F(\abs X)
  \end{equation*}
  where $s$ corresponds to
  $\psi \cdot \spn{n,m} \cdot \varphi \in \L(n,m)$, and $t$
  corresponds to some well-labelled $X \in \F(n,m)$. The composite $v$
  is sent by $\Sigma_\PP F \rightarrow U_\PP$ to a well-labelled
  $Y \in \L(n,m)$ with underlying polygraph $\abs X$; to calculate the
  leaf labellings $\ell^Y_1, \dots, \ell^Y_n$, we apply~\eqref{eq:12}
  and Remark~\ref{rk:1}, which tell us that they arise by taking
  generic covers as to the left in:
  \begin{equation*}
    \cd[@C+2em]{
      {\yoneda_\star} \ar[r]^-{\yoneda_{\sigma_i}} \ar[d]_{\widetilde{v.\yoneda_{\sigma_i}}} &
      {\yoneda_{(n,m)}} \ar[d]^{v} \\
      {\Sigma_\PP F(\yoneda_\star)} \ar[r]^-{\Sigma_\PP F(\ell^Y_i)} &
      {\Sigma_\PP F(\abs X)}
    } \qquad 
    \cd[@C+2em]{
      {\yoneda_\star} \ar[r]^-{\yoneda_{\sigma_i}} \ar[d]_{\widetilde{s.\yoneda_{\sigma_i}}} &
      {\yoneda_\star} \ar[d]^{s} \\
      {\Sigma_\PP(\yoneda_{\star})}
      \ar[r]^-{\Sigma_\PP(\yoneda_{\sigma_{\varphi(i)}})} 
      \ar[d]_{\Sigma_\PP(\widetilde{t.\yoneda_{\sigma_{\varphi(i)}}})} &
      {\Sigma_\PP(\yoneda_{(n,m)})} \ar[d]^-{\Sigma_\PP(t)}\\
      {\Sigma_\PP F(\yoneda_\star)} \ar[r]^-{\Sigma_\PP F(\ell^X_{\varphi(i)})} &
      {\Sigma_\PP F(\abs X)}\rlap{ .} &
    }
  \end{equation*}
  But from the given forms of $s$ and $t$, we have generic covers as
  to the right, and so must have that
  $\ell^Y_i = \ell^X_{\varphi(i)}$. The same argument shows that
  $r^Y_j = r^X_{\psi^{-1}(j)}$, and so we conclude that in fact
  $Y = \psi \cdot X \cdot \varphi$. The final possibility is that
  \begin{equation*}
    v = \yoneda_{(n+p-1,m+q-1)} \xrightarrow{s} \Sigma_\PP(\yoneda_{(p,q)} \ifcomp{j}{i} \yoneda_{(n,m)}) \xrightarrow{\Sigma_\PP(t)}
    \Sigma_\PP FA
  \end{equation*}
  where $s$ corresponds to
  $\spn{p,q} \ifcomp{j}{i} \spn{n,m} \in \L(n+p-1,m+q-1)$. As for the
  $F$-generic $t$, note that the map
  $\yoneda_{\sigma_j} \colon \yoneda_\star \rightarrow
  \yoneda_{(p,q)}$ is a relative finite $I_\mathsf{P}$-complex; whence
  by virtue of the pushout~\eqref{eq:36} and
  Proposition~\ref{prop:17}, $t$ must arise from a pointwise pushout
  in $\psh[\PP] \downarrow F$ of the form:
  \begin{equation*}
    \cd[@!@-5.1em@R-1.1em]{ 
      & \yoneda_\star \ar[dl]_{\yoneda_{\sigma_j}} \ar'[d][dd]_(0.35){t_0} \ar[rr]^{\yoneda_{\tau_i}} && 
      \yoneda_{(n,m)} \ar[dl]^-{} \ar[dd]^{t_1} \\ 
      \yoneda_{(p,q)} \ar[dd]_{t_2} \ar[rr]^(0.4){} && 
      \yoneda_{(p,q)} \ifcomp{j}{i} \yoneda_{(n,m)} \ar[dd]^(0.3){t} \\
      & F\yoneda_\star \ar'[r]^(0.6){Fr^X_i}[rr] \ar[dl]_{F\ell^Y_j\!} && 
      F\abs{X}\ar[dl]^{} \\
      F\abs{Y} \ar[rr]_{} & & 
      FA }
  \end{equation*}
   where $t_0$ is the unique $F$-generic operation at stage
  $\yoneda_\star$, and $t_1$ and $t_2$ are $F$-generic operations
  corresponding to well-labelled polygraphs $X \in \F(n,m)$ and
  $Y \in \F(p,q)$. Since the bottom face is a pushout, we conclude
  that the generic $v$ must correspond to a well-labelled polygraph
  $Z \in \L(n+p-1, m+q-1)$ with $\abs Z = \abs{Y \ifcomp{j}{i} X}$;
  now a similar calculation to before shows that the labellings of $Z$
  are such that, in fact, we have $Z = Y \ifcomp{j}{i} X$.

  In sum, we have now shown that, for any non-trivial shapely
  $U_\PP$-analytic endofunctor $F$ containing the well-labelled
  polygraphs $\F(n,m)$, the well-labelled polygraphs contained in the
  shapely composite $\Sigma_\PP \cdot F$ are given by:
  \begin{enumerate}[(i)]
  \item $\id \in \L(1,1)$;
  \item $\psi \cdot X \cdot \varphi \in \L(n,m)$ for all
    $X \in \F(n,m)$, $\varphi \in \aut_n$ and $\psi \in \aut_m$;
  \item $Y \ifcomp{j}{i} X \in \L(n+p-1, m+q-1)$ for all
    $X \in \F(n,m)$, $Y \in \F(p,q)$ and suitable indices $i,j$.
  \end{enumerate}

  Consequently, the well-labelled polygraphs contained in the free
  shapely monad $\bigvee_n (\Sigma_\PP)^{\cdot n}$ are those obtained
  by closing the $\spn{n,m}$'s under the operations (a)--(c) of
  Definition~\ref{def:43}, and by definition, these are precisely the
  finite labelled polycategorical trees $\T(n,m)$. It follows from
  Example~\ref{ex:4} that the free shapely monad $T$ on $\Sigma_\PP$
  is given by $TX(\star) = X(\star)$ and
  \begin{equation*}
    TX(n,m) = \textstyle \sum_{T \in \T(n,m)} \psh[\PP](\abs T, X)_{/
      \aut_T} = \textstyle \sum_{T \in \T(n,m)} \psh[\PP](\abs T, X)
  \end{equation*}
  where the second equality follows from the observation that a
  labelled polycategorical tree has \emph{no} non-trivial
  label-preserving automorphisms. Comparing this with the formula of
  Proposition~\ref{prop:multicat} gives the result.

  Adapting this result to the cases of properads and \textsc{prop}s is
  almost trivial in our framework; all we need to do is to replace the
  closure operations of Definition~\ref{def:43} which defined the
  class of polycategorical trees with the corresponding closure
  operations defining the properadic graphs or the graphs for
  \textsc{prop}s. Thus, for example, the monad for properads on the
  category $\psh[\PP]$ arises as the free shapely monad on the shapely
  endofunctor specified by the well-labelled polygraphs
  \begin{equation*}
    \id\text, \qquad \psi\cdot\spn{n,m}\cdot\varphi \quad \text{and} \quad \spn{p,q}
    \ifcomp{I}{J} \spn{n,m}\rlap{ .}
  \end{equation*}
  The case of \textsc{prop}s proceeds similarly.
\end{proof}

The argument just given for the free polycategory monad on $\psh[\PP]$
applies equally well to the free polycategory monad on $\psh[\PS]$. By
adapting Example~\ref{ex:4}, we see that the non-trivial shapely
endofunctors of $\psh[\PS]$ are specified by giving subsets
$\F_s(n,m) \subset \L_s(n,m)$ of well-labelled finite symmetric trees;
so we can define a shapely endofunctor $\Sigma_{\PS}$ by requiring it
to contain $\id \in \L_s(1,1)$ and each
$\spn{p,q} \ifcomp{j}{i} \spn{n,m} \in \L_s(n+p-1, m+q-1)$. Now
following the precise same argument as in Theorem~\ref{thm:polycats}
gives:

\begin{Thm}\label{thm:2}
  The free shapely monad on the shapely $U_{\PS}$-analytic endofunctor
  $\Sigma_{\PS}$ is the ``free polycategory'' monad on $\psh[\PS]$;
  similarly, the monads for properads and \textsc{prop}s are free
  shapely monads on $\psh[\PS]$.
\end{Thm}

\bibliographystyle{acm}
\bibliography{./bib}

\end{document}


%% file: preamble-common.tex
\newcommand{\blurb}[1]{}

\makeatletter
\newcommand\footnoteref[1]{\protected@xdef\@thefnmark{\ref{#1}}\@footnotemark}
\makeatother

\usepackage{tikz,tikz-cats}
\tikzset{todim/.style = {decoration={markings, mark=at position .5 with %
      {\draw (-1pt,-1pt) rectangle (1pt,1pt);}},postaction={decorate}}}
\renewcommand{\diaggrandhauteur}{.6}
\renewcommand{\diaggrandlargeur}{1.3}

\DeclareMathOperator{\ob}{ob}

\DeclareMathOperator{\el}{el}

\DeclareMathOperator{\colim}{colim}

\DeclareMathOperator{\im}{Im}

\newcommand{\Aut}{\mathrm{Aut}}

\renewcommand{\emptyset}{\varnothing}

\newcommand{\bigcat}[1]{\ensuremath{\mathsf{#1}}}

\newcommand{\cat}[1]{\mathbf{#1}}

\newcommand{\dbr}[1]{\llbracket {#1} \rrbracket}
\newcommand{\cd}[2][]{\vcenter{\hbox{\xymatrix#1{#2}}}}

\newcommand{\commentthis}[1]{}

\renewcommand{\P}{\maji{P}}

\def\framed{%
\setbox0=\vbox\bgroup%
\advance\hsize by -2\fboxsep\advance\hsize by -2\fboxrule%
\linewidth=\hsize%
}
\def\endframed{%
\egroup\noindent\framebox[\textwidth]{\box0}\vspace*{1mm}}

\usepackage{color}
\definecolor{dkgreen}{rgb}{0,0.2,0}

\newcommand{\A}{\maji{A}}

\newcommand{\B}{\maji{B}}

\newcommand{\C}{\maji{C}}

\newcommand{\tick}{\daimon}
\newcommand{\tausymb}{\spadesuit}

\newcommand{\D}{\maji{D}}

\newcommand{\E}{\maji{E}}

\newcommand{\F}{\maji{F}}

\newcommand{\G}{\maji{G}}

\newcommand{\I}{\maji{I}}

\newcommand{\J}{\maji{J}}

\newcommand{\K}{\maji{K}}

\renewcommand{\L}{\maji{L}}

\newcommand{\M}{\maji{M}}

\newcommand{\yoneda}{\bigcat{y}}

\newcommand{\PP}{\bigcat{P}}

\newcommand{\T}{\maji{T}} %
\newcommand{\Nat}{\mathbb{N}} 

\newcommand{\abs}{\ell}

\newcommand{\un}{1}

\newcommand{\thg}{{\mathord{\text{--}}}}

\newcommand{\rond}{\circ}

\newcommand{\id}{\mathit{id}}

\newcommand{\iso}{\cong}

\newcommand{\impll}{\multimap}
\newcommand{\tens}{\otimes}

\newcommand{\ens}[1]{\{ #1 \}}

\newcommand{\bureaucratic}[1]{}

\newcommand{\daimon}{\heartsuit}

\newcommand{\state}{\sigma}

\newcommand{\extafun}[1]{\overline{(-)}}

\tikzset{history/.style = {-open triangle 45}}

\tikzset{port/.style={draw,thick,-}}

\newcommand{\icomp}[2]{\mathbin{\leftidx{_{#1}}\circ_{#2}}}
\newcommand{\ifcomp}[2]{\mathbin{\leftidx{_{#1}}\bullet_{#2}}}

\makeatletter
\def\matrixobject@{%
  \edef \next@{={\DirectionfromtheDirection@ }}%
  \expandafter \toks@ \next@ \plainxy@
  \let\xy@@ix@=\xyq@@toksix@
  \xyFN@ \OBJECT@}
\let\xy@entry@@norm=\entry@@norm
\def\entry@@norm@patched{%
  \let\object@=\matrixobject@
  \xy@entry@@norm }
\AtBeginDocument{\let\entry@@norm\entry@@norm@patched}
\makeatother

\newcommand{\twocong}[2][0.5]{\ar@{}[#2] \save ?(#1)*{\cong}\restore}
\newcommand{\twoeq}[2][0.5]{\ar@{}[#2] \save ?(#1)*{=}\restore}
\newcommand{\rtwocell}[3][0.5]{\ar@{}[#2] \ar@{=>}?(#1)+/l 0.2cm/;?(#1)+/r 0.2cm/^{#3}}
\newcommand{\ltwocell}[3][0.5]{\ar@{}[#2] \ar@{=>}?(#1)+/r 0.2cm/;?(#1)+/l 0.2cm/^{#3}}
\newcommand{\ltwocello}[3][0.5]{\ar@{}[#2] \ar@{=>}?(#1)+/r 0.2cm/;?(#1)+/l 0.2cm/_{#3}}
\newcommand{\dtwocell}[3][0.5]{\ar@{}[#2] \ar@{=>}?(#1)+/u  0.2cm/;?(#1)+/d 0.2cm/^{#3}}
\newcommand{\dltwocell}[3][0.5]{\ar@{}[#2] \ar@{=>}?(#1)+/ur  0.2cm/;?(#1)+/dl 0.2cm/^{#3}}
\newcommand{\drtwocell}[3][0.5]{\ar@{}[#2] \ar@{=>}?(#1)+/ul  0.2cm/;?(#1)+/dr 0.2cm/^{#3}}
\newcommand{\dthreecell}[3][0.5]{\ar@{}[#2] \ar@3{->}?(#1)+/u  0.2cm/;?(#1)+/d 0.2cm/^{#3}}
\newcommand{\utwocell}[3][0.5]{\ar@{}[#2] \ar@{=>}?(#1)+/d 0.2cm/;?(#1)+/u 0.2cm/_{#3}}
\newcommand{\dtwocelltarg}[3][0.5]{\ar@{}#2 \ar@{=>}?(#1)+/u  0.2cm/;?(#1)+/d 0.2cm/^{#3}}
\newcommand{\utwocelltarg}[3][0.5]{\ar@{}#2 \ar@{=>}?(#1)+/d  0.2cm/;?(#1)+/u 0.2cm/_{#3}}

\newcommand{\pushoutcorner}[1][dr]{\save*!/#1+1.2pc/#1:(1,-1)@^{|-}\restore}

\newdir{(}{{}*!<0em,-.14em>-\cir<.14em>{l^r}}
\newdir{ (}{{}*!/-5pt/\dir{(}}
\newdir{ >}{{}*!/-5pt/\dir{>}}
\newcommand{\sh}[2]{**{!/#1 -#2/}}

\renewcommand{\abs}[1]{{\left|{#1}\right|}}

\newcommand{\spn}[1]{{\langle{#1}\rangle}}
\newcommand{\defeq}{\mathrel{\mathop:}=}

\newcommand{\xtor}[1]{\cdl[@1]{{} \ar[r]|-{\object@{|}}^{#1} & {}}}

\makeatletter

\def\hookleftarrowfill@{\arrowfill@\leftarrow\relbar{\relbar\joinrel\rhook}}
\def\twoheadleftarrowfill@{\arrowfill@\twoheadleftarrow\relbar\relbar}
\def\leftbararrowfill@{\arrowdoublefill@{\leftarrow\mkern-5mu}\relbar\mapstochar\relbar\relbar}
\def\Leftbararrowfill@{\arrowdoublefill@{\Leftarrow\mkern-2mu}\Relbar\Mapstochar\Relbar\Relbar}
\def\leftringarrowfill@{\arrowdoublefill@{\leftarrow\mkern-3mu}\relbar{\mkern-3mu\circ\mkern-2mu}\relbar\relbar}
\def\lefttriarrowfill@{\arrowfill@{\mathrel\triangleleft\mkern0.5mu\joinrel\relbar}\relbar\relbar}
\def\Lefttriarrowfill@{\arrowfill@{\mathrel\triangleleft\mkern1mu\joinrel\Relbar}\Relbar\Relbar}

\def\hookrightarrowfill@{\arrowfill@{\lhook\joinrel\relbar}\relbar\rightarrow}
\def\twoheadrightarrowfill@{\arrowfill@\relbar\relbar\twoheadrightarrow}
\def\rightbararrowfill@{\arrowdoublefill@{\relbar\mkern-0.5mu}\relbar\mapstochar\relbar\rightarrow}
\def\Rightbararrowfill@{\arrowdoublefill@{\Relbar\mkern-2mu}\Relbar\Mapstochar\Relbar\Rightarrow}
\def\rightringarrowfill@{\arrowdoublefill@\relbar\relbar{\mkern-2mu\circ\mkern-3mu}\relbar{\mkern-3mu\rightarrow}}
\def\righttriarrowfill@{\arrowfill@\relbar\relbar{\relbar\joinrel\mkern0.5mu\mathrel\triangleright}}
\def\Righttriarrowfill@{\arrowfill@\Relbar\Relbar{\Relbar\joinrel\mkern1mu\mathrel\triangleright}}

\def\leftrightarrowfill@{\arrowfill@\leftarrow\relbar\rightarrow}
\def\mapstofill@{\arrowfill@{\mapstochar\relbar}\relbar\rightarrow}

\renewcommand*\xleftarrow[2][]{\ext@arrow 20{20}0\leftarrowfill@{#1}{#2}}
\providecommand*\xLeftarrow[2][]{\ext@arrow 60{22}0{\Leftarrowfill@}{#1}{#2}}
\providecommand*\xhookleftarrow[2][]{\ext@arrow 10{20}0\hookleftarrowfill@{#1}{#2}}
\providecommand*\xtwoheadleftarrow[2][]{\ext@arrow 60{20}0\twoheadleftarrowfill@{#1}{#2}}
\providecommand*\xleftbararrow[2][]{\ext@arrow 10{22}0\leftbararrowfill@{#1}{#2}}
\providecommand*\xLeftbararrow[2][]{\ext@arrow 50{24}0\Leftbararrowfill@{#1}{#2}}
\providecommand*\xleftringarrow[2][]{\ext@arrow 10{26}0\leftringarrowfill@{#1}{#2}}
\providecommand*\xlefttriarrow[2][]{\ext@arrow 80{24}0\lefttriarrowfill@{#1}{#2}}
\providecommand*\xLefttriarrow[2][]{\ext@arrow 80{24}0\Lefttriarrowfill@{#1}{#2}}

\renewcommand*\xrightarrow[2][]{\ext@arrow 01{20}0\rightarrowfill@{#1}{#2}}
\providecommand*\xRightarrow[2][]{\ext@arrow 04{22}0{\Rightarrowfill@}{#1}{#2}}
\providecommand*\xhookrightarrow[2][]{\ext@arrow 00{20}0\hookrightarrowfill@{#1}{#2}}
\providecommand*\xtwoheadrightarrow[2][]{\ext@arrow 03{20}0\twoheadrightarrowfill@{#1}{#2}}
\providecommand*\xrightbararrow[2][]{\ext@arrow 01{22}0\rightbararrowfill@{#1}{#2}}
\providecommand*\xRightbararrow[2][]{\ext@arrow 04{24}0\Rightbararrowfill@{#1}{#2}}
\providecommand*\xrightringarrow[2][]{\ext@arrow 01{26}0\rightringarrowfill@{#1}{#2}}
\providecommand*\xrighttriarrow[2][]{\ext@arrow 07{24}0\righttriarrowfill@{#1}{#2}}
\providecommand*\xRighttriarrow[2][]{\ext@arrow 07{24}0\Righttriarrowfill@{#1}{#2}}

\providecommand*\xmapsto[2][]{\ext@arrow 01{20}0\mapstofill@{#1}{#2}}
\providecommand*\xleftrightarrow[2][]{\ext@arrow 10{22}0\leftrightarrowfill@{#1}{#2}}
\providecommand*\xLeftrightarrow[2][]{\ext@arrow 10{27}0{\Leftrightarrowfill@}{#1}{#2}}

\makeatother

\newcommand{\atwo}{{\mathbf 2}}